\documentclass[reqno]{amsart}
\usepackage{hyperref}
\usepackage{graphicx}
\usepackage[lofdepth,lotdepth]{subfig}
\usepackage{curves}
\usepackage{float}
\usepackage{caption}
\usepackage{dirtytalk}
\usepackage{amsthm}
\usepackage{amsbsy}
\usepackage{extarrows}
\usepackage{mathtools}
\usepackage{cool}
\theoremstyle{plain}
\usepackage{tikz-cd}
\usetikzlibrary{arrows}
\usetikzlibrary{decorations.text}
\usepackage{pgfplots}
\pgfplotsset{compat=1.14}

\newcommand*{\mailto}[1]{\href{mailto:#1}{\nolinkurl{#1}}}

\DeclareMathOperator{\arcosh}{arcosh}
\newcommand\restr[2]{{% we make the whole thing an ordinary symbol
  \left.\kern-\nulldelimiterspace % automatically resize the bar with \right
  #1 % the function
  \vphantom{\big|} % pretend it's a little taller at normal size
  \right|_{#2} % this is the delimiter
  }}
\DeclarePairedDelimiter{\floor}{\lfloor}{\rfloor}
\newtheorem{theorem}{Theorem}[section]
\newtheorem{definition}[theorem]{Definition}
\newtheorem{proposition}[theorem]{Proposition}
\newtheorem{corollary}[theorem]{Corollary}
\newtheorem{lemma}[theorem]{Lemma}
\newtheorem{remark}[theorem]{Remark}
\newtheorem{hypothesis}[theorem]{Hypothesis}

\newcommand{\R}{{\mathbb R}}
\newcommand{\N}{{\mathbb N}}

\newcommand{\C}{{\mathbb C}}

\newcommand{\nn}{\nonumber}
\newcommand{\be}{\begin{equation}}
\newcommand{\ee}{\end{equation}}
\newcommand{\bea}{\begin{eqnarray}}
\newcommand{\eea}{\end{eqnarray}}
\newcommand{\btheo}{\begin{theorem}}
\newcommand{\etheo}{\end{theorem}}

\newcommand{\ol}{\overline}
\newcommand{\ti}{\tilde}
\newcommand{\hb}{\hbar}

\newcommand{\asympt}{\mathcal{O}}
\newcommand{\var}{\mathcal{V}}
\newcommand{\Ai}{\operatorname{Ai}}
\newcommand{\Bi}{\operatorname{Bi}}
\newcommand{\W}{\mathcal{W}}

\newcommand{\esf}{\mathsf{E}}
\newcommand{\msf}{\mathsf{M}}
\newcommand{\nsf}{\mathsf{N}}
\newcommand{\xsf}{\mathsf{X}}

\newcommand{\eps}{\varepsilon}
\newcommand{\vphi}{\varphi}
\newcommand{\s}{\sigma}
\newcommand{\lam}{\lambda}
\newcommand{\g}{\gamma}
\newcommand{\G}{\Gamma}
\newcommand{\om}{\omega}
\newcommand{\Om}{\Omega}
\newcommand{\z}{\zeta}
\newcommand{\m}{\mu}
\newcommand{\al}{\alpha}

\newcommand{\ps}{\psi}
\newcommand{\8}{\theta}
\newcommand{\thv}{\vartheta}
\newcommand{\ro}{\rho}
\newcommand{\fisf}{\mathsf{\Phi}}

\numberwithin{equation}{section}

\newtheorem{assumption}[theorem]{Assumption}

\begin{document}

\author{}

\date{}

\title[WKB Problem for the Dirac Operator with a Multi-Humped Potential]
{Semiclassical WKB Problem for the Non-Self-Adjoint Dirac 
Operator with a Multi-Humped Decaying Potential}

\author[N. Hatzizisis]{Nicholas Hatzizisis $^{\dag}$}
\address{$^{\dag}$ Department of Mathematics and Applied Mathematics, 
University of Crete, Greece}
\email{\mailto{nhatzitz@gmail.com}}
\urladdr{\url{http://www.nikoshatzizisis.wordpress.com/home/}}

\author[S. Kamvissis]{Spyridon Kamvissis $^{\ddag}$}
\address{$^{\ddag}$ Department of Mathematics and Applied Mathematics, 
University of Crete, and Institute
of Applied and Computational Mathematics, FORTH, GR--711 10
Voutes Campus, Greece}
\email{\mailto{spyros@tem.uoc.gr}}
\urladdr{\url{http://www.tem.uoc.gr/~spyros/}}

\thanks{}


\subjclass[2000]{}

\bigskip

\begin{abstract}
In this paper we  continue the study (initiated in \cite{h+k}) of the semiclassical behavior of 
the scattering data of a non-self-adjoint Dirac operator with a real, positive, fairly smooth but 
not necessarily analytic potential decaying at infinity; in this paper we  allow this potential 
to have several local maxima and minima. We provide the  rigorous  semiclassical analysis 
of the Bohr-Sommerfeld condition for the location of the eigenvalues, the norming constants,  
and the reflection coefficient.
\end{abstract}

\maketitle

\thispagestyle{empty} 

\section{Introduction}

Consider the initial value problem (IVP) for the 
\textit{one-dimensional focusing nonlinear Schr\"odinger equation with cubic nonlinearity} 
(focusing NLS) for the complex field $u(x,t)$, i.e.
\be\label{ivp-nls}
\begin{cases}
i\hb\partial_t u+\frac{\hb^2}{2}\partial_x^2 u+|u|^2u=0,
\quad (x,t)\in\R\times\R\\\
u(x,0)=A(x),
\quad x\in\R
\end{cases}
\ee
in which $A$ is a real valued function  and $\hb$ is a fixed (at first) positive number;  it 
is a measure of the ratio of the  effect of dispersion to the effect of non-linearity.

A problem like (\ref{ivp-nls}) has attracted much interest due to the 
wide applicability of the NLS equation. Indeed, the NLS equation has been 
derived in many diverse fields of study, governing a plethora of phenomena. 
Just to name a few, it has applications 
\begin{itemize}
\item
to the propagation of light in nonlinear optical fibers 
(cf. \cite{agrawal2001})
\item
to Bose–Einstein condensates (see \cite{pita+stringa2016})
\item
to Langmuir waves in hot plasma physics (cf. \cite{novo1984})
\item
to superconductivity (NLS arises from the Ginzburg-Landau equation 
as a simplified $(1+1)$-dimensional form, see \cite{chiao1964})
\item
to hydrodynamics, for example in small amplitude gravity waves on the 
surface of deep inviscid (zero-viscosity) water (cf. \cite{bt},\cite{bf}).
\end{itemize}  
One can claim that the focusing cubic NLS equation is one of the basic canonical 
non-linear partial differential equations.

\textit{Zakharov} and \textit{Shabat} in \cite{z+s_1972} have proved 
back in 1972 that (\ref{ivp-nls}) is integrable via the 
\textit{Inverse Scattering Method}.
A crucial step of the method  is the analysis 
of the following eigenvalue (EV) problem

\be\label{zs-scattering}
\mathfrak{D}_{\hbar}[\mathbf{u}]=\lambda\mathbf{u}
\ee
where 
\begin{itemize}
\item
$\mathfrak{D}_{\hbar}$ is the 
\textit{Dirac (or Zakharov-Shabat) operator}
\be
\label{dirac-intro}
\mathfrak{D}_{\hbar}=
\begin{bmatrix}
i\hbar\partial_{x} & -iA \\
-iA & -i\hbar\partial_{x}
\end{bmatrix}
\ee
\item
$\mathbf{u}=[u_{1}\hspace{2pt}u_{2}]^T$
is a function from $\mathbb{R}$ to $\mathbb{C}^2$ and
\item
$\lambda\in\C$ is a ``spectral" parameter. 
\end{itemize}

If the solutions $\mathbf{u}$ are in $L^2(\mathbb{R};\mathbb{C}^2)$ the 
corresponding $\lambda$'s are eigenvalues.  The EVs of this problem are related to  
\textit{coherent structures} (e.g. \textit{solitons} and \textit{breathers}) for the  IVP 
(\ref{ivp-nls}) (see \cite{novo1984}). The real part of such an EV represents the speed 
of the soliton while the  imaginary part is related to  its amplitude.  On the other hand 
the continuous spectrum corresponds to bounded (but not $L^2$) ``generalized" 
eigenfunctions $\mathbf{u}$; in our case it is the real line.

In fact the method of Zakharov-Shabat  solves \ref{ivp-nls} by first studying and 
characterizing appropriate ``scattering data" for the potential $A$,  then following the 
(trivial) evolution of such data with respect to time (when we let the potential of the 
Dirac operator evolve according to the NLS equation) and finally using an inverse 
scattering procedure to recover the actual solution of (\ref{ivp-nls}).  The 
scattering data for the Dirac operator consist of 
\begin{itemize}
\item
eigenvalues
\item
``norming constants" related to the $L^2$-norms of the corresponding 
eigenfunctions $\mathbf{u}$ and
\item
the so-called ``reflection" coefficient defined on the continuous spectrum. 
\end{itemize}

Now let us suppose that $\hb$ is small compared to the magnitudes of  $x,t$ that
we are interested in. We are led to the mathematical question: 
what is the behavior of solutions of the IVP (\ref{ivp-nls}) 
as $\hb\downarrow0$? 
Because of the work of Zakharov and Shabat, 
the first step in the study of this IVP in the 
\textit{semiclassical limit}  $\hb\downarrow0$ has to be the 
\textit{asymptotic spectral analysis} of the \textit{scattering problem} 
(\ref{zs-scattering}) as $\hb\downarrow0$, keeping the function $A$ fixed.
This is our main object here.

The rigorous analysis of this  direct scattering problem was 
initiated in \cite{fujii+kamvi} (in the case of real analytic data) 
and more generally in \cite{h+k} for data which is only required to be somewhat 
smooth. The rigorous analysis of the $inverse$ scattering  problem was 
initiated much earlier in \cite{kmm} by use of an ansatz which was justified 
later in \cite{kr}. The eigenvalue problem (\ref{zs-scattering}) is not and cannot 
be written as an EV problem for a self-adjoint operator.  What we study here is a  
\textit{semiclassical WKB problem} (or \textit{LG problem}) for the corresponding 
\textit{non-self-adjoint} Dirac operator with \textit{potential} $A$. 

This work complements our  previous work \cite{h+k}  where
the potential is considered to be a \textit{positive}, \textit{smooth} 
and \textit{even bell-shaped function},  in which we employed 
\textit{Olver's theory}. Working on the same lines, we now discard the 
evenness assumption and additionally let the potential have multiple 
\say{humps} (instead of just a single assumed in \cite{h+k}). 
We should point out that our methods are necessarily different from the 
ones found in \cite{fujii+kamvi} and \cite{h+w}. Those works use the 
\textit{exact WKB method} which requires analyticity. Our ideas here are 
rather influenced by the paper \cite{yafa2018} of \textit{D. R. Yafaev} 
where an analogous problem is treated for the 
\textit{self-adjoint Schr\"odinger operator}, which in turn relies on 
the work \cite{olver1997} of \textit{F. W. J. Olver}
\footnote{Olver's work draws upon the studies of 
\textit{N. D. Kazarinoff, R. E. Langer} and \textit{R.W. McKelvey}
(see the references in \cite{olver1975}).}. We rely heavily 
on \cite{olver1975} instead. 

Under the hypothesis  that  the EVs of $\mathfrak{D}_{\hbar}$ are purely 
imaginary (at least for small enough values of $\hb$, see Hypothesis 
\ref{main-hypo}), the EV problem (\ref{zs-scattering}) under consideration 
becomes a single linear differential equation of second order
\be
\label{diff-equation-equiv}
\frac{d^2y}{dx^2}=
\Big[
\hbar^{-2}f(x,\mu)+g(x,\mu)
\Big]
y
\ee
where $y$ is related to $\mathbf{u}$ while $\mu\in\R_+$ is a parameter 
that substitutes $\lam$. The functions $f$, $g$ are given by the 
following formulae
\be\nn
f(x,\mu)=\mu^2-A(x)^2
\ee
and
\be\nn
g(x,\mu)=
\frac{3}{4}\bigg[\frac{A'(x)}{A(x)+\mu}\bigg]^2-
\frac{1}{2}\frac{A''(x)}{A(x)+\mu}.
\ee
As the zeros of $f$ play a crucial role in the study  of the solutions of 
(\ref{diff-equation-equiv}), we give the following definition.
\begin{definition}
\label{turn-pt}
Consider a differential equation of the form (\ref{diff-equation-equiv})
in which $\mu>0$ is a parameter and $x\in\Delta\subseteq\R$ an interval. 
The zeros in $\Delta$ (with respect to $x$) of the function $f(x,\mu)$ are 
called the \textbf{turning points} (or \textbf{transition points}) of the above 
differential equation.
\end{definition} 

The presentation of our work in the forthcoming sections will be as follows.
In section \S\ref{passage-barrier} we shall deal with  
approximate solutions of (\ref{diff-equation-equiv}) when $A$ behaves as a 
single hump (or a single lobe facing upwards in Klaus-and-Shaw's terminology 
found in \cite{klaus+shaw2003}) in some open neighborhood in $\R$. 
In this section we generalize results obtained in \cite{h+k}. But
we dispense with the eveness assumption considered in \cite{h+k} 
and account for all possible cases for the open neighborhood in which $A$ 
behaves like a single hump. More precisely, in \S\ref{liouville-transform-barrier}, 
we apply the \textit{Liouville transform} to change equation 
(\ref{diff-equation-equiv}) to a new one of the form
\be
\label{weber-equiv-1}
\frac{d^2X}{d\z^2}=
\big[\hb^{-2}(\z^2-\alpha^2)+\ps(\z,\alpha)\big]X
\ee
for some new variables $\z$, $X$ and a function $\ps$,
along the lines first discussed in \cite{olver1975};  here 
the role of the spectral parameter is assumed by the new variable $\alpha$. In 
\S\ref{error-cont-barrier} we prove a useful lemma concerning the 
continuity of $\psi$ which we use in \S\ref{barrier-approximate-solutions} 
to prove Theorem \ref{main-thm-barrier} about approximate solutions to 
(\ref{weber-equiv-1}) for $\z\geq0$. In this case, the approximate 
solutions are expressed in terms of \textit{Parabolic Cylinder Functions} 
(PCFs). 

Next, in \S\ref{asympt-behave-barrier-sols}, we compute  asymptotics 
for the solutions constructed previously and in 
\S\ref{connection-formulas-barrier} we \say{connect} the approximants 
for $\z\geq0$ to approximants for $\z\leq0$ using the so-called 
\textit{connection coefficients}. Finally, in subsection 
\S\ref{applications-barrier}, we combine the tools assembled in this 
section so far, which results in some theorems concerning 
\textit{action integrals} and \textit{quantization conditions}.

The presentation of the material in section \S\ref{passage-well} 
follows the same manner of that in \S\ref{passage-barrier}. 
The main difference 
now is that $A$ behaves locally as a single basin (or bowl; a single 
lobe facing downwards using Klaus-Shaw terminology). 
If we apply now the Liouville transform to (\ref{diff-equation-equiv}) 
we end up having an equation of the form
\be
\label{weber-equiv-2}
\frac{d^2X}{d\z^2}=
\big[\hb^{-2}(\beta^2-\z^2)+\overline{\ps}(\z,\beta)\big]X
\ee
for the same variables $\z$, $X$ as in (\ref{weber-equiv-1}) and a 
function $\overline{\ps}$ (here the bar does \textit{not} denote complex  
conjugation); the spectral paremeter is played now by $\beta$. 
Again, $\overline{\ps}$ can be proven to be continuous; 
we do this in \S\ref{error-cont-well}. 
In paragraph \S\ref{well-approximate-solutions} (cf. Theorem 
\ref{main-thm-well}) we construct approximants to (\ref{weber-equiv-2}) 
for $\z\geq0$ expressed in terms of 
\textit{modified Parabolic Cylinder Functions} (mPCFs). After finding 
their asymptotic behavior in \S\ref{asympt-behave-well-sols}, we 
\say{bridge} them with the approximate solutions for $\z\leq0$ and obtain 
their relevant \textit{connection formulas}. The final subsection, 
namely \S\ref{applications-well}, is the place where a 
\say{fixing behavior} is observed giving rise to a definition 
of \textit{fixing conditions} (along the lines of Yafaev; see 
Definition 5.7 in \cite{yafa2018}). 

\begin{remark}
Let us denote by $\mathcal{R}_A\subset\R_+$  the range of the potential function 
$A:\R\rightarrow\R_+$ and take an $\mu\in\mathcal{R}_A$.  Assuming that equation 
$A(x)=\mu$ has a finite number of solutions,  these divide the domain $\R$ of 
$A$ to a finite number of intervals where $A(x)>\mu$ and to (finitely many) intervals 
where $A(x)<\mu$ (see Figure \ref{bar-and-well-example}).  We call the former  ``barriers" 
and the latter``wells".  When  an interval giving rise to a barrier (well) is bounded, we say 
that we have a barrier (well) of finite width or simply a finite barrier (well).  Correspondingly, 
when we have unbounded intervals, we are in the presence of infinite barriers (or wells), i.e. 
barriers (wells) of infinite width. 
%Informally speaking, if $\mathcal{C}_A$
%looks locally like a hump (as in Figure \ref{pic-barrier}) then we are 
%in the presence of a potential barrier. Similarly, if $\mathcal{C}_A$
%behaves in a neighborhood like a basin (see Figure \ref{pic-well}) 
%then we have a potential well.   
\end{remark}

\begin{figure}[H]
\centering
\begin{tikzpicture}[scale=1.5]

\begin{axis}[
    legend pos = north west,
    axis lines = none,
    xlabel = {},
    ylabel = {},
    axis line style={draw=none},
    tick style={draw=none}
]

\addplot [
    domain=-12:12, 
    samples=500, 
    color=red,
]
{0.8/(1+(x+5)^2)+1/(1+(x+1)^2)+2/(1+(x-5)^2)};
\addlegendentry{$A(x)$} 

\end{axis}

\draw[scale=0.5,domain=0.6:3.6,dashed,variable=\y,blue]  
plot ({4.15},{\y});

\draw[scale=0.5,domain=0.6:3.6,dashed,variable=\y,blue]  
plot ({4.85},{\y});

\draw[scale=0.5,domain=0.6:3.6,dashed,variable=\y,blue]  
plot ({5.85},{\y});

\draw[scale=0.5,domain=0.6:3.6,dashed,variable=\y,blue]  
plot ({6.85},{\y});

\draw[scale=0.5,domain=0.6:3.6,dashed,variable=\y,blue]  
plot ({8.45},{\y});

\draw[scale=0.5,domain=0.6:3.6,dashed,variable=\y,blue]  
plot ({10},{\y});

\node[circle,inner sep=1.5pt,fill=black] at (2.06,1.8) {};

\node[circle,inner sep=1.5pt,fill=black] at (2.43,1.8) {};

\node[circle,inner sep=1.5pt,fill=black] at (2.93,1.8) {};

\node[circle,inner sep=1.5pt,fill=black] at (3.42,1.8) {};

\node[circle,inner sep=1.5pt,fill=black] at (4.22,1.8) {};

\node[circle,inner sep=1.5pt,fill=black] at (5,1.8) {};

\draw [line width=1mm] (2.06,0.3) -- (2.43,0.3);

\draw [line width=1mm] (2.93,0.3) -- (3.42,0.3);

\draw [line width=1mm] (4.22,0.3) -- (5,0.3);

\draw[dashed] (0.8,1.8) -- (6,1.8);

\draw (0,0.3) -- (7,0.3);

\draw (1.6,0.5) node {well};

\draw (2.25,0) node {barrier};

\draw (2.68,0.5) node {well};

\draw (3.2,0) node {barrier};

\draw (3.8,0.5) node {well};

\draw (4.6,0) node {barrier};

\draw (5.4,0.5) node {well};

\draw (6.5,1.8) node {$\mu$};

\end{tikzpicture}
\caption{Barriers and wells for a potential $A$ at a specific energy level $\mu$.}
\label{bar-and-well-example}
\end{figure}
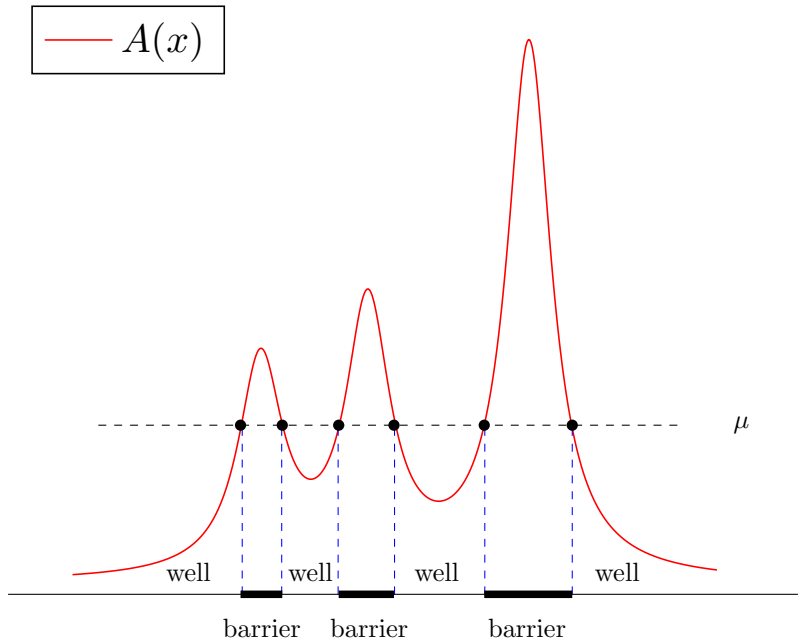

Next, in paragraph \S\ref{multi-hump-quanta}, we study the semiclassical 
spectrum of our operator with multiple potential humps. After 
the introduction of the necessary notation in \S\ref{notation}, we show 
in paragraphs \S\S \ref{problem-statement} - \ref{reformulation} how our 
problem can be transformed to one where Olver's theory (as adapted in 
sections \S\ref{passage-barrier} and \S\ref{passage-well}) can be applied. 
The results about the EVs and their corresponding quantization conditions 
are presented in \S\ref{quanta-spectrum-condo}. We show that 
for each EV there exists 
at least one barrier for which  an associated Bohr-Sommerfeld quantization condition 
can be obtained, essentially in the same way as for the one barrier problem. 
Also,  we establish a one-to-one correspondence between the EVs of the Dirac 
operator $\mathfrak{D}_\hb$ lying in $i\R$ (i.e. imaginary axis) and their 
\textit{WKB approximations}.

The  last component of the semiclassical scattering data is 
the  reflection coefficient. This has been studied semiclassically  
in \cite{h+k}; for completeness we  
present it briefly in paragraph \S\ref{refle} (the reflection coeeficient away from zero 
is presented in \S\ref{refle-away-0} while the behavior closer to zero is found in 
\S\ref{refle-close-0}).  
Since our motivation comes from the
application to semiclassical NLS,  we discuss the effect of our direct scattering
estimates to the inverse scattering problem in section \ref{nls}.  
It turns out that the asymptotic analysis of the inverse problem 
already conducted for the bell-shaped case in \cite{kmm} and \cite{kr} is still relevant.
The main change affects the new density of eigenvalues, which fortunately
still retains its  nice properties that enable the asymptotic analysis of the 
associated Riemann-Hilbert factorization problem. 

For the sake of the reader, 
as the approximate solutions to our problems involve \textit{Airy}, 
Parabolic Cylinder Functions and modified Parabolic Cylinder Functions, 
we present all the necessary results concerning these functions in 
sections  \ref{airy_functions} and \ref{parabolic-cylinder-functions} 
of the appendix.  Finally,  in section \ref{exist-proof} of the appendix we 
present a theorem concerning integral equations that is the backbone 
of the theory that we use in order to arrive at our results.

Before we start our main exposition, we specify  some notation used throughout our work.
\begin{itemize}
\item
Complex conjugation is denoted with a star superscript, \say{$*$}; 
i.e. $z^*$ is the complex conjugate of $z$ (we emphasize 
that a bar over a number, does not indicate its complex conjugate).
\item
The letters $c,C$ denote generically positive constants (unless specified 
otherwise), appearing mainly in estimates .
\item
For the \textit{Wronskian} of two functions $f$, $g$ we use the 
symbol $\W[f,g]$.
\item
The notation $f^2(x)$ denotes the square of the value of the function 
$f$ at $x$; hence, the symbols $f^2(x)$ and $f(x)^2$ are used 
interchangeably and are \textit{not} to be confused with the 
composition $f\circ f$.
\item
The \textit{transpose} of a matrix $M$ is denoted by $M^T$.
\item
For the complement of a set $B$ we write $\complement(B)$.
\item
For a set $\Sigma\subseteq\R$, when we write $i\Sigma$ we mean 
the set $\{i\kappa\mid\kappa\in\Sigma\}\subset\C$. Also, for 
$z_1, z_2\in\C$ the set $[z_1,z_2]\subset\C$ denotes the closed line segment 
starting at $z_1$ and ending at $z_2$.
%\item
%The integer part of $r\in\R$ is denoted by $\lfloor r\rfloor$.
\item
For the restriction of a function $f:\R\rightarrow\R$ to the interval 
$\Delta\subseteq\R$ we have $\restr{f}{\Delta}:\Delta\rightarrow\R$.
\item
The closure of a set $\Sigma\subseteq\R$
 is denoted by $\mathsf{clos}(\Sigma)$.
\item
Take a set $\Sigma\subseteq\R$ and consider a function 
$f:\Sigma\rightarrow\R$. The set 
$\mathcal{R}_f=\{f(x)\mid x\in\Sigma\}\subseteq\R$ 
represents its range.
\item
Unless otherwise specified, $f^{-1}$ shall always denote the inverse of 
an invertible function $f$.
\item
If $T$ is an operator, then $\sigma(T)$, $\sigma_{ess}(T)$ and 
$\sigma_p(T)$ denote its spectrum, essential (continuous) spectrum and 
point spectrum respectively.
\end{itemize}

\section{Passage through a Potential Barrier}
\label{passage-barrier}

We start the investigation for the behavior of solutions of equation
\be
\label{schrodi-barrier}
\frac{d^2y}{dx^2}=
\bigg\{
\hbar^{-2}[\mu^2-A^2(x)]+
\frac{3}{4}\bigg[\frac{A'(x)}{A(x)+\mu}\bigg]^2-
\frac{1}{2}\frac{A''(x)}{A(x)+\mu}
\bigg\}
y
\ee
where $\mu>0$ and the potential function 
$A:\mathbb{R}\rightarrow\mathbb{R}_+$ is
characterized by a barrier of finite width (also called finite barrier) 
in some bounded interval in $\mathbb{R}$. Precisely, we consider 
one of the following assumptions for $A$. 

For the case where the finite barrier lies between two wells of
finite width (called finite wells), we assume the following.
\begin{assumption}
\label{finite-barrier-finite-wells-assumption}
{\bf (finite barrier \& finite wells)}
The function $A$
is positive on some bounded interval 
$[x_1,x_2]$ in $\mathbb{R}$
and has a unique extremum in $(x_1,x_2)$; particularly, a maximum 
$A_{max}>0$ at $b_0\in(x_1,x_2)$. Also $A$ is assumed to be $C^4$ 
in $(x_1,x_2)$ and of class $C^5$ in a neighborhood of $b_0$. Additionally, 
$A'(x)>0$ for $x\in(x_1,b_0)$ and $A'(x)<0$ for $x\in(b_0,x_2)$. At $b_0$ we 
have $A'(b_0)=0$ and $A''(b_0)<0$. Furthermore, if we let 
$A_*=\max\{A(x_1),A(x_2)\}$ and take $\mu\in(A_*,A_{max})\subset\mathbb{R}_+$, 
the equation $A(x)=\mu$ has two solutions $b_-(\mu)$, $b_+(\mu)$ in $(x_1,x_2)$. 
These satisfy $b_-<b_+$, $A(x)>\mu$ for $x\in(b_-,b_+)$ and $A(x)<\mu$ for 
$x\in(x_1,b_-)\cup(b_+,x_2)$ (the above imply $\pm A'(b_\pm)<0$). Finally, when 
$\mu=A_{max}$ the two points $b_-$, $b_+$ coalesce into one double root at $b_0$.
\end{assumption}

\begin{figure}[H]
\centering
\begin{tikzpicture}

\begin{axis}[
    legend pos = north west,
    axis lines = none,
    xlabel = {},
    ylabel = {},
    axis line style={draw=none},
    tick style={draw=none}
]

\addplot [
    domain=0:3.1, 
    samples=100, 
    color=red,
]
{-x^3+3*x^2};
\addlegendentry{$A(x)$}   

\draw[scale=0.5,domain=-4.8:2.8,dashed,variable=\y,blue]  
plot ({1.31},{\y});

\draw[scale=0.5,domain=-4.8:9,dashed,variable=\y,blue]  
plot ({4},{\y});

\draw[scale=0.5,domain=-4.8:2.8,dashed,variable=\y,blue]  
plot ({5.76},{\y});

\draw[scale=0.5,domain=-4.8:1,dashed,variable=\y,blue]  
plot ({0},{\y});

\draw[scale=0.5,domain=-4.8:-1,dashed,variable=\y,blue]  
plot ({6.2},{\y});

\node[circle,inner sep=1.5pt,fill=black] at (0.65,1) {};

\node[circle,inner sep=1.5pt,fill=black] at (2.87,1) {};

\node[circle,inner sep=1.5pt,fill=black] at (2,4) {};

\draw[dashed] (1,4) -- (3.2,4);

\draw (-2,-1.45) -- (4,-1.45);

\draw[dashed] (0.2,1) -- (3.2,1);

\end{axis}

\draw (0.5,-0.2) node {$x_1$};

\draw (6.4,-0.2) node {$x_2$};

\draw (1.8,-0.2) node {$b_-$};

\draw (5.8,-0.2) node {$b_+$};

\draw (4.15,-0.2) node {$b_0$};

\draw (7,5.2) node {$A_{max}$};

\draw (7,2.35) node {$\mu$};

\end{tikzpicture}
\caption{An example of a finite potential barrier surrounded by
two finite wells.}
\label{pic-barrier}
\end{figure}
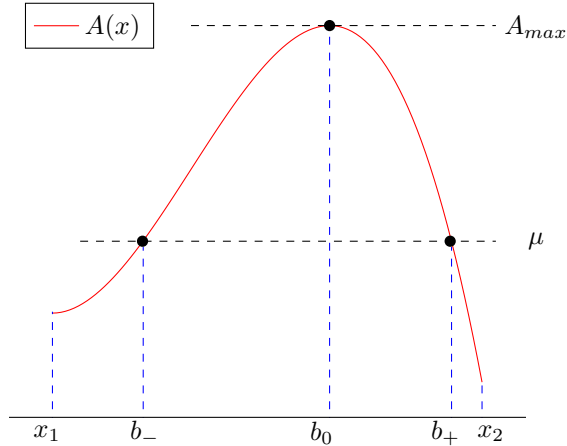

On the other hand, if the finite barrier is surrounded by one or
two infinite wells, we have the following variants of Assumption 
\ref{finite-barrier-finite-wells-assumption}.
In these cases, we need to put some additional decay assumptions 
on $A$, $A'$ and $A''$ at the infinite ends. Hence we have
one of the following.
   
\begin{assumption}
{\bf (finite barrier \& left infinite well)}
\label{finite-barrier-left-infinite-well-assumption}
The function $A$
is positive on $(-\infty,x_2]$ in $\mathbb{R}$ and
$\lim_{x\downarrow-\infty}A(x)=0$.
It has a unique extremum in $(-\infty,x_2)$; particularly, a maximum 
$A_{max}>0$ at $b_0\in(-\infty,x_2)$. Also $A$ is assumed to be $C^4$ 
in $(-\infty,x_2)$ and of class $C^5$ in a neighborhood of $b_0$. 
Additionally, $A'(x)>0$ for $x\in(-\infty,b_0)$ and $A'(x)<0$ for 
$x\in(b_0,x_2)$. At $b_0$ we have $A'(b_0)=0$ and $A''(b_0)<0$. 
Furthermore, if we take $\mu\in(A(x_2),A_{max})\subset\mathbb{R}_+$, 
the equation $A(x)=\mu$ has two solutions $b_-(\mu)$, $b_+(\mu)$ in 
$(-\infty,x_2)$. These satisfy $b_-<b_+$, $A(x)>\mu$ for $x\in(b_-,b_+)$ 
and $A(x)<\mu$ for $x\in(-\infty,b_-)\cup(b_+,x_2)$ (the above imply 
$\pm A'(b_\pm)<0$). When $\mu=A_{max}$ the two points $b_-$, $b_+$ 
coalesce into one double root  at $b_0$. Finally, there exists a number 
$\tau>0$ so that
\begin{center}
$A(x)=\mathcal{O}\Big(\tfrac{1}{|x|^{1+\tau}}\Big)$ as 
$x\downarrow-\infty$\\
$A'(x)=\mathcal{O}\Big(\tfrac{1}{|x|^{2+\tau}}\Big)$ as 
$x\downarrow-\infty$\\
$A''(x)=\mathcal{O}\Big(\tfrac{1}{|x|^{3+\tau}}\Big)$ as 
$x\downarrow-\infty$
\end{center}
\end{assumption}

\begin{figure}[H]
\centering
\begin{tikzpicture}

\begin{axis}[
    legend pos = north west,
    axis lines = none,
    xlabel = {},
    ylabel = {},
    axis line style={draw=none},
    tick style={draw=none}
]

\addplot [
    domain=-3:1.4, 
    samples=100, 
    color=red,
]
{1/(1+x^2)};
\addlegendentry{$A(x)$}   

\draw[scale=0.5,domain=0:0.85,dashed,variable=\y,blue]  
plot ({0.9},{\y});

\draw[scale=0.5,domain=0:1.9,dashed,variable=\y,blue]  
plot ({3},{\y});

\draw[scale=0.5,domain=0:0.85,dashed,variable=\y,blue]  
plot ({5.1},{\y});

%\draw[scale=0.5,domain=-4.8:1,dashed,variable=\y,blue]  
%plot ({0},{\y});

\draw[scale=0.5,domain=0:0.58,dashed,variable=\y,blue]  
plot ({5.8},{\y});

\node[circle,inner sep=1.5pt,fill=black] at (0,1) {};

\node[circle,inner sep=1.5pt,fill=black] at (-1.05,0.48) {};

\node[circle,inner sep=1.5pt,fill=black] at (1.05,0.48) {};

\draw[dashed] (-1,1) -- (1.5,1);

\draw (-4,0.05) -- (2,0.05);

\draw[dashed] (-2,0.48) -- (1.5,0.48);

\end{axis}

%\draw (0.5,-0.2) node {$x_1$};

\draw (6.4,-0.2) node {$x_2$};

\draw (3.1,-0.2) node {$b_-$};

\draw (5.9,-0.2) node {$b_+$};

\draw (4.5,-0.2) node {$b_0$};

\draw (7,5.2) node {$A_{max}$};

\draw (7,2.5) node {$\mu$};

\end{tikzpicture}
\caption{An example of a finite potential barrier accompanied
by an infinite well on the left and a finite well on the right.}
\label{pic-left-barrier}
\end{figure}
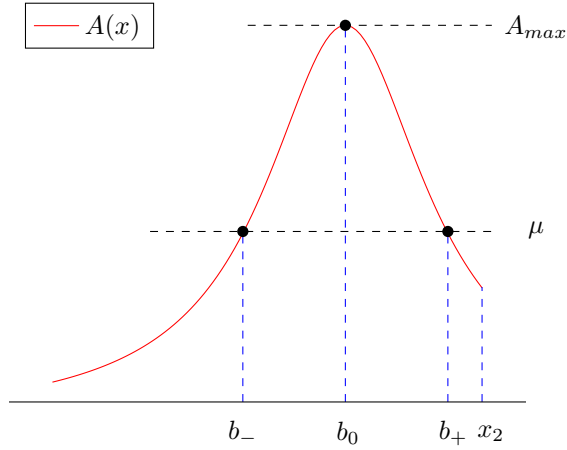

\begin{assumption}
{\bf (finite barrier \& right infinite well)}
\label{finite-barrier-right-infinite-well-assumption}
The function $A$
is positive on some interval 
$[x_1,+\infty)$ in $\mathbb{R}$ and
$\lim_{x\uparrow+\infty}A(x)=0$.
It has a unique extremum in $(x_1,+\infty)$; particularly, a maximum 
$A_{max}>0$ at $b_0\in(x_1,+\infty)$. Also $A$ is assumed to be $C^4$ 
in $(x_1,+\infty)$ and of class $C^5$ in a neighborhood of $b_0$. 
Additionally, $A'(x)>0$ for $x\in(x_1,b_0)$ and $A'(x)<0$ for 
$x\in(b_0,+\infty)$. At $b_0$ we have $A'(b_0)=0$ and $A''(b_0)<0$. 
Furthermore, if we take $\mu\in(A(x_1),A_{max})\subset\mathbb{R}_+$, 
the equation $A(x)=\mu$ has two solutions $b_-(\mu)$, $b_+(\mu)$ in 
$(x_1,+\infty)$. These satisfy $b_-<b_+$, $A(x)>\mu$ for $x\in(b_-,b_+)$ 
and $A(x)<\mu$ for $x\in(x_1,b_-)\cup(b_+,+\infty)$ (the above imply 
$\pm A'(b_\pm)<0$). When $\mu=A_{max}$ the two points $b_-$, $b_+$ 
coalesce into one double root at $b_0$. Finally, there exists a number 
$\tau>0$ so that
\begin{center}
$A(x)=\mathcal{O}\Big(\tfrac{1}{x^{1+\tau}}\Big)$ as 
$x\uparrow+\infty$\\
$A'(x)=\mathcal{O}\Big(\tfrac{1}{x^{2+\tau}}\Big)$ as 
$x\uparrow+\infty$\\
$A''(x)=\mathcal{O}\Big(\tfrac{1}{x^{3+\tau}}\Big)$ as 
$x\uparrow+\infty$
\end{center}
\end{assumption}

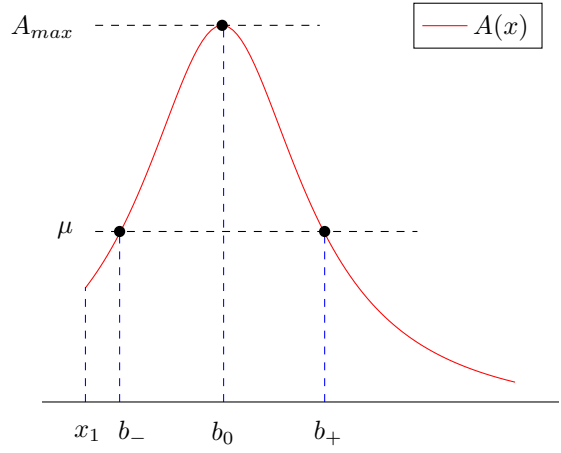
\begin{figure}[H]
\centering
\begin{tikzpicture}

\begin{axis}[
    legend pos = north east,
    axis lines = none,
    xlabel = {},
    ylabel = {},
    axis line style={draw=none},
    tick style={draw=none}
]

\addplot [
    domain=-1.4:3, 
    samples=100, 
    color=red,
]
{1/(1+x^2)};
\addlegendentry{$A(x)$}   

\draw[scale=0.5,domain=0:0.85,dashed,variable=\y,blue]  
plot ({-0.7},{\y});

\draw[scale=0.5,domain=0:1.9,dashed,variable=\y,blue]  
plot ({1.43},{\y});

\draw[scale=0.5,domain=0:0.85,dashed,variable=\y,blue]  
plot ({3.5},{\y});

%\draw[scale=0.5,domain=-4.8:1,dashed,variable=\y,blue]  
%plot ({0},{\y});

\draw[scale=0.5,domain=0:0.58,dashed,variable=\y,blue]  
plot ({-1.4},{\y});

\node[circle,inner sep=1.5pt,fill=black] at (0,1) {};

\node[circle,inner sep=1.5pt,fill=black] at (-1.05,0.48) {};

\node[circle,inner sep=1.5pt,fill=black] at (1.05,0.48) {};

\draw[dashed] (-1.3,1) -- (1,1);

\draw (-4,0.05) -- (4,0.05);

\draw[dashed] (-1.3,0.48) -- (2,0.48);

\end{axis}

\draw (0.6,-0.2) node {$x_1$};

%\draw (6.4,-0.2) node {$x_2$};

\draw (1.2,-0.2) node {$b_-$};

\draw (3.8,-0.2) node {$b_+$};

\draw (2.4,-0.2) node {$b_0$};

\draw (0,5.2) node {$A_{max}$};

\draw (0.3,2.5) node {$\mu$};

\end{tikzpicture}
\caption{An example of a finite potential barrier accompanied
by an infinite well on the right and a finite well on the left.}
\label{pic-right-barrier}
\end{figure}

\begin{assumption}
{\bf (finite barrier \& two infinite wells)}
\label{finite-barrier-left-right-infinite-wells-assumption}
The function $A$ is positive on $\mathbb{R}$ and 
$\lim_{x\to\pm\infty}A(x)=0$. It has a unique extremum in 
$(-\infty,+\infty)$; particularly, a maximum $A_{max}>0$ at 
$b_0\in(-\infty,+\infty)$. Also $A$ is assumed to be $C^4$ 
in $(-\infty,+\infty)$ and of class $C^5$ in a neighborhood of $b_0$. 
Additionally, $A'(x)>0$ for $x\in(-\infty,b_0)$ and $A'(x)<0$ for 
$x\in(b_0,+\infty)$. At $b_0$ we have $A'(b_0)=0$ and $A''(b_0)<0$. 
Furthermore, if we take $\mu\in(0,A_{max})\subset\mathbb{R}_+$, 
the equation $A(x)=\mu$ has two solutions $b_-(\mu)$, $b_+(\mu)$ in 
$(-\infty,+\infty)$. These satisfy $b_-<b_+$, $A(x)>\mu$ for 
$x\in(b_-,b_+)$ and $A(x)<\mu$ for $x\in(-\infty,b_-)\cup(b_+,+\infty)$ 
(the above imply $\pm A'(b_\pm)<0$). When $\mu=A_{max}$ the two points 
$b_-$, $b_+$ coalesce into one double root at $b_0$. Finally, there exists a 
number $\tau>0$ so that
\begin{center}
$A(x)=\mathcal{O}\Big(\tfrac{1}{|x|^{1+\tau}}\Big)$ as 
$x\to\pm\infty$\\
$A'(x)=\mathcal{O}\Big(\tfrac{1}{|x|^{2+\tau}}\Big)$ as 
$x\to\pm\infty$\\
$A''(x)=\mathcal{O}\Big(\tfrac{1}{|x|^{3+\tau}}\Big)$ as 
$x\to\pm\infty$
\end{center}
\end{assumption}

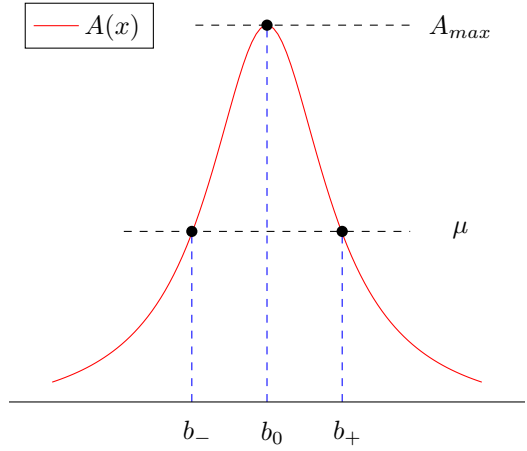
\begin{figure}[H]
\centering
\begin{tikzpicture}

\begin{axis}[
    legend pos = north west,
    axis lines = none,
    xlabel = {},
    ylabel = {},
    axis line style={draw=none},
    tick style={draw=none}
]

\addplot [
    domain=-3:3, 
    samples=100, 
    color=red,
]
{1/(1+x^2)};
\addlegendentry{$A(x)$}   

\draw[scale=0.5,domain=0:0.85,dashed,variable=\y,blue]  
plot ({0.9},{\y});

\draw[scale=0.5,domain=0:1.9,dashed,variable=\y,blue]  
plot ({3},{\y});

\draw[scale=0.5,domain=0:0.85,dashed,variable=\y,blue]  
plot ({5.1},{\y});

%\draw[scale=0.5,domain=-4.8:1,dashed,variable=\y,blue]  
%plot ({0},{\y});

%\draw[scale=0.5,domain=0:0.58,dashed,variable=\y,blue]  
%plot ({5.8},{\y});

\node[circle,inner sep=1.5pt,fill=black] at (0,1) {};

\node[circle,inner sep=1.5pt,fill=black] at (-1.05,0.48) {};

\node[circle,inner sep=1.5pt,fill=black] at (1.05,0.48) {};

\draw[dashed] (-1,1) -- (2,1);

\draw (-4,0.05) -- (4,0.05);

\draw[dashed] (-2,0.48) -- (2,0.48);

\end{axis}

%\draw (0.5,-0.2) node {$x_1$};

%\draw (6.4,-0.2) node {$x_2$};

\draw (2.5,-0.2) node {$b_-$};

\draw (4.5,-0.2) node {$b_+$};

\draw (3.5,-0.2) node {$b_0$};

\draw (6,5.2) node {$A_{max}$};

\draw (6,2.5) node {$\mu$};

\end{tikzpicture}
\caption{An example of a finite potential barrier accompanied
by two infinite wells.}
\label{pic-double-barrier}
\end{figure}

\begin{remark}
The case where $A$ is an even function, satisfying
Assumption \ref{finite-barrier-left-right-infinite-wells-assumption}
is treated in complete detail in \cite{h+k}.
\end{remark}

\subsection{The Liouville transform for a barrier}
\label{liouville-transform-barrier}

We start by assuming the first one of the assumptions above. 
All of them can be treated in a similar manner. Assume
 \ref{finite-barrier-finite-wells-assumption} 
(see Figure \ref{pic-barrier}), with $\mu=A(b_-)=A(b_+)$. 
We temporarily  drop the subscript and set
\be\nn
b_+\equiv b
\ee
\be\nn
I^-=(x_1,b_0),\quad I^+=(b_0,x_2) 
\ee
and define
\be
\label{g-plus-minus}
G^\pm=\Big(\restr{A}{\mathsf{clos}(I^\pm)}\Big)^{-1}.
\ee
Take an arbitrary $\mu_1\in(A_*,A_{max})$
and consider the $b_1\in(b_0,G^+(A_*))$ so that $A(b_1)=\mu_1$ 
(cf. Assumption \ref{finite-barrier-finite-wells-assumption} and 
(\ref{g-plus-minus})); then $\mu\in[\mu_1,A_{max}]$ implies 
$b\in[b_0,b_1]$. For every $\hb>0$, equation (\ref{schrodi-barrier}) 
reads
\be
\label{final-schrodi-barrier}
\frac{d^2y}{dx^2}=[\hbar^{-2}f(x,b)+g(x,b)]y,
\quad (x,b)\in(x_1,x_2)\times[b_0,b_1]
\ee
in which the functions $f$ and $g$ satisfy
\be
\label{f-schrodi-barrier}
f(x,b)=A^2(b)-A^2(x)
\ee
and
\be
\label{g-schrodi-barrier}
g(x,b)=
\frac{3}{4}\bigg[\frac{A'(x)}{A(x)+A(b)}\bigg]^2-
\frac{1}{2}\frac{A''(x)}{A(x)+A(b)}.
\ee
We see that our equation (\ref{final-schrodi-barrier}) has two 
turning points (cf. Definition \ref{turn-pt}) at $x=b_\pm$ when 
$b\in(b_0,b_1]$ coalescing into one double at $x=b_0$; then $b$ 
becomes $b_0$.

Next, we introduce new variables $X$ and $\z$ according 
to the Liouville transform
\be\nn
X=\dot{x}^{-\frac{1}{2}}y
\ee
where the dot signifies differentiation with respect to $\z$. 
Equation (\ref{final-schrodi-barrier}) becomes
\be
\label{schrodi-barrier-liouville-initial-form}
\frac{d^2X}{d\z^2}=
\Big[\hbar^{-2}\dot{x}^2f(x,b)+\dot{x}^2g(x,b)+
\dot{x}^{\frac{1}{2}}\frac{d^2}{d\z^2}(\dot{x}^{-\frac{1}{2}})\Big]X.
\ee
Let us treat the noncritical case $\mu\in[\mu_1,A_{max})$ first; 
two turning points $b_\pm$ being present. In this case $f(\cdot,b)$ 
is negative in $(b_-,b_+)$ and positive in $(x_1,b_-)\cup(b_+,x_2)$. 
Hence we prescribe
\be\label{barrier-case}
\dot{x}^2f(x,b)=\z^2-\alpha^2
\ee
where $\alpha>0$ is chosen in such a way that $x=b_-$ corresponds 
to $\z=-\alpha$ and $x=b_+$ to $\z=\alpha$ accordingly. 

After integration, (\ref{barrier-case}) yields
\be
\label{barrier-integral-form}
\int_{b_-}^{x}[-f(t,b)]^{\frac{1}{2}}dt=
\int_{-\alpha}^{\z}(\alpha^2-\tau^2)^{\frac{1}{2}}d\tau
\ee
provided that $b_-\leq x\leq b_+$ (notice that by taking these 
integration limits, $b_-$ corresponds to $-\alpha$). 
For the remaining correspondence we require
\be\nn
\int_{b_-}^{b_+}[-f(t,b)]^{\frac{1}{2}}dt=
\int_{-\alpha}^{\alpha}(\alpha^2-\tau^2)^{\frac{1}{2}}d\tau
\ee
and hence
\be\label{alpha}
\alpha^2(\mu)=\frac{2}{\pi}\int_{b_-(\mu)}^{b_+(\mu)}
\sqrt{A^2(t)-\mu^2}dt.
\ee
For every fixed value of $\hb$, relation (\ref{alpha}) defines 
$\alpha$ as a continuous and decreasing function of $\mu$ which 
vanishes as $\mu\uparrow A_{max}$. Set 
\be
\label{alpha-1}
\alpha_1=\alpha(\mu_1)>0.
\ee 
Then $\mu\in[\mu_1,A_{max})$ implies $\alpha\in(0,\alpha_1]$.

Next, from (\ref{barrier-integral-form}) we find
\be
\label{zeta-barrier-center}
\int_{b_-}^{x}[-f(t,b)]^{\frac{1}{2}}dt=
\frac{1}{2}\alpha^2\arccos\Big(-\frac{\z}{\alpha}\Big)+
\frac{1}{2}\z\big(\alpha^2-\z^2\big)^{\frac{1}{2}}
\quad\text{for}\quad
b_-\leq x\leq b_+
\ee
with  the principal value choice for the inverse cosine taking 
values in $[0,\pi]$. For the remaining $x$-intervals, we 
integrate (\ref{barrier-case}) to obtain
\be
\label{zeta-barrier-left}
\int_{x}^{b_-}f(t,b)^{\frac{1}{2}}dt=-
\frac{1}{2}\alpha^2\arcosh\Big(-\frac{\z}{\alpha}\Big)-
\frac{1}{2}\z\big(\z^2-\alpha^2\big)^{\frac{1}{2}}
\quad\text{for}\quad
x_1<x\leq b_-
\ee
and
\be
\label{zeta-barrier-right}
\int_{b_+}^{x}f(t,b)^{\frac{1}{2}}dt=-
\frac{1}{2}\alpha^2\arcosh\Big(\frac{\z}{\alpha}\Big)+
\frac{1}{2}\z\big(\z^2-\alpha^2\big)^{\frac{1}{2}}
\quad\text{for}\quad
b_+\leq x<x_2
\ee
with $\arcosh(x)=\ln\big(x+\sqrt{x^2-1}\big)$ for $x\geq1$.

Equations (\ref{zeta-barrier-center}), (\ref{zeta-barrier-left}) 
and (\ref{zeta-barrier-right}) show that $\z$ is a continuous and 
increasing function of $x$ which shows that there is a one-to-one 
correspondence between these two variables. Thus, if we set 
\be
\label{zj-barrier}
\z_j=\lim_{x\to x_j}\z(x)\quad\text{for}\hspace{4pt}j=1,2
\ee
then $(x_1,x_2)$ is mapped by $\z$ to $(\z_1,\z_2)$. 
Notice that $-\infty<\z_1<0<\z_2<+\infty$ since both
$x_1$, $x_2$ are finite by Assumption
\ref{finite-barrier-finite-wells-assumption}
(if $x_1=-\infty$ then $\z_1=-\infty$ and 
 if $x_2=+\infty$ then  $\z_2=+\infty$).

\begin{remark}
\label{critical-barrier-liouville}
In the critical case in which the two (simple) 
turning points coalesce into one (double) point, we get a limit of 
the above transformation with $b=b_0$. In this case, the analogous 
relations to (\ref{zeta-barrier-center}), (\ref{zeta-barrier-left}), 
(\ref{zeta-barrier-right}) are
\be
\label{zeta-barrier-0-left}
\int_{x}^{b_0}f(t,b_0)^{\frac{1}{2}}dt=\frac{1}{2}\z^2
\quad\text{for}\quad
x_1<x\leq b_0
\ee
\be
\label{zeta-barrier-0-right}
\int_{b_0}^{x}f(t,b_0)^{\frac{1}{2}}dt=\frac{1}{2}\z^2
\quad\text{for}\quad
b_0\leq x<x_2
\ee
and $\al=0$.
\end{remark}

Finally, having in mind Remark \ref{critical-barrier-liouville}, 
we substitute (\ref{barrier-case}) in 
(\ref{schrodi-barrier-liouville-initial-form}) and obtain the 
following proposition.
\begin{proposition}
\label{propo-schrodi-to-liouvi}
For every $\hb>0$ equation
\be\nn
\frac{d^2y}{dx^2}=[\hbar^{-2}f(x,b)+g(x,b)]y,
\quad (x,b)\in(x_1,x_2)\times[b_0,b_1]
\ee
where $f$, $g$ as in (\ref{f-schrodi-barrier}), 
(\ref{g-schrodi-barrier}) respectively, is transformed to 
the equation
\be
\label{schrodi-barrier-liouville-final-form}
\frac{d^2X}{d\z^2}=
\big[\hb^{-2}(\z^2-\alpha^2)+\ps(\z,\alpha)\big]X,
\quad 
(\z,\al)\in(\z_1,\z_2)\times[0,\al_1]
\ee
in which $\z$ is given by the Liouville transform 
(\ref{barrier-case}), $\al$ is given by (\ref{alpha}), 
$\z_j$, $j=1,2$ are given by (\ref{zj-barrier}), $\al_1$ as 
in (\ref{alpha-1}) and the function $\ps(\z,\alpha)$ is given 
by the formula
\be
\label{psi-barrier}
\ps(\z,\alpha)=
\dot{x}^2g(x,b)+
\dot{x}^{\frac{1}{2}}\frac{d^2}{d\z^2}(\dot{x}^{-\frac{1}{2}}).
\ee
\end{proposition}

Since in the following paragraphs we shall be interested in 
approximate solutions of equation 
(\ref{schrodi-barrier-liouville-final-form}), we have the following.
\begin{definition}
The function  $\psi$ found in the differential equation 
(\ref{schrodi-barrier-liouville-final-form}) shall be called the 
\textbf{error term} of this equation.
\end{definition}

For the error term we have the following proposition.
\begin{proposition}
The error term $\psi$ can be written equivalently as
\begin{multline}
\label{psi-barrier-equivalent}
\ps(\z,\alpha)=
\frac{1}{4}\frac{3\z^2 +2\alpha^2}{(\z^2-\alpha^2)^2}
+\frac{1}{16}\frac{\z^2-\alpha^2}{f^3(x,b)}
\Big\{4f(x,b)f''(x,b)-5[f'(x,b)]^2\Big\}\\
+(\z^2-\alpha^2)\frac{g(x,b)}{f(x,b)}
\end{multline}
where prime denotes differentiation with respect to $x$. The 
same formula can be used in the critical case of one double 
turning point simply by setting $b=b_0$ and $\alpha=0$.
\end{proposition}
\begin{proof}
Using (\ref{psi-barrier}), (\ref{g-schrodi-barrier}) and 
(\ref{barrier-case}), simple algebraic manipulations shown that 
$\ps$ takes the desired form. 
\end{proof}

\subsection{Continuity of the error term}
\label{error-cont-barrier}

In this subsection we prove a lemma concerning the continuity 
of the function $\psi(\z,\alpha)$ defined in (\ref{psi-barrier}) 
or (\ref{psi-barrier-equivalent}). This fact will be used 
subsequently in \S\ref{barrier-approximate-solutions} to prove 
the existence of approximate solutions of equation 
(\ref{schrodi-barrier-liouville-final-form}). We state it
explicitly.

\begin{lemma}
\label{lemma-error-cont-barrier}
The function $\psi(\z,\alpha)$ defined in (\ref{psi-barrier}),
is continuous in $\z$ and $\alpha$ in the region
$(\z_1,\z_2)\times[0,\alpha_1]$ of the $(\z,\alpha)$-plane.
\end{lemma}
\begin{proof}
For $x\in(x_1,x_2)$, $\mu\in[\mu_1,A_{max}]$ and 
$b\in[b_0,b_1]$ we introduce an auxiliary function $p$ 
by setting
\be
\label{p-barrier}
f(x,b)=(x-b_-)(x-b)p(x,b).
\ee
Having in mind that $A(b_-)=A(b)=\mu$, we see that for 
$\mu\in[\mu_1,A_{max})$
\be\nn
p(b_\pm,b)=\mp\frac{2\mu}{b-b_-} A'(b_\pm)>0 
\ee
while for $\mu=A_{max}$
\be\nn
p(b_0,b_0)=-A_{max}A''(b_0)>0.
\ee 

Our functions $f$, $g$ and $p$ defined by (\ref{f-schrodi-barrier}), 
(\ref{g-schrodi-barrier}) and (\ref{p-barrier}) respectively satisfy 
the following properties
\begin{itemize}
\item[(i)]
$p$, $\frac{\partial p}{\partial x}$, 
$\frac{\partial^2 p}{\partial x^2}$ and $g$ are continuous 
functions of $x$ and $b$ (this means in $x$ and $b$ 
simultaneously and not separately) in the region 
$(x_1,x_2)\times[b_0,b_1]$ 
\item[(ii)]
$p$ is positive throughout the same region
\item[(iii)]
$|\frac{\partial^3 p}{\partial x^3}|$ is bounded in a neighborhood 
of the point $(x,b)=(b_0,b_0)$ in the same region and
\item[(iv)]
$f$ is a non-increasing function of $b\in[b_0,b_1]$ when 
$x\in[b_-,b]$.
\end{itemize}
Indeed, (i) and (iii) follow from (\ref{f-schrodi-barrier}), 
(\ref{g-schrodi-barrier}), (\ref{p-barrier}) and the fact that 
$A$ is in $C^4$ and of class $C^5$ in some neighborhood of $b_0$ 
(see Assumption \ref{finite-barrier-finite-wells-assumption}). 
For (ii), use the definition (\ref{p-barrier}) of $p$ and recall 
the sign of $f$ using (\ref{f-schrodi-barrier}). Finally (iv) is a 
consequence of (\ref{f-schrodi-barrier}) and the monotonicity of 
$A$ in $[b_0,x_2)$ (again cf. Assumption 
\ref{finite-barrier-finite-wells-assumption}). 
By Lemma I in Olver's paper
\cite{olver1975}, the function $\ps$ defined by (\ref{psi-barrier}) 
is continuous in the corresponding region of the $(\z,\alpha)$-plane.
\end{proof}

\subsection{Approximate solutions in the barrier case}
\label{barrier-approximate-solutions}

We return to equation (\ref{schrodi-barrier-liouville-final-form}) 
and state an existence theorem concerning its approximate solutions. 
To this goal, we need a way to assess the error. We do this by 
introducing an \textit{error-control function} $H$ along with a 
\textit{balancing function} $\Om$.
\begin{definition}
Define the \textbf{balancing function} $\Om$ by
\be\label{omega-barrier}
\Om(x)=1+|x|^{\frac{1}{3}}.
\ee 
As an \textbf{error-control function} $H(\z,\al,\hb)$ of equation 
(\ref{schrodi-barrier-liouville-final-form}) we consider any 
primitive of the function
\be\nn
\frac{\ps(\z,\al)}{\Om(\z\sqrt{2\hb^{-1}})}.
\ee
\end{definition}

Furthermore, we need the notion of the \textit{variation} of 
the error-control function $H$ in a given interval. We have
the following.
\begin{definition}
Take $(\gamma,\delta)\subseteq(\z_1,\z_2)\subseteq\R$ 
(cf. (\ref{zj-barrier})). 
The \textbf{variation} $\var_{\gamma,\delta}[H]$ 
in the interval $(\gamma,\delta)$ of the error-control function $H$ 
of equation (\ref{schrodi-barrier-liouville-final-form}) is defined 
by
\be\nn
\var_{\gamma,\delta}[H](\alpha,\hb)=
\int_{\gamma}^{\delta}\frac{|\ps(t,\al)|}{\Om(t\sqrt{2\hb^{-1}})}dt.
\ee
\end{definition}
 
Finally, for any $c\leq0$ set
\be\label{lambda-1-function}
l_1(c)=\sup_{x\in(0,+\infty)}
\bigg\{\Om(x)\frac{\msf(x,c)^2}{\G(\tfrac{1}{2}-c)}\bigg\}
\ee
where $\msf$ is a function defined in terms of Parabolic Cylinder 
Functions in section \ref{pcfs} of the 
appendix and $\Gamma$ denotes the \textit{Gamma function}.
We note that the above supremum is finite for each value of $c$. 
This fact is a consequence of (\ref{omega-barrier}) and the 
first relation in (\ref{M,N-asymptotics}). Furthermore, because 
the relations (\ref{M,N-asymptotics}) hold uniformly in compact 
intervals of $(-\infty,0]$, the function $l_1$ is continuous.

We are now ready for the  main theorem of this paragraph.
\btheo\label{main-thm-barrier}
For each value of $\hb>0$ the equation
\be\nn
\frac{d^2X}{d\z^2}=
\big[\hb^{-2}(\z^2-\alpha^2)+\ps(\z,\alpha)\big]X
\ee 
has in the region 
$[0,\z_2)\times[0,\al_1]$ of the $(\z,\al)$-plane, two 
solutions $Y_+$ and $Z_+$ satisfying
\bea
\label{y1-approx}
Y_+(\z,\al,\hb)=U(\z\sqrt{2\hb^{-1}},-\tfrac{1}{2}\hb^{-1}\al^2)+
\eps_1 (\z,\al,\hb)\\
\label{y2-approx}
Z_+(\z,\al,\hb)=\ol U(\z\sqrt{2\hb^{-1}},-\tfrac{1}{2}\hb^{-1}\al^2)+
\eps_2 (\z,\al,\hb)
\eea
where $U$, $\ol U$ are the PCFs defined in appendix 
\ref{pcfs}.
These two solutions $Y_+$, $Z_+$ are continuous and have continuous 
first and second partial $\z$-derivatives. The errors $\eps_1$, 
$\eps_2$ in the relations above satisfy the estimates
\footnote{
The functions $\esf$, $\msf$ and $\nsf$ are related with the PCF theory 
found in appendix \ref{pcfs}.
}
\begin{multline}\label{barrier-junk1}
\frac{|\eps_1 (\z,\al,\hb)|}
{\msf(\z\sqrt{2\hb^{-1}},-\tfrac{1}{2}\hb^{-1}\al^2)},
\frac{\Big|\frac{\partial \eps_1 }{\partial\z}(\z,\al,\hb)\Big|}
{\sqrt{2\hb^{-1}}
\nsf(\z\sqrt{2\hb^{-1}},-\tfrac{1}{2}\hb^{-1}\al^2)}\\
\leq
\frac{1}{\esf(\z\sqrt{2\hb^{-1}},-\tfrac{1}{2}\hb^{-1}\al^2)}
\Big(\exp\big\{\tfrac{1}{2}(\pi\hb)^{\frac{1}{2}}l_1(-\tfrac{1}{2}
\hb^{-1}\al^2)
\mathcal{V}_{\z,\z_2}[H](\alpha,\hb)\big\}-1\Big)
\end{multline}
and
\begin{multline}\label{barrier-junk2}
\frac{|\eps_2 (\z,\al,\hb)|}
{\msf(\z\sqrt{2\hb^{-1}},-\tfrac{1}{2}\hb^{-1}\al^2)},
\frac{\Big|\frac{\partial \eps_2 }{\partial\z}(\z,\al,\hb)\Big|}
{\sqrt{2\hb^{-1}}
\nsf(\z\sqrt{2\hb^{-1}},-\tfrac{1}{2}\hb^{-1}\al^2)}\\
\leq
\esf(\z\sqrt{2\hb^{-1}},-\tfrac{1}{2}\hb^{-1}\al^2)
\Big(\exp\big\{\tfrac{1}{2}(\pi\hb)^{\frac{1}{2}}l_1(-\tfrac{1}{2}
\hb^{-1}\al^2)
\mathcal{V}_{0,\z}[H](\alpha,\hb)\big\}-1\Big).
\end{multline}
\etheo
\begin{proof}
In order to prove this theorem, we rely on Theorem I in 
\cite{olver1975}. There, it is stated that it suffices to prove 
two things. First that the function $\ps$ is continuous in the region 
$[0,\z_2)\times[0,\al_1]$, a fact that has already been 
proven in \S \ref{error-cont-barrier} and second that the integral
\be\label{variation-total}
\var_{0,\z_2}[H](\alpha,\hb)=
\int_{0}^{\z_2}\frac{|\ps(t,\al)|}{\Om(t\sqrt{2\hb^{-1}})}dt
\ee 
converges  uniformly in $\al$. But this is obvious 
since $\z_2<+\infty$. 
\end{proof}

\begin{remark}
If we were assuming either Assumption
\ref{finite-barrier-right-infinite-well-assumption} or Assumption
\ref{finite-barrier-left-right-infinite-wells-assumption} we 
would have $\z_2=+\infty$. In such a case, Theorem \ref{main-thm-barrier}
would still be true. To obtain it, we have to argue as in the proof of 
Theorem 6.1 in \cite{h+k}.
\end{remark}

\subsection{Asymptotics of the approximate solutions for the barrier}
\label{asympt-behave-barrier-sols}

In order to extract  the asymptotic behavior of the solutions 
$Y_+(\z,\al,\hb)$, $Z_+(\z,\al,\hb)$ when $\hb\downarrow0$, we need 
to determine the asymptotic form of the error bounds
(\ref{barrier-junk1}), (\ref{barrier-junk2}) examining closely
$l_1(-\tfrac{1}{2}\hb^{-1}\al^2)$ and $\var_{0,\z_2}[H](\alpha,\hb)$
as $\hb\downarrow0$. 

Let us deal with the noncritical case $\al\in(0,\al_1]$ first. By 
applying the same analysis found in \S8 of \cite{h+k} we obtain
\be\label{lambda-asympt}
l_1(-\tfrac{1}{2}\hb^{-1}\al^2)=\asympt(1)
\quad\text{as}\quad
\hb\downarrow0.
\ee
Next, we examine $\var_{0,\z_2}[H](\alpha,\hb)$. Again in \S 8 of 
\cite{h+k} it is shown that
\begin{align}\label{v-operator-asympt-barrier}
\var_{0,\z_2}[H](\alpha,\hb)=
\int_{0}^{\z_2}
\frac{|\ps(t,\al)|}{1+(t\sqrt{2\hb^{-1}})^{\frac{1}{3}}}dt=
\asympt(\hb^{1/6})
\quad\text{as}\quad
\hb\downarrow0
\end{align}
when $\z_2=+\infty$. Clearly the same asymptotics hold in
the case when $\z_2<+\infty$ too.

The last two relations applied to (\ref{barrier-junk1}) and 
(\ref{barrier-junk2}) supply us with the desired results 
as $\hb\downarrow0$
\begin{align}
\label{e1-asympt}
\eps_1(\z,\al,\hb) & =
\frac{\msf(\z\sqrt{2\hb^{-1}},-\tfrac{1}{2}\hb^{-1}\al^2)}
{\esf(\z\sqrt{2\hb^{-1}},-\tfrac{1}{2}\hb^{-1}\al^2)}
\asympt(\hb^{\frac{2}{3}})
\\
\nn
\eps_2(\z,\al,\hb) & =
\esf(\z\sqrt{2\hb^{-1}},-\tfrac{1}{2}\hb^{-1}\al^2)
\msf(\z\sqrt{2\hb^{-1}},-\tfrac{1}{2}\hb^{-1}\al^2)
\asympt(\hb^{\frac{2}{3}})
\\
\nn
\frac{\partial \eps_1 }{\partial\z}(\z,\al,\hb) & =
\frac{\nsf(\z\sqrt{2\hb^{-1}},-\tfrac{1}{2}\hb^{-1}\al^2)}
{\esf(\z\sqrt{2\hb^{-1}},-\tfrac{1}{2}\hb^{-1}\al^2)}
\asympt(\hb^{\frac{1}{6}})
\\
\nn
\frac{\partial \eps_2 }{\partial\z}(\z,\al,\hb) & =
\esf(\z\sqrt{2\hb^{-1}},-\tfrac{1}{2}\hb^{-1}\al^2)
\nsf(\z\sqrt{2\hb^{-1}},-\tfrac{1}{2}\hb^{-1}\al^2)
\asympt(\hb^{\frac{1}{6}})
\end{align}
uniformly for $\z\in[0,\z_2)$ and $\al\in(0,\al_1]$.

\begin{remark}
In the special case $\al=0$ (i.e. when equation 
(\ref{schrodi-barrier-liouville-final-form}) has a double 
turning point at $\z=0$), $l_1(0)$ is independent of $\hb$. 
Using the definition (\ref{omega-barrier}) of $\Om$,
we see that we have a similar estimate to 
(\ref{v-operator-asympt-barrier}); namely
$\var_{0,\z_2}[H](0,\hb)=\asympt(\hb^{\frac{1}{6}})$ 
as $\hb\downarrow0$. Hence the results above about the errors, 
hold for the case $\alpha=0$ too.
\end{remark}

\subsection{Connection formulae for a barrier}
\label{connection-formulas-barrier}

We can determine the 
asymptotic behavior of $Y_+$, $Z_+$ for small $\hb>0$ and $\z<0$ by 
establishing appropriate \textit{connection formulae}. We can 
replace $\z$ by $-\z$ in Theorem
\ref{main-thm-barrier} to ensure two more solutions $Y_-$, $Z_-$ of 
equation (\ref{schrodi-barrier-liouville-final-form}) satisfying
as $\hb\downarrow0$
\be
\label{y_minus-z_minus-asympt}
\begin{split}
Y_-(\z,\al,\hb)=
U(-\z\sqrt{2\hb^{-1}},-\tfrac{1}{2}\hb^{-1}\al^2)+
\frac{\msf(-\z\sqrt{2\hb^{-1}},-\tfrac{1}{2}\hb^{-1}\al^2)}
{\esf(-\z\sqrt{2\hb^{-1}},-\tfrac{1}{2}\hb^{-1}\al^2)}
\asympt(\hb^{\frac{2}{3}})\\
Z_-(\z,\al,\hb)=
\ol U(-\z\sqrt{2\hb^{-1}},-\tfrac{1}{2}\hb^{-1}\al^2)+\hspace{4.9cm}\\
\hspace{2cm}\esf(-\z\sqrt{2\hb^{-1}},-\tfrac{1}{2}\hb^{-1}\al^2)
\msf(-\z\sqrt{2\hb^{-1}},-\tfrac{1}{2}\hb^{-1}\al^2)
\asympt(\hb^{\frac{2}{3}})
\end{split}
\ee
uniformly for $\z\in(\z_1,0]$ and $\al\in[0,\al_1]$.

\begin{remark}
\label{linear-independent-barrier}
The two sets $\{Y_+,Z_+\}$ and $\{Y_-,Z_-\}$ consist of two
linearly independent functions. This can be seen by their Wronskians.
For example, using \ref{wronskian-pcf} we have
\be\nn
\W
[U(\cdot,-\tfrac{1}{2}\hb^{-1}\al^2),
\ol U(\cdot,-\tfrac{1}{2}\hb^{-1}\al^2)]=
\sqrt{\frac{2}{\pi}}
\G\Big(\tfrac{1}{2}+\tfrac{1}{2}\hb^{-1}\al^2\Big)
\ee
Using this and (\ref{y1-approx}), (\ref{y2-approx}) we see that
$\W[Y_+,Z_+]\neq0$. Similarly, we have $\W[Y_-,Z_-]\neq0$ as well.
\end{remark}

We express $Y_+$, $Z_+$ in terms of $Y_-$, $Z_-$. So for
$(\z,\alpha)\in(\z_1,0]\times[0,\alpha_1]$ we may write
\begin{align}\label{sigma1}
Y_+(\z,\al,\hb) & =
\s_{11}(\al,\hb)Y_-(\z,\al,\hb)+\s_{12}(\al,\hb)Z_-(\z,\al,\hb)\\
\label{sigma2}
Z_+(\z,\al,\hb) & =
\s_{21}(\al,\hb)Y_-(\z,\al,\hb)+\s_{22}(\al,\hb)Z_-(\z,\al,\hb).
\end{align}
The connection will become clear once  we  find approximations
for the coefficients $\s_{ij}$, $i,j=1,2$ in the linear relations
(\ref{sigma1}) and (\ref{sigma2}).
We evaluate at $\z=0$ equations (\ref{sigma1}), (\ref{sigma2})
and their derivatives. After algebraic manipulations we obtain
\begin{align*}
\s_{11}(\al,\hb)&=
\frac{\W[Y_+(\cdot,\al,\hb),Z_-(\cdot,\al,\hb)](0)}
{\W[Y_-(\cdot,\al,\hb),Z_-(\cdot,\al,\hb)](0)}\\
\s_{12}(\al,\hb)&=
-
\frac
{\W[Y_+(\cdot,\al,\hb),Y_-(\cdot,\al,\hb)](0)}
{\W[Y_-(\cdot,\al,\hb),Z_-(\cdot,\al,\hb)](0)}\\
\s_{21}(\al,\hb)&=
\frac
{\W[Z_+(\cdot,\al,\hb),Z_-(\cdot,\al,\hb)](0)}
{\W[Y_-(\cdot,\al,\hb),Z_-(\cdot,\al,\hb)](0)}\\
\s_{22}(\al,\hb)&=
-
\frac
{\W[Z_+(\cdot,\al,\hb),Y_-(\cdot,\al,\hb)](0)}
{\W[Y_-(\cdot,\al,\hb),Z_-(\cdot,\al,\hb)](0)}.
\end{align*}

Now set 
\be\nn
\vphi(\alpha,\hb)=(1+\hb^{-1}\al^2)\frac{\pi}{4}.
\ee
By using the results and properties of Parabolic Cylinder Functions 
and their auxiliary functions from section \ref{pcf} in the 
appendix, we find that as $\hb\downarrow0$
\begin{align*}
Y_1(0,\al,\hb)&=
\msf(0)[\sin\vphi(\alpha,\hb)+\asympt(\hb^{\frac{2}{3}})]\\
Y_2(0,\al,\hb)&=
\msf(0)[\cos\vphi(\alpha,\hb)+\asympt(\hb^{\frac{2}{3}})]\\
Y_3(0,\al,\hb)&=
\msf(0)[\sin\vphi(\alpha,\hb)+\asympt(\hb^{\frac{2}{3}})]\\
Y_4(0,\al,\hb)&=
\msf(0)[\cos\vphi(\alpha,\hb)+\asympt(\hb^{\frac{2}{3}})]\\
\dot{Y}_1(0,\al,\hb)&=-
\sqrt{2\hb^{-1}}\nsf(0)[\cos\vphi(\alpha,\hb)+
\asympt(\hb^{\frac{2}{3}})]\\
\dot{Y}_2(0,\al,\hb)&=
\sqrt{2\hb^{-1}}\nsf(0)[\sin\vphi(\alpha,\hb)+
\asympt(\hb^{\frac{2}{3}})]\\
\dot{Y}_3(0,\al,\hb)&=
\sqrt{2\hb^{-1}}\nsf(0)[\cos\vphi(\alpha,\hb)+
\asympt(\hb^{\frac{2}{3}})]\\
\dot{Y}_4(0,\al,\hb)&=-
\sqrt{2\hb^{-1}}\nsf(0)[\sin\vphi(\alpha,\hb)+
\asympt(\hb^{\frac{2}{3}})].
\end{align*} 
Recall that the dot denotes differentiation with respect to $\z$.
Finally, using these estimates we obtain
as $\hb\downarrow0$
\be
\label{barrier-connection}
\begin{split}
\sigma_{11}(\al,\hb)=
\sin(\tfrac{1}{2}\pi\hb^{-1}\al^2)+\asympt(\hb^{\frac{2}{3}})\\
\sigma_{12}(\al,\hb)=
\cos(\tfrac{1}{2}\pi\hb^{-1}\al^2)+\asympt(\hb^{\frac{2}{3}})\\
\sigma_{21}(\al,\hb)=
\cos(\tfrac{1}{2}\pi\hb^{-1}\al^2)+\asympt(\hb^{\frac{2}{3}})\\
\sigma_{22}(\al,\hb)=
-\sin(\tfrac{1}{2}\pi\hb^{-1}\al^2)+\asympt(\hb^{\frac{2}{3}})
\end{split}
\ee
uniformly for $\al\in[0,\al_1]$.

\subsection{Applications in the barrier case}
\label{applications-barrier}

We are assumping 
\ref{finite-barrier-finite-wells-assumption} (similar arguments hold
for the other cases as well). Recalling (\ref{alpha}), we define the 
following.
\begin{definition}
If we assume that equation (\ref{final-schrodi-barrier}) describes the 
motion of a system (e.g. a particle), then the function
\be\label{bold-phi}
\fisf(\mu)=\frac{\pi}{2}\alpha^2(\mu)=
\int_{b_{-}(\mu)}^{b_{+}(\mu)}\sqrt{A(x)^2-\mu^2}dx.
\ee
is called the abbreviated \textbf{action} of this motion.
\end{definition}
It is easily checked that $\fisf$ is of class $C^1$. Differentiating
relation (\ref{bold-phi}) while using $A(b_\pm)=\mu$, we obtain
\be
\label{derivative-bold-phi}
\frac{d\fisf(\mu)}{d\mu}=
-2\mu
\int_{b_{-}(\mu)}^{b_{+}(\mu)}[A(x)^2-\mu^2]^{-\frac{1}{2}}dx<0.
\ee

The asymptotic behavior as $\hb\downarrow0$
of an arbitrary non-trivial \textit{real} solution $X$ of equation 
(\ref{schrodi-barrier-liouville-final-form}) on the $\z$-interval
corresponding to the finite $x$-barrier $(b_-,b_+)$ of $A$, can
be examined through the functions $Y_+$, $Z_+$ and $Y_-$, $Z_-$.
Since $\{Y_+,Z_+\}$ and $\{Y_-,Z_-\}$ are two sets of linearly 
independent functions (cf. Remark \ref{linear-independent-barrier}), 
for $X$ we can write
\be
\label{X-linear-combo}
\begin{split}
X(\z,\al,\hb) 
=
\gamma_+(\al,\hb)Y_+(\z,\al,\hb)+\delta_+(\al,\hb)Z_+(\z,\al,\hb)\\
=
\gamma_-(\al,\hb)Y_-(\z,\al,\hb)+\delta_-(\al,\hb)Z_-(\z,\al,\hb)
\end{split}
\ee
for some $\gamma_\pm(\al,\hb),\delta_\pm(\al,\hb)\in\mathbb{R}$.
We put
\be
\label{X-aux}
\begin{split}
v_\pm(\al,\hb)
=
\sqrt{\gamma_\pm(\al,\hb)^2+\delta_\pm(\al,\hb)^2}\\
\gamma_\pm(\al,\hb)
=
v_\pm(\al,\hb)\cos\xi_\pm(\al,\hb)\\
\delta_\pm(\al,\hb)
=
v_\pm(\al,\hb)\sin\xi_\pm(\al,\hb)\\
\xi_\pm(\al,\hb)
\in
\mathbb{R}/(2\pi\mathbb{Z})
\end{split}
\ee
and define
\be\label{xi-sum}
\xi(\mu,\hb)=\xi_+(\al(\mu),\hb)+\xi_-(\al(\mu),\hb).
\ee
Recall that $\alpha$ is function of $\mu$. Whence we can see that 
$\xi_\pm$, $\xi$ depend on $\mu$. Sometimes we shall simply write 
$\xi_\pm(\mu,\hb)$ meaning $\xi_\pm(\al(\mu),\hb)$.

The ideas that follow are essentially the same as those used in the
derivation of the \textit{Bohr-Sommerfeld quantization condition}
found in \S 10 of \cite{h+k}. We start with a theorem.
\begin{theorem}
Under Assumption \ref{finite-barrier-finite-wells-assumption}, 
there is a non-negative integer $n=n(\mu,\hb)$ 
(i.e. that depends on $\mu$ and $\hb$) such that the functions 
$\fisf$ and $\xi$ in (\ref{bold-phi}) and (\ref{xi-sum}) respectively, 
satisfy the formula
\be
\label{B-R-similar}
\fisf(\mu)=\Big[(2n+1)\frac{\pi}{2}-\xi(\mu,\hb)\Big]\hb+
\asympt(\hb^{\frac{5}{3}})
\quad
\text{as}
\quad
\hb\downarrow0.
\ee
\end{theorem}
\begin{proof}
Using (\ref{X-linear-combo}) and (\ref{X-aux}) we have
\begin{align*}
0
&=
\W[X,X]=\W[\gamma_+Z_++\delta_+Y_+,\gamma_-Z_-+\delta_-Y_-]\\
&=
v_+v_-\W[Y_-,Z_-]
(
-\sigma_{12}\cos\xi_+\cos\xi_-
+\sigma_{11}\cos\xi_+\sin\xi_-\\
&
\hspace{4cm}
-\sigma_{22}\sin\xi_+\cos\xi_-
+\sigma_{21}\sin\xi_+\sin\xi_-
)
\end{align*}
(we have suppressed the dependence on $\al$ and $\hb$ 
for notational simplicity). 
From (\ref{X-aux}),  (\ref{barrier-connection}) and (\ref{xi-sum})  
\be\nn
\cos\Big[\hb^{-1}\fisf(\mu)+\xi(\mu,\hb)\Big]=
\asympt(\hb^{\frac{2}{3}})
\quad
\text{as}
\quad
\hb\downarrow0
\ee
from which the result follows.
\end{proof}

If $\xi(\mu,\hb)=0\pmod{\pi}$, then relation (\ref{B-R-similar}) 
reduces to the Bohr-Sommerfeld quantization condition. In particular, 
this is true if $\xi_\pm(\mu,\hb)=0\pmod{\pi}$ at both 
turning points $b_\pm(\mu)$. We state this explicitly.
\begin{theorem}
\label{B-R-existence}
Under Assumption \ref{finite-barrier-finite-wells-assumption}, 
suppose that a non-trivial real solution of 
(\ref{schrodi-barrier-liouville-final-form}) fulfills 
(\ref{X-linear-combo}) and (\ref{X-aux}) with 
$\xi_\pm(\mu,\hb)=0\pmod{\pi}$. Then function $\fisf$ in 
(\ref{bold-phi}) satisfies the condition
\be
\label{B-R-1}
\cos\Big[\hb^{-1}\fisf(\mu)\Big]=
\asympt(\hb^{\frac{2}{3}})
\quad
\text{as}
\quad
\hb\downarrow0
\ee
whence 
\be
\label{B-R-2}
\fisf(\mu)=\pi\Big(n+\tfrac{1}{2}\Big)\hb+
\asympt(\hb^{\frac{5}{3}})
\quad
\text{as}
\quad
\hb\downarrow0
\ee
for some non-negative integer $n=n(\mu,\hb)$.
\end{theorem}

\begin{remark}
It is possible that $\xi_\pm(\mu,\hb)=0\pmod{\pi}$ only 
for $\hb$ in some set $\Sigma\subset\R_+$ such that 
$0\in\mathsf{clos}(\Sigma)$. Then conditions (\ref{B-R-1}), (\ref{B-R-2}) 
are also satisfied for $\hb\in\Sigma$.
\end{remark}

%\begin{proposition}
%\label{extra-implicit-assumption}
%Under Assumption \ref{finite-barrier-finite-wells-assumption}, 
%let $\mu\in(\mu_1,A_{max})$. Also for some non-negative integer $n$,
%suppose that $\pi\Big(n+\tfrac{1}{2}\Big)\hb\in(0,\frac{\pi}{2}\al_1^2)$. 
%Then there exists at most one value $\mu_n(\hb)$ satisfying
%\be
%\label{B-R-at-most-one}
%\fisf(\mu_n(\hb))=
%\pi\Big(n+\tfrac{1}{2}\Big)\hb+\epsilon(\mu_n(\hb),\hb)
%\ee
%where $\epsilon(\mu,\hb)$ is a function satisfying the 
%asymptotic relations
%\be
%\label{epsilon-asympotics}
%\epsilon(\mu,\hb)=\asympt(\hb^{\frac{5}{3}}),
%\quad
%\frac{\partial\epsilon}{\partial\mu}(\mu,\hb)=o(1)
%\quad
%\text{as}
%\quad
%\hb\downarrow0.
%\ee 
%\end{proposition}
%\begin{proof}
%Assume that there are two values $\mu_n^1(\hb)<\mu_n^2(\hb)$ 
%satisfying (\ref{B-R-at-most-one}). Then there is a point
%$\nu(\hb)\in(\mu_n^1(\hb),\mu_n^2(\hb))$ such that
%\be\nn
%\frac{d\fisf}{d\mu}(\nu(\hb))=\frac{d\epsilon}{d\mu}(\nu(\hb),\hb).
%\ee
%By (\ref{derivative-bold-phi}), we have 
%$\frac{d\fisf}{d\mu}(\nu(\hb))\leq-c<0$, while the right asymptotic 
%relation in (\ref{epsilon-asympotics}) shows that 
%$\frac{\partial\epsilon}{\partial\mu}(\nu(\hb),\hb)$ tends to $0$ 
%as $\hb\downarrow0$. So, we arrived at a contradiction.
%\end{proof}

What follows is a result converse to Theorem \ref{B-R-existence}.
\begin{theorem}
\label{exist-ev-barrier}
Under Assumption \ref{finite-barrier-finite-wells-assumption}, 
suppose that for some non-negative integer $n$, the point 
$\pi\Big(n+\tfrac{1}{2}\Big)\hb$ lies in $(0,\frac{\pi}{2}\al_1^2)$. 
Then there exists a value $\widetilde{\mu}=\widetilde{\mu}(n,\hb)$ such 
that 
\be\nn
\fisf(\widetilde{\mu})=
\pi\Big(n+\tfrac{1}{2}\Big)\hb+
\asympt(\hb^{\frac{5}{3}})
\quad
\text{as}
\quad
\hb\downarrow0
\ee
and
\be\nn
Y_+(\z,\al(\widetilde{\mu}),\hb)=
\sigma_{11}(\al(\widetilde{\mu}),\hb)
Y_-(\z,\al(\widetilde{\mu}),\hb)
\ee
where
\be\nn
\sigma_{11}(\al(\widetilde{\mu}),\hb)=(-1)^n+\asympt(\hb^\frac{2}{3})
\quad
\text{as}
\quad
\hb\downarrow0.
\ee
\end{theorem}
\begin{proof}
Recall the connection coefficients $\sigma_{ij}$, $i,j=1,2$ from 
\S\ref{connection-formulas-barrier} and define the function
\be\label{definition-sigma}
\sigma(\mu,\hb)=\sigma_{12}(\al(\mu),\hb).
\ee
From (\ref{sigma1}), it is enough to show that $\sigma$ vanishes for 
some $\widetilde{\mu}=\widetilde{\mu}(n,\hb)$ satisfying
\be\nn
\Big|
\fisf(\widetilde{\mu})
-\pi\Big(n+\tfrac{1}{2}\Big)\hbar\Big|\leq C\hbar^{\frac{5}{3}}
\ee
where $C$ does not depend neither on $n$ nor on $\hb$. Then, the 
rest follow from the first asymptotic relation in 
(\ref{barrier-connection}).

From (\ref{derivative-bold-phi}) we know that $\fisf$ maps a 
neighborhood of $\widetilde{\mu}$ in a one-to-one 
way onto a neighborhood of $\fisf(\widetilde{\mu})$. 
Let $\xsf=\fisf(\mu)$ and set
\begin{equation*}
\chi(\xsf,\hbar)=
\sigma\big(\fisf^{-1}(\xsf),\hb\big)-\cos(\hb^{-1}\xsf).
\end{equation*} 
By definition (\ref{definition-sigma}) of $\sigma$ and 
the second relation in (\ref{barrier-connection}) we have 
\be\nn
|\chi(\xsf,\hb)|\leq C\hb^{\frac{2}{3}}
\ee 
for a constant $C$ independent of $\hbar$ and $\xsf$. 
With the above definitions, our equation now reads
\begin{align*}
0 & = \sigma(\mu,\hb)\\
  & = \chi(\xsf,\hb)+\cos(\hb^{-1}\xsf).
\end{align*}
So this equation has to have a solution 
$\widetilde{\xsf}=\widetilde{\xsf}(n,\hb)=
\fisf(\widetilde{\mu}(n,\hb))$ 
satisfying the estimate
\begin{equation*}
\Big|\widetilde{\xsf}-\pi\Big(n+\tfrac{1}{2}\Big)\hbar\Big|\leq 
C\hbar^{\frac{5}{3}}.
\end{equation*} 

A change of variables $s=\hbar^{-1}\xsf$ transforms our 
problem to the equivalent assertion that equation
\begin{equation}\label{s-eq}
\chi(\hbar s,\hbar)+\cos s=0
\end{equation}
has to have a solution with respect to $s$, namely 
$\widetilde{s}=\widetilde{s}(n,\hb)=\hb^{-1}\widetilde{\xsf}$,
such that
\begin{equation}\label{s-estim}
\Big|\widetilde{s}-\pi\Big(n+\tfrac{1}{2}\Big)\Big|\leq 
C\hbar^{\frac{2}{3}}.
\end{equation}
But this is true because 
\begin{equation*}
\chi(\hbar s,\hbar)=
\mathcal{O}(\hbar^{\frac{2}{3}})
\quad\text{as}\quad\hbar\downarrow0.
\end{equation*}
\end{proof}

\section{The Case of One Potential Well}
\label{passage-well}

In this section we are interested in the solutions of equation
\be
\label{schrodi-well}
\frac{d^2y}{dx^2}=
\bigg\{
\hbar^{-2}[\mu^2-A^2(x)]+
\frac{3}{4}\bigg[\frac{A'(x)}{A(x)+\mu}\bigg]^2-
\frac{1}{2}\frac{A''(x)}{A(x)+\mu}
\bigg\}
y
\ee
where $\mu>0$ and the potential function 
$A:\mathbb{R}\rightarrow\mathbb{R}_+$ behaves as a finite well 
(a well of finite width) in some bounded interval in $\mathbb{R}$. We assume the 
following (see Figure \ref{pic-well}).  

\begin{assumption}
\label{well-assumption}
The function $A$ is positive on some bounded interval $[x_1,x_2]$ in 
$\mathbb{R}$ and has a unique extremum in $(x_1,x_2)$; particularly, a 
minimum $A_{min}>0$ at $w_0\in(x_1,x_2)$. Also $A$ is assumed to be $C^4$ 
in $(x_1,x_2)$ and of class $C^5$ in a neighborhood of $w_0$. 
Additionally, $A'(x)<0$ for $x\in(x_1,w_0)$ and $A'(x)>0$ for 
$x\in(w_0,x_2)$. At $w_0$, we have $A'(w_0)=0$ and $A''(w_0)>0$. 
Furthermore, if we let $A_{**}=\min\{A(x_1),A(x_2)\}$ and take
$\mu\in(A_{min},A_{**})\subset\mathbb{R}_+$, the equation $A(x)=\mu$ 
has two solutions $w_-(\mu)$, $w_+(\mu)$ in $(x_1,x_2)$.
These satisfy $w_-<w_+$, $A(x)<\mu$ for $x\in(w_-,w_+)$ and $A(x)>\mu$ 
for $x\in(x_1,w_-)\cup(w_+,x_2)$ (the above imply $\pm A'(w_\pm)>0$). 
Finally, when $\mu=A_{min}$ the two points $w_-$, $w_+$ coalesce into
one double root at $w_0$.
\end{assumption}

\begin{figure}[H]
\centering
\begin{tikzpicture}

\begin{axis}[
    legend pos = north west,
    axis lines = none,
    xlabel = {},
    ylabel = {},
    axis line style={draw=none},
    tick style={draw=none}
]

\addplot [
    domain=0:3.2, 
    samples=100, 
    color=red,
]
{x^3-3*x^2};
\addlegendentry{$A(x)$}

%\addplot [
%    domain=0:3.2, 
%    samples=100, 
%    color=black,
%    ]
%    {1/2};
    
%\addplot [
%    domain=-4:6, 
%    samples=100, 
%    color=black,
%    ]
%    {0.85};    

\draw[scale=0.5,domain=-4.8:2,dashed,variable=\y,blue]  
plot ({1.31},{\y});

\draw[scale=0.5,domain=-4.8:-4,dashed,variable=\y,blue]  
plot ({4},{\y});

\draw[scale=0.5,domain=-4.8:2,dashed,variable=\y,blue]  
plot ({5.76},{\y});

\draw[scale=0.5,domain=-4.8:4,dashed,variable=\y,blue]  
plot ({0},{\y});

\draw[scale=0.5,domain=-4.8:8.15,dashed,variable=\y,blue]  
plot ({6.4},{\y});

\node[circle,inner sep=1.5pt,fill=black] at (0.65,-1) {};

\node[circle,inner sep=1.5pt,fill=black] at (2.87,-1) {};

\node[circle,inner sep=1.5pt,fill=black] at (2,-4) {};

\draw[dashed] (-0.2,-1) -- (3.4,-1);

\draw (-2,-4.4) -- (4,-4.4);

\draw[dashed] (-0.2,-4) -- (3.4,-4);

%\draw (0,-1) -- (0,1.5);

\end{axis}

\draw (0.5,-0.2) node {$x_1$};

\draw (6.4,-0.2) node {$x_2$};

\draw (1.8,-0.2) node {$w_-$};

\draw (5.8,-0.2) node {$w_+$};

\draw (4.15,-0.2) node {$w_0$};

\draw (7,0.5) node {$A_{min}$};

\draw (6.8,2.8) node {$\mu$};

\end{tikzpicture}
\caption{An example of a potential well between $w_-$ and $w_+$.}
\label{pic-well}
\end{figure}
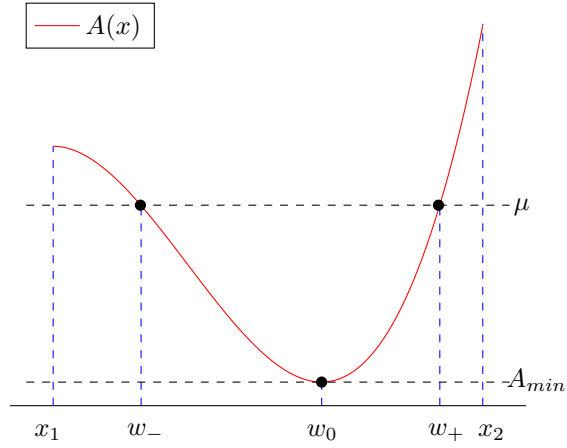

\subsection{The Liouville transform for the case of a well}
\label{liouville-transform-well}

Let us first fix some notation. We set
\be\nn
w_-\equiv w
\ee
\be\nn
J^-=(x_1,w_0),\quad J^+=(w_0,x_2)
\ee
and define
\be
\label{gi-plus-minus}
G^\pm=\Big(\restr{A}{\mathsf{clos}(J^\pm)}\Big)^{-1}.
\ee
We take an arbitrary $\mu_1\in(A_{min},A_{**})$ and consider the 
$w_{-1}\in(G^-(A_{**}),w_0)$ such that $A(w_{-1})=\mu_1$; then 
$\mu\in[A_{min},\mu_1]$ implies $w\in[w_{-1},w_0]$. For every 
$\hb>0$ our equation (\ref{schrodi-well}) reads
\be
\label{final-schrodi-well}
\frac{d^2y}{dx^2}=[\hbar^{-2}f(x,w)+g(x,w)]y,
\quad (x,w)\in(x_1,x_2)\times[w_{-1},w_0]
\ee
in which the functions $f$ and $g$ satisfy
\be
\label{f-schrodi-well}
f(x,w)=A^2(w)-A^2(x)
\ee
and
\be
\label{g-schrodi-well}
g(x,w)=
\frac{3}{4}\bigg[\frac{A'(x)}{A(x)+A(w)}\bigg]^2-
\frac{1}{2}\frac{A''(x)}{A(x)+A(w)}.
\ee
Observe that our equation possesses two simple turning points at 
$x=w_\pm$ when $w\in[w_{-1},w_0)$ which combine into one double 
at $x=w_0$ when $w$ equals $w_0$.

We introduce new variables $X$ and $\z$ according 
to the Liouville transform
\be\nn
X=\dot{x}^{-\frac{1}{2}}y
\ee
where the dot denotes differentiation with respect to $\z$. 
Equation (\ref{final-schrodi-well}) becomes
\be
\label{schrodi-well-liouville-initial-form}
\frac{d^2X}{d\z^2}=
\Big[\hbar^{-2}\dot{x}^2f(x,w)+\dot{x}^2g(x,w)+
\dot{x}^{\frac{1}{2}}\frac{d^2}{d\z^2}(\dot{x}^{-\frac{1}{2}})\Big]X.
\ee
We begin with the noncritical case $\mu\in(A_{min},\mu_1]$ with 
two turning points $w_\pm$. In this case $f(\cdot,w)$ is positive in 
$(w_-,w_+)$ and negative in $(x_1,w_-)\cup(w_+,x_2)$. Hence we 
prescribe
\be\label{well-case}
\dot{x}^2f(x,w)=\beta^2-\z^2
\ee
where $\beta>0$ is chosen in such a way that $x=w_-$ corresponds 
to $\z=-\beta$ and $x=w_+$ to $\z=\beta$ accordingly. 

The integration of (\ref{well-case}) yields
\be
\label{well-integral-form}
\int_{w}^{x}f(t,w)^{\frac{1}{2}}dt=
\int_{-\beta}^{\z}(\beta^2-\tau^2)^{\frac{1}{2}}d\tau
\ee
provided that $w_-\leq x\leq w_+$ (notice that by taking these 
integration limits, $w_-$ corresponds to $-\beta$). For the 
remaining correspondence we require
\be\nn
\int_{w_-}^{w_+}f(t,w)^{\frac{1}{2}}dt=
\int_{-\beta}^{\beta}(\beta^2-\tau^2)^{\frac{1}{2}}d\tau
\ee
yielding
\be\label{beta}
\beta^2(\mu)=
\frac{2}{\pi}\int_{w_-(\mu)}^{w_+(\mu)}
\sqrt{\mu^2-A^2(t)}dt.
\ee

For every fixed value of $\hb$, relation (\ref{beta}) defines 
$\beta$ as a continuous and increasing function of $\mu$ which 
vanishes as $w\downarrow A_{min}$. Set 
\be
\label{beta-1}
\beta_1=\beta(w_{-1})>0. 
\ee
Then $\mu\in(A_{min},\mu_1]$ implies $\beta\in(0,\beta_1]$.

Next, from (\ref{well-integral-form}) we find
\be
\label{zeta-well-center}
\int_{w_-}^{x}f(t,w)^{\frac{1}{2}}dt=
\frac{1}{2}\alpha^2\arccos\Big(-\frac{\z}{\alpha}\Big)+
\frac{1}{2}\z\big(\alpha^2-\z^2\big)^{\frac{1}{2}}
\quad\text{for}\quad
w_-\leq x\leq w_+
\ee
with  the principal value choice for the inverse cosine taking 
values in $[0,\pi]$. For the remaining $x$-intervals, we 
integrate (\ref{well-case}) to obtain
\be
\label{zeta-well-left}
\int_{x}^{w_-}[-f(t,w)]^{\frac{1}{2}}dt=-
\frac{1}{2}\alpha^2\arcosh\Big(-\frac{\z}{\alpha}\Big)-
\frac{1}{2}\z\big(\z^2-\alpha^2\big)^{\frac{1}{2}}
\quad\text{for}\quad
x_1<x\leq w_-
\ee
and
\be
\label{zeta-well-right}
\int_{w_+}^{x}[-f(t,w)]^{\frac{1}{2}}dt=-
\frac{1}{2}\alpha^2\arcosh\Big(\frac{\z}{\alpha}\Big)+
\frac{1}{2}\z\big(\z^2-\alpha^2\big)^{\frac{1}{2}}
\quad\text{for}\quad
w_+\leq x<x_2
\ee
with $\arcosh(x)=\ln\big(x+\sqrt{x^2-1}\big)$ for $x\geq1$.

Equations (\ref{zeta-well-center}), (\ref{zeta-well-left}) 
and (\ref{zeta-well-right}) show that $\z$ is a continuous and 
increasing function of $x$ which shows that there is a one-to-one 
correspondence between these two variables. Thus, if we set 
\be
\label{zj-well}
\z_j=\lim_{x\to x_j}\z(x)\quad\text{for}\hspace{4pt}j=1,2
\ee
then $(x_1,x_2)$ is mapped by $\z$ to $(\z_1,\z_2)$. 
Notice that $-\infty<\z_1<0<\z_2<+\infty$.

\begin{remark}
\label{critical-well-liouville}
In the critical case in which the two (simple) turning points 
coalesce into one (double) point, we get a limit of 
the above transformation with $w=w_0$. In this case, the relevant 
relations to (\ref{zeta-well-center}), (\ref{zeta-well-left}), 
(\ref{zeta-well-right}) are
\be
\label{zeta-well-0-left}
\int_{x}^{w_0}[-f(t,w_0)]^{\frac{1}{2}}dt=
\frac{1}{2}\z^2
\quad\text{for}\quad
x_1<x\leq w_0
\ee
\be
\label{zeta-well-0-right}
\int_{w_0}^{x}[-f(t,w_0)]^{\frac{1}{2}}dt=
\frac{1}{2}\z^2
\quad\text{for}\quad
w_0\leq x<x_2
\ee
and $\beta=0$.
\end{remark}

Consequently, noticing Remark \ref{critical-barrier-liouville}, 
we substitute (\ref{well-case}) in 
(\ref{schrodi-well-liouville-initial-form}) and obtain the 
following proposition.
\begin{proposition}
For every $\hb>0$ equation
\be\nn
\frac{d^2y}{dx^2}=[\hbar^{-2}f(x,w)+g(x,w)]y,
\quad (x,w)\in(x_1,x_2)\times[w_{-1},w_0]
\ee
where $f$, $g$ as in (\ref{f-schrodi-well}), 
(\ref{g-schrodi-well}) respectively, is transformed to the 
equation
\be
\label{schrodi-well-liouville-final-form}
\frac{d^2X}{d\z^2}=
\big[\hb^{-2}(\beta^2-\z^2)+\overline{\ps}(\z,\beta)\big]X,
\quad 
(\z,\beta)\in(\z_1,\z_2)\times[0,\beta_1]
\ee
in which $\z$ is given by the Liouville transform 
(\ref{well-case}), $\beta$ is given by (\ref{beta}), 
$\z_j$, $j=1,2$ are given by 
(\ref{zj-well}), $\beta_1$ as in (\ref{beta-1}) and 
the function $\overline{\ps}(\z,\beta)$ is given by the 
formula
\be
\label{psi-well}
\overline{\ps}(\z,\beta)=
\dot{x}^2g(x,w)+
\dot{x}^{\frac{1}{2}}\frac{d^2}{d\z^2}(\dot{x}^{-\frac{1}{2}}).
\ee
\end{proposition}

In the following paragraphs we shall be interested in 
approximate solutions of equation 
(\ref{schrodi-well-liouville-final-form}), so we introduce the 
following terminology.
\begin{definition}
The function $\overline{\psi}$ found in the differential equation 
(\ref{schrodi-well-liouville-final-form}) shall be called the 
\textbf{error term} of this equation.
\end{definition}
For the error term we have the following proposition.
\begin{proposition}
The error term $\overline{\psi}$ can be written equivalently as
\begin{multline}
\label{psi-well-equivalent}
\overline{\ps}(\z,\beta)=
\frac{1}{4}\frac{3\z^2 +2\beta^2}{(\beta^2-\z^2)^2}
+\frac{1}{16}\frac{\beta^2-\z^2}{f^3(x,w)}
\Big\{4f(x,w)f''(x,w)-5[f'(x,w)]^2\Big\}\\
+(\beta^2-\z^2)\frac{g(x,w)}{f(x,w)}
\end{multline}
where prime denotes differentiation with respect to $x$. The 
same formula can be used in the critical case of one double 
turning point simply by setting $w=w_0$ and $\beta=0$.
\end{proposition}
\begin{proof}
Using (\ref{psi-well}), (\ref{g-schrodi-well}) and 
(\ref{well-case}), simple algebraic manipulations shown that 
$\overline{\psi}$ takes the desired form. 
\end{proof}

\subsection{Continuity of the error term in the case of a well}
\label{error-cont-well}

In this subsection we prove that the function 
$\overline{\ps}(\z,\beta)$ resulting from the Liouville 
transformation defined above, is continuous in $\z$ and $\beta$. 
This will be used subsequently to prove the existence of 
approximate solutions of equation 
(\ref{schrodi-well-liouville-final-form}). We have the following.

\begin{lemma}
The function $\overline{\ps}(\z,\beta)$ defined in 
(\ref{psi-well}), is continuous in $\z$ and $\beta$ in the region
$(\z_1,\z_2)\times[0,\beta_1]$ of the $(\z,\beta)$-plane.
\end{lemma}
\begin{proof}
For $x\in(x_1,x_2)$, $\mu\in[A_{min},\mu_1]$ and 
$w\in[w_{-1},w_0]$ we introduce an auxiliary function $q$ 
by setting
\be
\label{p-well}
f(x,w)=(w-x)(w_+-x)q(x,w).
\ee
Having in mind that $A(w)=A(w_+)=\mu$, we see that for 
$\mu\in(A_{min},\mu_1]$
\be\nn
q(w_\pm,w)=\pm\frac{2\mu}{w-w_+} A'(w_\pm)<0 
\ee
while for $\mu=A_{min}$
\be\nn
q(w_0,w_0)=-A_{min}A''(w_0)<0.
\ee 

Our functions $f$, $g$ and $q$ defined by (\ref{f-schrodi-well}), 
(\ref{g-schrodi-well}) and (\ref{p-well}) respectively satisfy the 
following properties
\begin{itemize}
\item[(i)]
$q$, $\frac{\partial q}{\partial x}$, $\frac{\partial^2 q}{\partial x^2}$ 
and $g$ are continuous functions of $x$ and $w$ in the region 
$(x_1,x_2)\times[w_{-1},w_0]$ 
\item[(ii)]
$q$ is negative throughout the same region
\item[(iii)]
$|\frac{\partial^3 q}{\partial x^3}|$ is bounded in a neighborhood 
of the point $(x,w)=(w_0,w_0)$ in the same region and
\item[(iv)]
$f$ is a non-increasing function of $w\in[w_{-1},w_0]$ when 
$x\in[w,w_+]$.
\end{itemize}
As in \S\ref{error-cont-barrier} these relations follow directly
from (\ref{f-schrodi-well}), (\ref{g-schrodi-well}), (\ref{p-well})
and Assumption \ref{well-assumption}. By referring again to Lemma I 
in \cite{olver1975} (actually a slight variant of it properly defined for 
case III treated in Olver's \cite{olver1975}), the function $\overline{\ps}$  
defined by (\ref{psi-well}) (or (\ref{psi-well-equivalent})) is continuous in 
the corresponding region of the $(\z,\beta)$-plane.
\end{proof}

\subsection{Approximate solutions in the case of a well}
\label{well-approximate-solutions}

Here we state a theorem concerning approximate solutions of 
equation (\ref{schrodi-well-liouville-final-form}).
First we define a balancing function $\Om$ as in the barrier 
case using (\ref{omega-barrier}). Now we define an 
\textit{error-control function} which will provide
us with a way to assess the error.  
\begin{definition}
As an \textbf{error-control function} $\overline{H}(\z,\beta,\hb)$ of 
equation (\ref{schrodi-well-liouville-final-form}) we consider 
any primitive of the function
\be\nn
\frac{\overline{\ps}(\z,\beta)}{\Om(\z\sqrt{2\hb^{-1}})}.
\ee
\end{definition}

As in \S\ref{barrier-approximate-solutions}, we define the 
\textit{variation} of $\overline{H}$ in an interval 
$(\gamma,\delta)\subseteq(\z_1,\z_2)\subset\R$ 
(cf. (\ref{zj-well})).
\begin{definition}
The variation $\var_{\gamma,\delta}\Big[\overline{H}\Big]$ 
in the interval $(\gamma,\delta)$ of the error-control function 
$\overline{H}$ of equation (\ref{schrodi-well-liouville-final-form}) 
is defined by
\be\nn
\var_{\gamma,\delta}\Big[\overline{H}\Big](\beta,\hb)=
\int_{\gamma}^{\delta}
\frac{|\overline{\ps}(t,\beta)|}{\Om(t\sqrt{2\hb^{-1}})}dt.
\ee
\end{definition}

Finally, for any $c\geq0$ set
\be\label{lambda-2-function}
l_2(c)=\sup_{x\in(0,+\infty)}
\Big\{\Om(x)\overline{\msf}(x,c)^2\Big\}
\ee
where $\overline{\msf}$ is a function defined in terms of modified
Parabolic Cylinder Functions in section 
\ref{modified-pcfs} of the appendix.
We note that the above supremum is finite for each value of $c$. 
This fact is a consequence of (\ref{omega-barrier}) and the 
first relation in (\ref{M,N-bar-asymptotics}). Furthermore, because 
the relations (\ref{M,N-bar-asymptotics}) hold uniformly in compact 
intervals of the parameter $c$, the function $l_2$ is continuous.

Now the existence of approximate solutions is quaranted by the following.
\btheo\label{main-thm-well}
For each value of $\hb>0$, equation 
\be\nn
\frac{d^2X}{d\z^2}=
\big[\hb^{-2}(\beta^2-\z^2)+\overline{\ps}(\z,\beta)\big]X
\ee
has in the region $[0,\z_2)\times[0,\beta_1]$ of the 
$(\z,\beta)$-plane, two solutions $Y_+$ and $Z_+$. They satisfy
\bea
\label{z1-approx}
Y_+(\z,\beta,\hb)=
k(\tfrac{1}{2}\hb^{-1}\beta^2)^{\frac{1}{2}}
W(-\z\sqrt{2\hb^{-1}},\tfrac{1}{2}\hb^{-1}\beta^2)+
\overline{\eps}_1 (\z,\beta,\hb)\\
\label{z2-approx}
Z_+(\z,\beta,\hb)=
k(\tfrac{1}{2}\hb^{-1}\beta^2)^{-\frac{1}{2}}
W(\z\sqrt{2\hb^{-1}},\tfrac{1}{2}\hb^{-1}\beta^2)+
\overline{\eps}_2 (\z,\beta,\hb)
\eea
where $k$, $W$ are functions found in appendix 
\ref{modified-pcfs} about modified PCFs. These $Y_+$, $Z_+$
are continuous and have continuous first and second partial 
$\z$-derivatives. The errors $\overline{\eps}_1$, $\overline{\eps}_2$ 
satisfy
\begin{multline}
\label{well-junk1}
\frac{|\overline{\eps}_1(\z,\beta,\hb)|}
{\overline{\msf}(\z\sqrt{2\hb^{-1}},\tfrac{1}{2}\hb^{-1}\beta^2)},
\frac{\Big|
\frac{\partial\overline{\eps}_1}{\partial\z}(\z,\beta,\hb)\Big|}
{\sqrt{2\hb^{-1}}\hspace{4pt}
\overline{\nsf}(\z\sqrt{2\hb^{-1}},\tfrac{1}{2}\hb^{-1}\beta^2)}\\
\leq
\overline{\esf}(\z\sqrt{2\hb^{-1}},\tfrac{1}{2}\hb^{-1}\beta^2)
\Big(\exp\bigg\{\frac{l_2(\tfrac{1}{2}\hb^{-1}\beta^2)}{\sqrt{2\hb^{-1}}}
\mathcal{V}_{0,\z}\Big[\overline{H}\Big](\beta,\hb)\bigg\}-1\Big)
\end{multline}
and
\begin{multline}
\label{well-junk2}
\frac{|\overline{\eps}_2(\z,\beta,\hb)|}
{\overline{\msf}(\z\sqrt{2\hb^{-1}},\tfrac{1}{2}\hb^{-1}\beta^2)},
\frac{\Big|
\frac{\partial\overline{\eps}_2}{\partial\z}(\z,\beta,\hb)\Big|}
{\sqrt{2\hb^{-1}}\hspace{4pt}
\overline{\nsf}(\z\sqrt{2\hb^{-1}},\tfrac{1}{2}\hb^{-1}\beta^2)}\\
\leq
\frac{1}{\overline{\esf}(\z\sqrt{2\hb^{-1}},\tfrac{1}{2}\hb^{-1}\beta^2)}
\Big(\exp\bigg\{\frac{l_2(\tfrac{1}{2}\hb^{-1}\beta^2)}{\sqrt{2\hb^{-1}}}
\mathcal{V}_{\z,\z_2}\Big[\overline{H}\Big](\beta,\hb)\bigg\}-1\Big).
\end{multline}
\etheo
\begin{proof}
The proof is similar to that of Theorem \ref{main-thm-barrier} in
\S\ref{barrier-approximate-solutions} and details need not be recorded.  
\end{proof}

\subsection{Asymptotics of the approximate solutions for the well}
\label{asympt-behave-well-sols}

As in the case of the function $l_1$ in 
\S\ref{asympt-behave-barrier-sols}, we find that $l_2$ is continuous 
in $[0,+\infty)$. Using (\ref{omega-barrier}) and an analysis similar 
to that mentioned in \S\ref{asympt-behave-barrier-sols} we find
\be\label{lambda2-asympt}
l_2(\tfrac{1}{2}\hb^{-1}\beta^2)=\asympt(1)
\quad\text{as}\quad
\hb\downarrow0.
\ee
Next, $\var_{0,\z_2}\Big[\overline{H}\Big](\beta,\hb)$ can be examined 
as in \S 8 of \cite{h+k}. We find that
\begin{align}\label{v-operator-asympt-well}
\var_{0,\z_2}\Big[\overline{H}\Big](\beta,\hb)=
\int_{0}^{\z_2}\frac{|\overline{\ps}(t,\beta)|}
{1+(t\sqrt{2\hb^{-1}})^{\frac{1}{3}}}dt=
\asympt(\hb^{1/6})
\quad\text{as}\quad
\hb\downarrow0.
\end{align}

The last two relations applied to (\ref{well-junk1}) and 
(\ref{well-junk2}) return as $\hb\downarrow0$
\begin{align}
\label{e1-asympt-well}
\overline{\eps}_1(\z,\beta,\hb) & =
\overline{\esf}(\z\sqrt{2\hb^{-1}},\tfrac{1}{2}\hb^{-1}\beta^2)
\overline{\msf}(\z\sqrt{2\hb^{-1}},\tfrac{1}{2}\hb^{-1}\beta^2)
\asympt(\hb^{\frac{2}{3}})
\\
\nn
\overline{\eps}_2(\z,\beta,\hb) & =
\frac{\overline{\msf}(\z\sqrt{2\hb^{-1}},\tfrac{1}{2}\hb^{-1}\beta^2)}
{\overline{\esf}(\z\sqrt{2\hb^{-1}},\tfrac{1}{2}\hb^{-1}\beta^2)}
\asympt(\hb^{\frac{2}{3}})
\\
\nn
\frac{\partial\overline{\eps}_1}{\partial\z}(\z,\beta,\hb) & =
\overline{\esf}(\z\sqrt{2\hb^{-1}},\tfrac{1}{2}\hb^{-1}\beta^2)
\overline{\nsf}(\z\sqrt{2\hb^{-1}},\tfrac{1}{2}\hb^{-1}\beta^2)
\asympt(\hb^{\frac{1}{6}})
\\
\nn
\frac{\partial\overline{\eps}_2}{\partial\z}(\z,\beta,\hb) & =
\frac{\overline{\nsf}(\z\sqrt{2\hb^{-1}},\tfrac{1}{2}\hb^{-1}\beta^2)}
{\overline{\esf}(\z\sqrt{2\hb^{-1}},\tfrac{1}{2}\hb^{-1}\beta^2)}
\asympt(\hb^{\frac{1}{6}})
\end{align}
uniformly for $\z\in[0,\z_2)$ and $\beta\in(0,\beta_1]$.

\begin{remark}
In the special case $\beta=0$ (i.e. when equation 
(\ref{schrodi-well-liouville-final-form}) has a double 
turning point at $\z=0$), $l_2(0)$ is independent of $\hb$. 
Using the definition (\ref{omega-barrier}) of $\Om$,
we see that we have a similar estimate to 
(\ref{v-operator-asympt-barrier}); namely
$\var_{0,\z_2}\Big[\overline{H}\Big](0,\hb)=\asympt(\hb^{\frac{1}{6}})$ 
as $\hb\downarrow0$. Hence the results about the errors above  
hold for the case $\beta=0$ too.
\end{remark}

\subsection{Connection formulae for a well}
\label{connection-formulas-well}

Here, we determine the 
asymptotic behavior of $Y_+$, $Z_+$ for small $\hb>0$ and $\z<0$ by 
establishing appropriate \textit{connection formulae}. We can 
replace $\z$ by $-\z$ in Theorem
\ref{main-thm-well} to ensure two more solutions $Y_-$, $Z_-$ of 
equation (\ref{schrodi-well-liouville-final-form}) satisfying
as $\hb\downarrow0$
\begin{align*}
Y_-(\z,\al,\hb) & =
k(\tfrac{1}{2}\hb^{-1}\beta^2)^{\frac{1}{2}}
W(\z\sqrt{2\hb^{-1}},\tfrac{1}{2}\hb^{-1}\beta^2)+\\
&\hspace{2cm}
\overline{\esf}(\z\sqrt{2\hb^{-1}},\tfrac{1}{2}\hb^{-1}\beta^2)
\overline{\msf}(\z\sqrt{2\hb^{-1}},\tfrac{1}{2}\hb^{-1}\beta^2)
\asympt(\hb^{\frac{2}{3}})
\\
\nn
Z_-(\z,\al,\hb) & =
k(\tfrac{1}{2}\hb^{-1}\beta^2)^{-\frac{1}{2}}
W(-\z\sqrt{2\hb^{-1}},\tfrac{1}{2}\hb^{-1}\beta^2)+
\frac{\overline{\msf}(\z\sqrt{2\hb^{-1}},\tfrac{1}{2}\hb^{-1}\beta^2)}
{\overline{\esf}(\z\sqrt{2\hb^{-1}},\tfrac{1}{2}\hb^{-1}\beta^2)}
\asympt(\hb^{\frac{2}{3}})
\end{align*}
uniformly for $\z\in(\z_1,0]$ and $\beta\in[0,\beta_1]$.

\begin{remark}
\label{linear-independent-well}
The two sets $\{Y_+,Z_+\}$ and $\{Y_-,Z_-\}$ consist of two
linearly independent functions. This can be seen by their Wronskians.
For example, using \ref{wronskian-mpcf} we have
\be\nn
\W
\Big[
W(\cdot,\tfrac{1}{2}\hb^{-1}\beta^2),
W(-\cdot,\tfrac{1}{2}\hb^{-1}\beta^2)
\Big]
=1
\ee
Using this and (\ref{z1-approx}), (\ref{z2-approx}) we see that
$\W[Y_+,Z_+]\neq0$. Similarly, we have $\W[Y_-,Z_-]\neq0$ as well.
\end{remark}

We express $Y_+$, $Z_+$ in terms of $Y_-$, $Z_-$. So for
$(\z,\beta)\in(\z_1,0]\times[0,\beta_1]$ we write
\begin{align}\label{tau1}
Y_+(\z,\beta,\hb) & =
\tau_{11}(\beta,\hb)Y_-(\z,\beta,\hb)+
\tau_{12}(\beta,\hb)Z_-(\z,\beta,\hb)\\
\label{tau2}
Z_+(\z,\beta,\hb) & =
\tau_{21}(\beta,\hb)Y_-(\z,\beta,\hb)+
\tau_{22}(\beta,\hb)Z_-(\z,\beta,\hb).
\end{align}
As in \S \ref{connection-formulas-barrier}, we  find approximations
for the coefficients $\tau_{ij}$, $i,j=1,2$ in the linear relations
(\ref{tau1}) and (\ref{tau2}). We take equations (\ref{tau1}), 
(\ref{tau2}) along with their derivatives and evaluate them at $\z=0$. 
We obtain
\be
\label{tau-wronski}
\begin{split}
\tau_{11}(\beta,\hb)=
\frac{\W[Y_+(\cdot,\beta,\hb),Z_-(\cdot,\beta,\hb)](0)}
{\W[Y_-(\cdot,\beta,\hb),Z_-(\cdot,\beta,\hb)](0)}\\
\tau_{12}(\beta,\hb)=
-
\frac
{\W[Y_+(\cdot,\beta,\hb),Y_-(\cdot,\beta,\hb)](0)}
{\W[Y_-(\cdot,\beta,\hb),Z_-(\cdot,\beta,\hb)](0)}\\
\tau_{21}(\beta,\hb)=
\frac
{\W[Z_+(\cdot,\beta,\hb),Z_-(\cdot,\beta,\hb)](0)}
{\W[Y_-(\cdot,\beta,\hb),Z_-(\cdot,\beta,\hb)](0)}\\
\tau_{22}(\beta,\hb)=
-
\frac
{\W[Z_+(\cdot,\beta,\hb),Y_-(\cdot,\beta,\hb)](0)}
{\W[Y_-(\cdot,\beta,\hb),Z_-(\cdot,\beta,\hb)](0)}.
\end{split}
\ee

By using the results and properties of modified Parabolic Cylinder 
Functions and their auxiliary functions from section 
\ref{modified-pcfs} in the appendix, we find that as $\hb\downarrow0$
\be
\label{well-connection}
\begin{split}
\tau_{11}(\beta,\hb)=
\asympt(\hb^{\frac{2}{3}})\\
\tau_{12}(\beta,\hb)=
k(\tfrac{1}{2}\hb^{-1}\beta^2)
[1+\asympt(\hb^{\frac{2}{3}})]
\\
\tau_{21}(\beta,\hb)=
k(\tfrac{1}{2}\hb^{-1}\beta^2)^{-1}
[1+\asympt(\hb^{\frac{2}{3}})]
\\
\tau_{22}(\beta,\hb)=
\asympt(\hb^{\frac{2}{3}})
\end{split}
\ee
uniformly for $\beta\in[0,\beta_1]$.

We close this section with a useful lemma that shall be used in 
next paragraph's main theorem.
\begin{lemma}
The matrix $\tau$ formed by the connection coefficients in 
(\ref{tau1}), (\ref{tau2}) satisfies
\be
\label{tau-determinant}
\det\tau=
\det
\begin{bmatrix}
\tau_{11} & \tau_{12}\\
\tau_{21} & \tau_{22}
\end{bmatrix}
=-1+\asympt(\hb^\frac{2}{3})
\quad\text{as}\quad
\hb\downarrow0.
\ee
\end{lemma}
\begin{proof}
Simply use formulae (\ref{well-connection}).
\end{proof}

\begin{remark}
Using (\ref{tau-wronski}), a straightforward calculation yields
$$
\W[Y_+,Z_+]=\W[\tau_{11}Y_-+\tau_{12}Z_-,Z_+]=(\det\tau)\W[Y_-,Z_-]
$$
whence
\be
\label{wronski-connection}
\det\tau=\frac{\W[Y_+,Z_+]}{\W[Y_-,Z_-]}.
\ee
\end{remark}

\subsection{Applications in the case of a well}
\label{applications-well}
%We start with a definition.
%\begin{definition}
%Define the functions $\fisth_1(x,\mu)$, $\fisth_2(x,\mu)$ by
%\be
%\label{bold-theta-1-2}
%\begin{split}
%\fisth_1(x,\mu)=
%\int_{w_+(\mu)}^x\sqrt{A(t)^2-\mu^2}dt
%\quad\text{for}\quad
%(x,\mu)\in[w_+,x_2)\times[A_{min},\mu_1]\\
%\fisth_2(x,\mu)=
%\int_x^{w_-(\mu)}\sqrt{A(t)^2-\mu^2}dt
%\quad\text{for}\quad
%(x,\mu)\in(x_1,w_-]\times[A_{min},\mu_1].
%\end{split}
%\ee
%\end{definition}
%We can differentiate $\fisth_j(x,\mu)$, $j=1,2$ with respect to $\mu$. 
%Indeed, we take
%\be
%\label{derivative-bold-theta-1-2}
%\begin{split}
%\frac{\partial\fisth_1}{\partial\mu}(x,\mu)=
%-\mu
%\int_{w_{+}(\mu)}^{x}[A(t)^2-\mu^2]^{-\frac{1}{2}}dt<0\\
%\frac{\partial\fisth_2}{\partial\mu}(x,\mu)=
%-\mu
%\int_{x}^{w_{-}(\mu)}[A(t)^2-\mu^2]^{-\frac{1}{2}}dt<0
%\end{split}
%\ee
%whence we see that $\fisth_j(x,\mu)$, $j=1,2$ are bounded 
%functions with respect to $\mu$.

The asymptotic behavior as $\hb\downarrow0$
of an arbitrary non-trivial \textit{real} solution $X$ of equation 
(\ref{schrodi-well-liouville-final-form}) on the $\z$-interval
corresponding to the finite $x$-well $(w_-,w_+)$ of $A$, can
be examined through the functions $Y_+$, $Z_+$ and $Y_-$, $Z_-$.
Since $\{Y_+,Z_+\}$ and $\{Y_-,Z_-\}$ are two sets of linearly 
independent functions (cf. Remark \ref{linear-independent-barrier}), 
for $X$ we can write
\be
\label{X-beta-linear-combo}
\begin{split}
X(\z,\beta,\hb) 
=
\gamma_1(\beta,\hb)Y_+(\z,\beta,\hb)+
\delta_1(\beta,\hb)Z_+(\z,\beta,\hb)\\
=
\gamma_2(\beta,\hb)Y_-(\z,\beta,\hb)+
\delta_2(\beta,\hb)Z_-(\z,\beta,\hb)
\end{split}
\ee
for some $\gamma_j(\beta,\hb),\delta_j(\beta,\hb)\in\mathbb{R}$, 
$j=1,2$. For $j=1,2$ we put
\be
\label{X-beta-aux}
\begin{split}
v_j(\beta,\hb)
=
\sqrt{\gamma_j(\beta,\hb)^2+\delta_j(\beta,\hb)^2}\\
\gamma_j(\beta,\hb)
=
v_j(\beta,\hb)\cos\xi_j(\beta,\hb)\\
\delta_j(\beta,\hb)
=
v_j(\beta,\hb)\sin\xi_j(\beta,\hb)\\
\xi_j(\beta,\hb)
\in
\mathbb{R}/(2\pi\mathbb{Z}).
\end{split}
\ee

We start with a theorem.
\begin{theorem}
Under Assumption \ref{well-assumption}, 
an arbitrary real solution $X$ of equation 
(\ref{schrodi-well-liouville-final-form}) is given by the 
formulae (\ref{X-beta-linear-combo}), where the phases 
$\xi_j$, $j=1,2$ satisfy the estimate
\be
\label{well-phases}
\sin\xi_1(\beta,\hb)\sin\xi_2(\beta,\hb)=
\asympt(\hb^\frac{4}{3})
\quad\text{as}\quad
\hb\downarrow0. 
\ee
\end{theorem}
\begin{proof}
We start with (\ref{X-beta-linear-combo}), i.e.
\be\nn
\gamma_1Y_++\delta_1Z_+=\gamma_2Y_-+\delta_2Z_-
\ee
where we do not mention the dependence on $\beta$, $\hb$ for simplicity 
and take the Wronskian of both sides with $Y_+$. Using 
(\ref{X-beta-aux}), (\ref{tau-wronski}) and (\ref{wronski-connection}) 
we see that
\be\nn
\det\tau\cdot v_1\cdot\sin\xi_1=
v_2\cdot(\tau_{11}\sin\xi_2-\tau_{12}\cos\xi_2).
\ee
Finally, relying on (\ref{tau-determinant}), (\ref{well-connection}) 
and (\ref{kappa-asymptotics}) we obtain
\be\nn
v_1(\beta,\hb)\sin\xi_1(\beta,\hb)=
v_2(\beta,\hb)\asympt(\hb^\frac{2}{3})
\quad\text{as}\quad
\hb\downarrow0.
\ee
Similarly, one has
\be\nn
v_2(\beta,\hb)\sin\xi_2(\beta,\hb)=
v_1(\beta,\hb)\asympt(\hb^\frac{2}{3})
\quad\text{as}\quad
\hb\downarrow0.
\ee
Multiplying the last two equations and neglecting the 
common factor $v_1(\beta,\hb)v_2(\beta,\hb)$ we arrive at the 
desired result.
\end{proof}

The theorem above gives rise to the next corollary, the proof of 
which is staightforward.
\begin{corollary}
For every $\hb>0$, at least one of the phases $\xi_j$, $j=1,2$
satisfies the condition
\be
\label{well-one-phase}
\sin\xi_j(\beta,\hb)=
\asympt(\hb^\frac{2}{3})
\quad\text{as}\quad
\hb\downarrow0. 
\ee
\end{corollary}
The results above can be reformulated in the following theorem
\begin{theorem}
\label{fix-condo-at-least-on-one-end}
Under Assumption \ref{well-assumption}, an arbitrary real solution 
$X$ of equation (\ref{schrodi-well-liouville-final-form}) admits
\be
\label{left-outside-well}
X(\z,\beta,\hb)=
\gamma_2(\beta,\hb)Y_-(\z,\beta,\hb)+
\delta_2(\beta,\hb)Z_-(\z,\beta,\hb)
\quad\text{for}\quad
\z\in(\z_1,\z(w_-))
\ee
and
\be
\label{right-outside-well}
X(\z,\beta,\hb)=
\gamma_1(\beta,\hb)Y_+(\z,\beta,\hb)+
\delta_1(\beta,\hb)Z_+(\z,\beta,\hb)
\quad\text{for}\quad
\z\in(\z(w_+),\z_2)
\ee
where for the phases in (\ref{X-beta-aux}) we have 
\be\label{fix-cond-def}
\sin\xi_j(\beta,\hb)=\asympt(\hb^{2/3}), \quad\text{as}\quad\hb\downarrow0
\ee  
at least for one $j=1,2$.  We call (\ref{fix-cond-def}) a \textbf{fixing condition}. 
\end{theorem}

\section{Using the Liouville Transform for our Problem}

In this section we show how our initial problem can be shaped to  one 
for which the Liouville transform (i.e.  Olver's theory) can be applied.  After some 
preparatory notational comments, we state the problem explicitly and transform it 
to a new one relevant with that ones on paragraphs \S\S 
\ref{passage-barrier}, \ref{passage-well}.   The main assumption that shall be used 
for the potential of our Dirac operator is the following.

\begin{assumption}
\label{multi-hump-potential-assumption}
The function $A:\mathbb{R}\rightarrow\mathbb{R}$ is positive,  
of class $C^4(\mathbb{R})$ and such that $\lim_{x\to\pm\infty}A(x)=0$. 
It has finitely many local extrema and a maximum denoted by $A_{max}$
\footnote{
This maximum can be realized in more than one points simultaneously.
}. 
Furthermore, in some neighborhoods of these extrema it is of class $C^5$. 
Additionally, at these extreme points $A'$ vanishes, while $A''$ is either 
positive (leading to local minima) or negative (for local maxima and maximum). 
Also, for $\mu>0$, equation $A(x)=\mu$ has only finitely many solutions. 
Finally, there exists a number $\tau>0$ such that as $|x|\to+\infty$ we have
\begin{center}
$A(x)=\mathcal{O}\Big(\tfrac{1}{|x|^{1+\tau}}\Big)$\\
$A'(x)=\mathcal{O}\Big(\tfrac{1}{|x|^{2+\tau}}\Big)$\\
$A''(x)=\mathcal{O}\Big(\tfrac{1}{|x|^{3+\tau}}\Big)$.
\end{center}
\end{assumption}

\begin{figure}[H]
\centering
\begin{tikzpicture}[scale=1.5]

\begin{axis}[
    legend pos = north west,
    axis lines = none,
    xlabel = {},
    ylabel = {},
    axis line style={draw=none},
    tick style={draw=none}
]

\addplot [
    domain=-12:12, 
    samples=500, 
    color=red,
]
{1/(1+(x+5)^2)+2/(1+x^2)+3/(1+(x-5)^2)};
\addlegendentry{$A(x)$} 

\end{axis}

\draw[scale=0.5,domain=0.6:3.6,dashed,variable=\y,blue]  
plot ({4.25},{\y});

\draw[scale=0.5,domain=0.6:3.6,dashed,variable=\y,blue]  
plot ({4.75},{\y});

\draw[scale=0.5,domain=0.6:3.6,dashed,variable=\y,blue]  
plot ({6.25},{\y});

\draw[scale=0.5,domain=0.6:3.6,dashed,variable=\y,blue]  
plot ({7.55},{\y});

\draw[scale=0.5,domain=0.6:3.6,dashed,variable=\y,blue]  
plot ({8.4},{\y});

\draw[scale=0.5,domain=0.6:3.6,dashed,variable=\y,blue]  
plot ({10},{\y});

\node[circle,inner sep=1.5pt,fill=black] at (2.12,1.8) {};

\node[circle,inner sep=1.5pt,fill=black] at (2.37,1.8) {};

\node[circle,inner sep=1.5pt,fill=black] at (3.12,1.8) {};

\node[circle,inner sep=1.5pt,fill=black] at (3.77,1.8) {};

\node[circle,inner sep=1.5pt,fill=black] at (4.2,1.8) {};

\node[circle,inner sep=1.5pt,fill=black] at (5,1.8) {};

\draw [line width=1mm] (2.12,0.3) -- (2.37,0.3);

\draw [line width=1mm] (3.12,0.3) -- (3.77,0.3);

\draw [line width=1mm] (4.2,0.3) -- (5,0.3);

\draw[dashed] (0.8,1.8) -- (6,1.8);

\draw (0,0.3) -- (7,0.3);

\draw (2,0) node {$x_1^-$};

\draw (2.4,0) node {$x_1^+$};

\draw (3.12,0) node {$x_2^-$};

\draw (3.77,0) node {$x_2^+$};

\draw (4.2,0) node {$x_3^-$};

\draw (5,0) node {$x_3^+$};

\draw (6.5,1.8) node {$\mu$};

\end{tikzpicture}
\caption{An example of a multi-humped potential.}
\label{pic-3-barrier}
\end{figure}
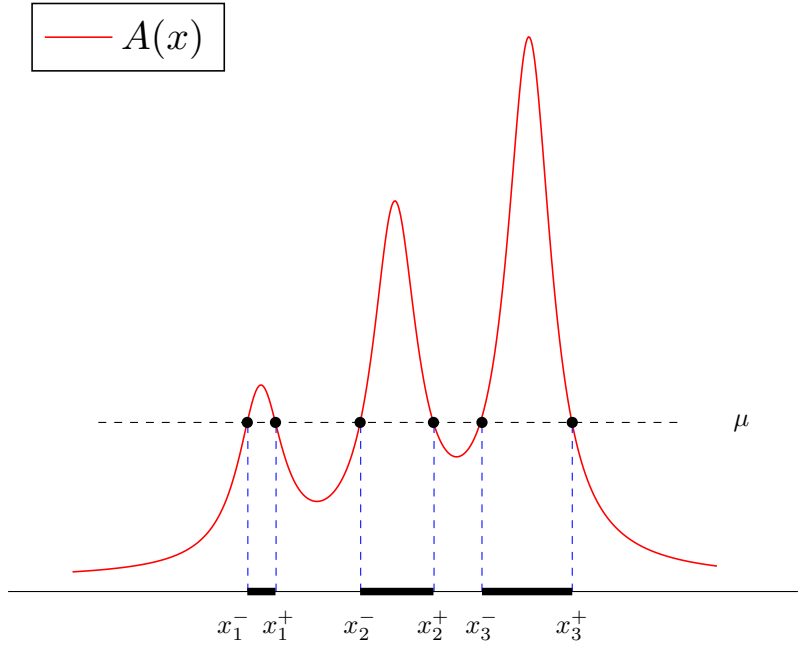

\subsection{Notation}
\label{notation}

Let us begin by fixing some notation so that we can use it for our purposes. 
The zeros of equation $A(x)=\mu$ for $\mu>0$ can either be simple or double 
(when they hit an extreme point). Let us first deal with the (non-critical) 
case where \textit{all} the zeros of this equation are \textit{simple}. 
In such a case, there is a number $L\in\mathbb{N}$ so that we can set 
$x_\ell^{\pm}=x_\ell^{\pm}(\mu)$, $\ell=1,\dots,L$, for these solutions. 
We enumerate them as follows (see Figure \ref{pic-3-barrier})
\be\nn
x_1^{-}<x_1^{+}<x_2^{-}<x_2^{+}<\cdots<x_L^{-}<x_L^{+}.
\ee
%The equalities above are realized when we are in the presence of a 
%double zero; for such $\ell$ we then have 
%$x_\ell^{-}\equiv x_\ell^{+}$ (when we hit a local maximum) and/or 
%$x_\ell^{+}\equiv x_{\ell+1}^{-}$ (in the case of hitting a local 
%minimum). 
%The  case where all the zeros of equation $A(x)=\mu$ 
%are simple, i.e.
%\be\nn
%x_1^{-}<x_1^{+}<x_2^{-}<x_2^{+}<\cdots<x_L^{-}<x_L^{+}
%\ee
Obviously, the number $L$ counts the number of finite barriers that are 
present. Hence, this yields $L$ barriers 
$\mathfrak{B}_\ell(\mu)=(x_\ell^{-}(\mu),x_\ell^{+}(\mu))$, 
$\ell=1,\dots,L$ of finite width (finite barriers) separated 
by $L-1$ wells 
$\mathfrak{W}_\ell(\mu)=(x_\ell^{+}(\mu),x_{\ell+1}^{-}(\mu))$,
$\ell=1,\dots,L-1$ of finite width (finite wells). 
We also have two infinite wells (i.e. wells of infinite width) 
$\mathfrak{W}_0(\mu)=(-\infty,x_{1}^{-}(\mu))$ 
and
$\mathfrak{W}_L(\mu)=(x_L^{+}(\mu),+\infty)$. Observe that
$\pm A'(x_\ell^{\pm}(\mu))<0$ for all $\ell=1,\dots,L$. 
Also, let $b_\ell^{0}(\mu)\in\mathfrak{B}_\ell(\mu)$, $\ell=1,\dots,L$ and 
$w_\ell^{0}(\mu)\in\mathfrak{W}_\ell(\mu)$, $\ell=1,\dots,L-1$ denote
the points where $A$ has its extremes. 

Using this notation, we define 
for $\ell=1,\dots,L$ the intervals 
\be\nn
I_\ell^{-}=(x_\ell^{-}(\mu),b_\ell^{0}(\mu))
\quad\text{and}\quad
I_\ell^{+}=(b_\ell^{0}(\mu),x_\ell^{+}(\mu))
\ee
and for $\ell=1,\dots,L-1$ the intervals
\be\nn
J_\ell^{-}=(x_\ell^{+}(\mu),w_\ell^{0}(\mu))
\quad\text{and}\quad
J_\ell^{+}=(w_\ell^{0}(\mu),x_{\ell+1}^{-}(\mu))
\ee
Having done this, we define for $\ell=1,\dots,L$ the
functions
\be\nn
F_\ell^\pm=\Big(\restr{A}{\mathsf{clos}(I_\ell^\pm)}\Big)^{-1}
\ee 
and for $\ell=1,\dots,L-1$ the functions
\be\nn
G_\ell^\pm=\Big(\restr{A}{\mathsf{clos}(J_\ell^\pm)}\Big)^{-1}.
\ee

Lastly, for each such barrier, we introduce the function
\be
\label{phi-lambda}
\fisf_\ell(\mu)=
\int_{x_\ell^{-}(\mu)}^{x_\ell^{+}(\mu)}\sqrt{A(t)^2-\mu^2}dt.
\ee
It is easy to check that $\fisf_\ell$ is $C^1$. Moreover, 
differentiating (\ref{phi-lambda}) and using the relations 
$A(x_\ell^\pm)=\mu$, we obtain
\be
\label{derivative-phi-lambda}
\frac{d\fisf_\ell}{d\mu}(\mu)=
-2\mu
\int_{x_\ell^{-}(\mu)}^{x_\ell^{+}(\mu)}\big[A(t)^2-\mu^2\big]^{-1/2}dt
<0.
\ee
Thus, $\fisf_\ell$ is a one-to-one mapping.

Let us now pass to the case of double zeros. In such a case, we hit local
minima and/or local maxima. Without any loss of generality and for clarity 
and simplicity of notation, we shall deal with the case of a potential 
function with two humps presented in Figure \ref{different-cases-two-humps}. 
In this situation we have a potential $A$ that attains a single local 
minimum $m_1$ and two local maxima $M_1<M_2$, the biggest of which is the 
total maximum. Let us examine in detail the two (critical) situations of 
hitting either a local minimum or a local maximum.
\begin{itemize}
\item
\underline{Hitting a local minimum}\\
When $0<\mu<m_1$ (cf. subfigure \ref{fig:subfig1}) we have only one finite 
barrier $(x_1^-(\mu),x_1^+(\mu))$. When $\mu$ grows to reach $m_1$, equation 
$A(x)=m_1$ has now three zeros; two simple at $x_1^\pm(m_1)$ and one double at  
$x_1^0(m_1)$(see subfigure \ref{fig:subfig2}). Observe that in such a case, 
there emerges a new point $x_1^0$ between $x_1^-<x_1^+$ that previously 
(i.e. when $\mu<m_1$) defined the barrier. 
\item
\underline{Hitting a local maximum}\\
When $M_1<\mu<M_2$ (cf. subfigure \ref{fig:subfig5}) we see that we again have 
only one finite barrier $(x_1^-(\mu),x_1^+(\mu))$. When $\mu$ grows to reach $M_2$, 
equation $A(x)=M_2$ has now only one zero; a double one at $x_1$. 
(see subfigure \ref{fig:subfig6}). Observe that in such a case, the two points 
that previously (i.e. when $M_1<\mu<M_2$) defined a barrier coalesce to a single 
point $x_1$. The same behavior is observed in subfigures \ref{fig:subfig3} and 
\ref{fig:subfig4} when $\mu=M_1$. In this latter case we are left with 
a double zero $x_1(M_1)$ and a finite barrier having as endpoints the simple zeros 
$x_2^-(M_1)$ and $x_2^+(M_1)$.
\end{itemize}

The general case follows exactly by arguing along the same lines of the 
observations just made. In short, when we hit a local minimum, a new point 
is being created in a barrier, while when we hit a local maximum, a barrier 
is supressed to a point.

\begin{figure}[H]
\centering
\subfloat[Subfigure 1 list of figures text][$0<\mu<m_1$]{
\begin{tikzpicture}[scale=0.8]
\begin{axis}[
    legend pos = north west,
    axis lines = none,
    xlabel = {},
    ylabel = {},
    axis line style={draw=none},
    tick style={draw=none}]
\addplot [
    domain=-14:14, 
    samples=500, 
    color=red,]
{1/(1+(x+5)^2)+2/(1+(x-5)^2)};
\addlegendentry{$A(x)$} 
\end{axis}
\draw[scale=0.5,domain=0.6:1.25,dashed,variable=\y,blue]  
plot ({2.7},{\y});
\draw[scale=0.5,domain=0.6:1.25,dashed,variable=\y,blue]  
plot ({11.5},{\y});
\node[circle,inner sep=1.5pt,fill=black] at (1.35,0.55) {};
\node[circle,inner sep=1.5pt,fill=black] at (5.75,0.55) {};
\draw [line width=1mm] (1.34,0.3) -- (5.75,0.3);
\draw[dashed] (0.5,0.55) -- (6.5,0.55);
\draw (0,0.3) -- (7,0.3);
\draw (1.3,0) node {$x_1^-$};
\draw (5.9,0) node {$x_1^+$};
\draw (6.5,1) node {$\mu$};
\end{tikzpicture}
\label{fig:subfig1}}
\subfloat[Subfigure 2 list of figures text][$\mu=m_1$]{
\begin{tikzpicture}[scale=0.8]
\begin{axis}[
    legend pos = north west,
    axis lines = none,
    xlabel = {},
    ylabel = {},
    axis line style={draw=none},
    tick style={draw=none}]
\addplot [
    domain=-14:14, 
    samples=500, 
    color=red,]
{1/(1+(x+5)^2)+2/(1+(x-5)^2)};
\addlegendentry{$A(x)$} 
\end{axis}
\draw[scale=0.5,domain=0.6:1.25,dashed,variable=\y,blue]  
plot ({3.6},{\y});
\draw[scale=0.5,domain=0.6:1.25,dashed,variable=\y,blue]  
plot ({6.7},{\y});
\draw[scale=0.5,domain=0.6:1.25,dashed,variable=\y,blue]  
plot ({10.6},{\y});
\node[circle,inner sep=1.5pt,fill=black] at (1.8,0.7) {};
\node[circle,inner sep=1.5pt,fill=black] at (3.35,0.7) {};
\node[circle,inner sep=1.5pt,fill=black] at (5.3,0.7) {};
\draw [line width=1mm] (1.8,0.3) -- (5.3,0.3);
\draw[dashed] (0.5,0.7) -- (6.5,0.7);
\draw (0,0.3) -- (7,0.3);
\draw (1.5,0) node {$x_1^-$};
\draw (3.2,-0.1) node {$x_1^0$};
\draw (5.7,0) node {$x_1^+$};
\draw (6.5,1) node {$\mu$};
\end{tikzpicture}
\label{fig:subfig2}}
\qquad
\subfloat[Subfigure 3 list of figures text][$m_1<\mu<M_1$]{
\begin{tikzpicture}[scale=0.8]
\begin{axis}[
    legend pos = north west,
    axis lines = none,
    xlabel = {},
    ylabel = {},
    axis line style={draw=none},
    tick style={draw=none}]
\addplot [
    domain=-14:14, 
    samples=500, 
    color=red,]
{1/(1+(x+5)^2)+2/(1+(x-5)^2)};
\addlegendentry{$A(x)$} 
\end{axis}
\draw[scale=0.5,domain=0.6:3.4,dashed,variable=\y,blue]  
plot ({4.4},{\y});
\draw[scale=0.5,domain=0.6:3.4,dashed,variable=\y,blue]  
plot ({5.15},{\y});
\draw[scale=0.5,domain=0.6:3.4,dashed,variable=\y,blue]  
plot ({8.2},{\y});
\draw[scale=0.5,domain=0.6:3.4,dashed,variable=\y,blue]  
plot ({9.6},{\y});
\node[circle,inner sep=1.5pt,fill=black] at (2.2,1.7) {};
\node[circle,inner sep=1.5pt,fill=black] at (2.6,1.7) {};
\node[circle,inner sep=1.5pt,fill=black] at (4.1,1.7) {};
\node[circle,inner sep=1.5pt,fill=black] at (4.8,1.7) {};
\draw [line width=1mm] (2.18,0.3) -- (2.6,0.3);
\draw [line width=1mm] (4.1,0.3) -- (4.82,0.3);
\draw[dashed] (0.5,1.7) -- (6.5,1.7);
\draw (0,0.3) -- (7,0.3);
\draw (2,0) node {$x_1^-$};
\draw (2.7,0) node {$x_1^+$};
\draw (4,0) node {$x_2^-$};
\draw (5,0) node {$x_2^+$};
\draw (6.5,2) node {$\mu$};
\end{tikzpicture}
\label{fig:subfig3}}
\subfloat[Subfigure 4 list of figures text][$\mu=M_1$]{
\begin{tikzpicture}[scale=0.8]
\begin{axis}[
    legend pos = north west,
    axis lines = none,
    xlabel = {},
    ylabel = {},
    axis line style={draw=none},
    tick style={draw=none}]
\addplot [
    domain=-12:12, 
    samples=500, 
    color=red,]
{1/(1+(x+5)^2)+2/(1+(x-5)^2)};
\addlegendentry{$A(x)$}
\end{axis}
\draw[scale=0.5,domain=0.6:5.7,dashed,variable=\y,blue]  
plot ({4.47},{\y});
\draw[scale=0.5,domain=0.6:5.7,dashed,variable=\y,blue]  
plot ({8.8},{\y});
\draw[scale=0.5,domain=0.6:5.7,dashed,variable=\y,blue]  
plot ({9.7},{\y});
\node[circle,inner sep=1.5pt,fill=black] at (2.225,2.86) {};
\node[circle,inner sep=1.5pt,fill=black] at (4.38,2.86) {};
\node[circle,inner sep=1.5pt,fill=black] at (4.85,2.86) {};
\node[circle,inner sep=1.5pt,fill=black] at (2.23,0.3) {};
\draw [line width=1mm] (4.39,0.3) -- (4.86,0.3);
\draw[dashed] (0.8,2.85) -- (6,2.85);
\draw (0,0.3) -- (7,0.3);
\draw (2.25,0) node {$x_1$};
\draw (4.37,0) node {$x_2^-$};
\draw (4.88,0) node {$x_2^+$};
\draw (6.5,2.85) node {$\mu$};
\end{tikzpicture}
\label{fig:subfig4}}
\qquad
\subfloat[Subfigure 5 list of figures text][$M_1<\mu<M_2$]{
\begin{tikzpicture}[scale=0.8]
\begin{axis}[
    legend pos = north west,
    axis lines = none,
    xlabel = {},
    ylabel = {},
    axis line style={draw=none},
    tick style={draw=none}]
\addplot [
    domain=-12:12, 
    samples=500, 
    color=red,]
{1/(1+(x+5)^2)+2/(1+(x-5)^2)};
\addlegendentry{$A(x)$} 
\end{axis}
\draw[scale=0.5,domain=0.6:6.5,dashed,variable=\y,blue]  
plot ({8.9},{\y});
\draw[scale=0.5,domain=0.6:6.5,dashed,variable=\y,blue]  
plot ({9.6},{\y});
\node[circle,inner sep=1.5pt,fill=black] at (4.42,3.3) {};
\node[circle,inner sep=1.5pt,fill=black] at (4.82,3.3) {};
\draw [line width=1mm] (4.42,0.3) -- (4.82,0.3);
\draw[dashed] (0.8,3.3) -- (6,3.3);
\draw (0,0.3) -- (7,0.3);
\draw (4.2,0) node {$x_1^-$};
\draw (5.1,0) node {$x_1^+$};
\draw (6.5,2.85) node {$\mu$};
\end{tikzpicture}
\label{fig:subfig5}}
\subfloat[Subfigure 6 list of figures text][$\mu=M_2$]{
\begin{tikzpicture}[scale=0.8]
\begin{axis}[
    legend pos = north west,
    axis lines = none,
    xlabel = {},
    ylabel = {},
    axis line style={draw=none},
    tick style={draw=none}]
\addplot [
    domain=-12:12, 
    samples=500, 
    color=red,]
{1/(1+(x+5)^2)+2/(1+(x-5)^2)};
\addlegendentry{$A(x)$} 
\end{axis}
\draw[scale=0.5,domain=0.6:10.5,dashed,variable=\y,blue]  
plot ({9.25},{\y});
\node[circle,inner sep=1.5pt,fill=black] at (4.62,5.2) {};
\node[circle,inner sep=1.5pt,fill=black] at (4.62,0.3) {};
\draw[dashed] (2.8,5.2) -- (6,5.2);
\draw (0,0.3) -- (7,0.3);
\draw (4.5,0) node {$x_1$};
\draw (6.5,5) node {$\mu$};
\end{tikzpicture}
\label{fig:subfig6}}
\caption{The case of a potential function $A$ with two humps. In each subfigure, 
the ``energy" (spectral parameter) $\mu$ takes on different values in 
$\mathcal{R}_A$.}
\label{different-cases-two-humps}
\end{figure}
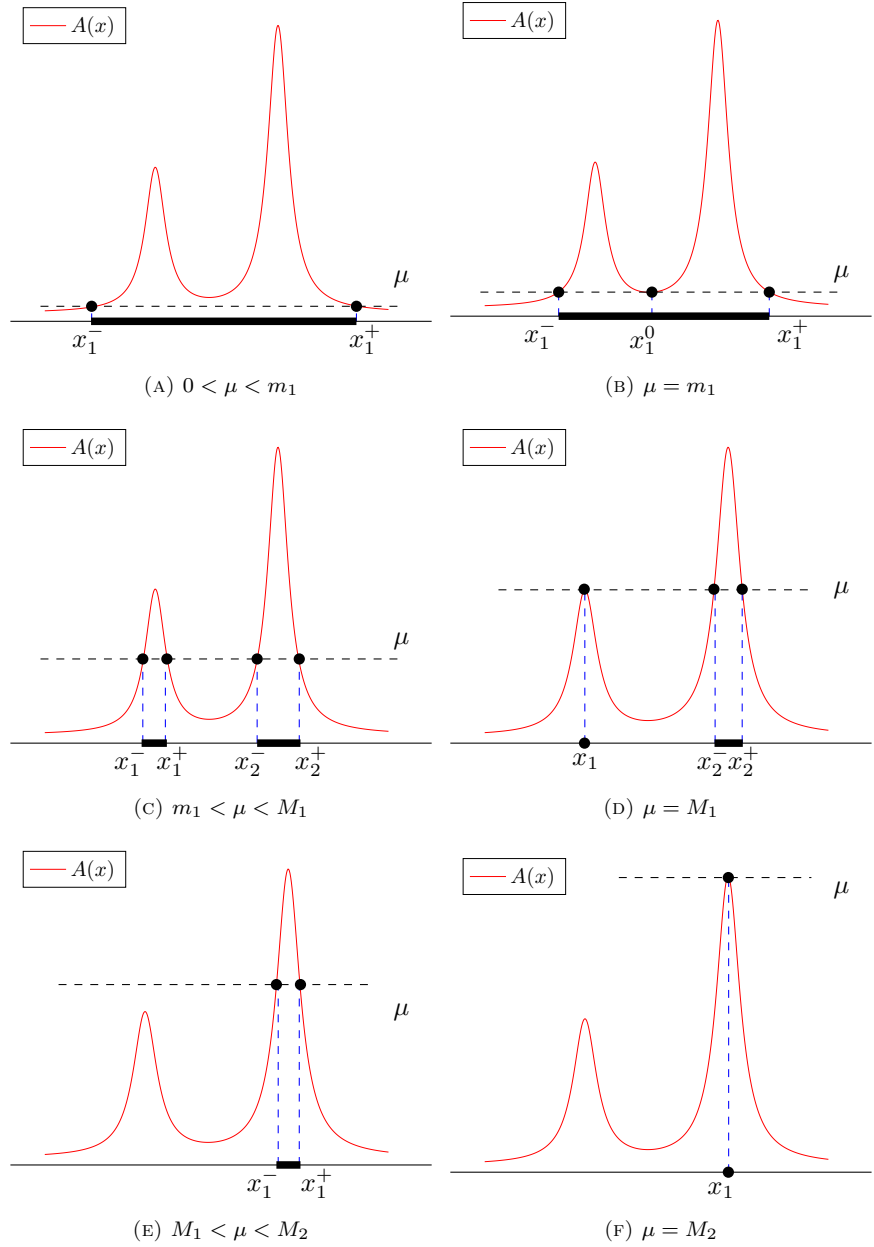

\subsection{Statement of the problem}
\label{problem-statement}

We study the problem
\be\label{ev-problem}
\mathfrak{D}_{\hbar}[\mathbf{u}]=\lambda\mathbf{u}
\ee
where $\mathfrak{D}_{\hbar}$ is the following Dirac (or Zakharov-Shabat) 
operator
\be\label{dirac}
\mathfrak{D}_{\hbar}=
\begin{bmatrix}
i\hbar\partial_{x} & -iA \\
-iA & -i\hbar\partial_{x}
\end{bmatrix}
\ee
with $\hb$ a positive parameter, $A$ a function satisfying 
Assumption \ref{multi-hump-potential-assumption} and
$
\mathbf{u}=[u_{1}\hspace{2pt}u_{2}]^T 
$
a function from $\mathbb{R}$ to $\mathbb{C}^2$. As usual,
$\lambda\in\C$ plays the role of the spectral parameter. 

To be more precise, we treat $\mathfrak{D}_{\hbar}$ as a densely 
defined operator on $L^2(\R;\C^2)$ and want to investigate the EVs 
of problem (\ref{ev-problem}) as $\hb\downarrow0$. So, first we 
explain what we mean when we talk about eigenvalues of this equation.

\begin{definition}
\label{lambda-ev}
For a fixed $\hb>0$, we say that $\lambda\in\mathbb{C}$ is an 
\textbf{eigenvalue (EV)} of the operator $\mathfrak{D}_{\hbar}$ in (\ref{dirac}),
if equation (\ref{ev-problem}) -with this value of $\lambda$- 
has a non-trivial solution 
$\mathbf{u}=[u_{1}\hspace{2pt}u_{2}]^T
\in L^{2}(\mathbb{R};\mathbb{C}^{2})$; that is
\begin{equation*}
0<\int_{-\infty}^{+\infty}\Big[|u_{1}(x)|^{2}+|u_{2}(x)|^{2}\Big]dx
<+\infty.
\end{equation*}
\end{definition}

In general, a non-self-adjoint operator like $\mathfrak{D}_{\hbar}$ 
has complex EVs. For such an operator (with a potential $A$ 
satisfying Assumption \ref{multi-hump-potential-assumption}), 
we know the following about its spectrum 
(see article \cite{klaus+shaw2003} by Klaus and Shaw and 
\cite{h+w} by Hirota and Wittsten).
\begin{itemize}
\item
If $\mathfrak{D}_{\hbar}$ has EVs, then there is a purely 
imaginary EV whose imaginary part is strictly larger than 
the imaginary part of any other EV.
\item
The EV formation threshold is
$$\hb^{-1}\|A\|_{L^1(\R)}>\frac{\pi}{2}$$
and is hence always achieved for sufficiently small $\hb$.
\item
Let $N$ be the largest nonnegative integer such that
$$
\hb^{-1}\|A\|_{L^1(\R)}>(2N-1)\frac{\pi}{2}.
$$
Then there are at least $N$ purely imaginary EVs.
\item
The spectrum of $\mathfrak{D}_{\hbar}$ is symmetric with respect 
to reflection in $\R$.
\item
The continuous (essential) spectrum consists of the entire real 
line $\R$, i.e.
$$\sigma_{ess}(\mathfrak{D}_{\hbar})=\R.$$
\item
Apart from the origin, there are no real EVs.
\end{itemize}  
%Finally, it is not necessary to consider EVs away from 
%$\R\cup i[-A_{max},A_{max}]$ because of the following result 
%due to Dencker \cite{dencker08}. 
%\begin{itemize}
%\item
%The spectrum of $\mathfrak{D}_{\hbar}$, 
%namely $\sigma(\mathfrak{D}_{\hbar})$, accumulates on 
%$\R\cup i[-A_{max},A_{max}]$ as $\hb\downarrow0$. 
%In fact, if $\Sigma\subseteq\complement(\R\cup i[-A_{max},A_{max}])$, 
%then for sufficiently small $\hb$, the operator $\mathfrak{D}_{\hbar}$ 
%has no spectrum in the set $\Sigma$.
%\end{itemize}
We procceed by supposing the following for our Dirac operator.
\begin{hypothesis}
\label{main-hypo}
There exists a positive number $\hb_0$, such that for every $0<\hb<\hb_0$, 
$\mathfrak{D}_{\hbar}$ has only purely imaginary EVs. 
Equivalently stated, for the point spectrum of $\mathfrak{D}_{\hbar}$
we suppose that
$$
\exists\hspace{2pt}\hb_0>0
\quad\text{such that}\quad
\forall\hspace{2pt}0<\hb<\hb_0,\quad
\sigma_{p}(\mathfrak{D}_{\hbar})\subset i[-A_{max},A_{max}].$$
\end{hypothesis}
We conjecture that the forementioned hypothesis is always true for 
all potentials $A$ satisfying Assumption 
\ref{multi-hump-potential-assumption}. It has been proved in the case of
two lobes \cite{h+w} and it looks probable that the argument can be extended 
to the multi-hump case.  Hence,  from now on we always assume that $0<\hb<\hb_0$. 

\subsection{Trasforming spectral parameter \& changing variables}

From now on, we assume that the Hypothesis \ref{main-hypo} is satisfied. 
Recall that the spectrum of $\mathfrak{D}_{\hbar}$ is symmetric with 
respect to reflection in $\R$. Hence, we start by changing the spectral 
parameter $\lambda\in i(0,A_{max}]$ to $\mu\in\mathbb{R}_+$ by setting
\be
\label{lambda-to-mu}
\lambda=i\mu.
\ee
Hence, 
(\ref{ev-problem}) is written as
\be\label{z+sh}
\hbar
\begin{bmatrix}
u_{1}'(x,\mu,\hb)\\
u_{2}'(x,\mu,\hb)
\end{bmatrix}=
\begin{bmatrix}
\mu & A(x) \\
-A(x) & -\mu
\end{bmatrix}
\begin{bmatrix}
u_{1}(x,\mu,\hb)\\
u_{2}(x,\mu,\hb)
\end{bmatrix},
\quad
x\in\mathbb{R}.
\ee 

Under the change of variables (cf. equation (4) in \cite{pdmiller2001})
\be\label{change-of-var-dirac-to-schrodi}
y_{\pm}=
\frac{u_{2}\pm u_{1}}{\sqrt{A\mp\mu}}
\ee
system (\ref{z+sh}) is equivalent to the following two independent 
eigenvalue equations
\be\label{schrodi-initial-version}
y_{\pm}''(x,\mu,\hb)=
\bigg\{
\hbar^{-2}[\mu^2-A^2(x)]+
\frac{3}{4}\Big[\frac{A'(x)}{A(x)\mp\mu}\Big]^2-
\frac{1}{2}\frac{A''(x)}{A(x)\mp\mu}
\bigg\}
y_{\pm}(x,\mu,\hb),
\quad
x\in\mathbb{R}.
\ee

Since $A(x)\in\mathbb{R}_+$, $x\in\mathbb{R}$, we will only consider 
the “minus” case for the lower index in (\ref{schrodi-initial-version})
and thus work with the equation
\be\label{schrodi-intermediate-version}
\frac{d^2y}{dx^2}=
\bigg\{\hbar^{-2}[\mu^2-A^2(x)]+
\frac{3}{4}\Big[\frac{A'(x)}{A(x)+\mu}\Big]^2-
\frac{1}{2}\frac{A''(x)}{A(x)+\mu}\bigg\}y,
\quad
x\in\mathbb{R}.
\ee

Observe that the change of variables 
(\ref{change-of-var-dirac-to-schrodi}) does not alter the discrete 
spectrum. Hence we are led to the following important fact.
\begin{proposition}
\label{mu-spectrum}
Under Assumption \ref{multi-hump-potential-assumption} and 
Hypothesis \ref{main-hypo}, finding the discrete spectrum of 
$\mathfrak{D}_{\hbar}$ in (\ref{dirac}), is equivalent to finding 
the values $\mu\in(0,A_{max}]$ for which (\ref{schrodi-intermediate-version}) 
has an $L^{2}(\mathbb{R};\C)$ solution.
\end{proposition} 

\subsection{Reformulating the equation}
\label{reformulation}

We introduce another symbol for the spectral parameter in order 
to rely on the results from sections 
\S\ref{passage-barrier} and \S\ref{passage-well}. Recall that we 
started with $\lambda$ and then changed to $\mu$ using 
(\ref{lambda-to-mu}). 

In a neighborhood of a finite (noncritical) barrier 
$\mathfrak{B}_\ell=(x_\ell^-,x_\ell^+)$, $\ell=1,\dots,L$, equation 
(\ref{schrodi-intermediate-version}) can be written as
\be\nn
\frac{d^2y}{dx^2}=[\hbar^{-2}f(x,x_\ell^+)+g(x,x_\ell^+)]y
\ee
where $f$ and $g$ satisfy
\be\nn
f(x,x_\ell^+)=A^2(x_\ell^+)-A^2(x)
\ee
and
\be\nn
g(x,x_\ell^+)=
\frac{3}{4}\bigg[\frac{A'(x)}{A(x)+A(x_\ell^+)}\bigg]^2-
\frac{1}{2}\frac{A''(x)}{A(x)+A(x_\ell^+)}.
\ee
This simply says that in a neighborhood of a barrier, we
can use the results obtained in \S\ref{passage-barrier}.

Similarly, in a neighborhood of a finite (noncritical) well 
$\mathfrak{W}_\ell=(x_\ell^+,x_{\ell+1}^-)$, $\ell=1,\dots,L-1$, equation 
(\ref{schrodi-intermediate-version}) can be put in the form
\be\nn
\frac{d^2y}{dx^2}=[\hbar^{-2}f(x,x_\ell^+)+g(x,x_\ell^+)]y
\ee
where $f$ and $g$ satisfy
\be\nn
f(x,x_\ell^+)=A^2(x_\ell^+)-A^2(x)
\ee
and
\be\nn
g(x,x_\ell^+)=
\frac{3}{4}\bigg[\frac{A'(x)}{A(x)+A(x_\ell^+)}\bigg]^2-
\frac{1}{2}\frac{A''(x)}{A(x)+A(x_\ell^+)}.
\ee
But as before, this quarantees us that in a neighborhood of
a well, all the results of \S\ref{passage-well} can be used
freely.

From paragraphs \S\ref{passage-barrier} and \S\ref{passage-well} we know 
that after applying the Liouville transform, the above differential 
equations are transformed correspondingly to anothers of the form 
\be
\label{equation-equiv}
\frac{d^2X}{d\z^2}=
\big[\pm\hb^{-2}(\z^2-\gamma^2)+\varphi(\z,\gamma)\big]X
\ee
where for the \say{$+$} sign, $\gamma=\alpha$ and $\varphi=\ps$ 
(cf. \S \ref{liouville-transform-barrier})
while for the \say{$-$} case, $\gamma=\beta$ and 
$\varphi=\overline{\ps}$ (see \S \ref{liouville-transform-well}). 
Hence having in mind Proposition \ref{mu-spectrum} we are led to 
the following.
\begin{proposition}
\label{alpha-spectrum}
Under Assumption \ref{multi-hump-potential-assumption} and Hypothesis 
\ref{main-hypo}, finding the discrete spectrum of $\mathfrak{D}_{\hbar}$ in 
(\ref{dirac}) is equivalent to referring to equation (\ref{equation-equiv}) 
for the `` $+$" sign with $\gamma=\alpha$, $\varphi=\ps$ and finding the 
values $\alpha\geq0$ for which it possesses an $L^{2}(\mathbb{R};\C)$ solution.
\end{proposition} 

\section{Semiclassical Spectral Results for Multiple Barriers}
\label{multi-hump-quanta}

In this section, we use the results from paragraphs \S\ref{passage-barrier} 
and \S\ref{passage-well} to study the EVs and their corresponding norming 
constants of a Dirac operator with potential $A$. Here, we let this potential 
have multiple humps (see Figure \ref{pic-3-barrier}). To be precise, 
we assume the following.

\subsection{Quantization conditions for the EVs}
\label{quanta-spectrum-condo}

In this subsection, using Assumption \ref{multi-hump-potential-assumption},  
Hypothesis \ref{main-hypo} and what we have gathered so far, we present the 
results for the EVs of the Dirac operator $\mathfrak{D}_\hb$.
\begin{theorem}
\label{theorem-b-r-multi}
Consider a potential $A$ of the Dirac operator $\mathfrak{D}_\hb$ 
in (\ref{dirac}) that satisfies Assumption \ref{multi-hump-potential-assumption}. 
Also assume Hypothesis \ref{main-hypo} and take $0<\mu_1<\mu_2\leq A_{max}$. 
Suppose that $\lambda=i\mu\in i[\mu_1,\mu_2]$ 
(where $\lambda=\lambda(\hb)$ and $\mu=\mu(\hb)$) is an EV of 
$\mathfrak{D}_\hb$. Then using the notation from \S\ref{notation}, 
at least for one $\ell=\ell(\hb)\in\{1,2,\dots,L\}$, there is a 
non-negative integer $n=n(\mu,\ell,\hb)$ such that
\be\label{B-R-multi}
\fisf_\ell(\mu)=
\pi\Big(n+\tfrac{1}{2}\Big)\hb+\asympt(\hb^{\frac{5}{3}})
\quad
\text{as}
\quad
\hb\downarrow0.
\ee
\end{theorem}
\begin{proof}
From Theorem \ref{fix-condo-at-least-on-one-end} we see that each 
well $\mathfrak{W}_\ell(\mu)$, $\ell=1,\dots,L-1$ yields at least one 
fixing condition (cf.  (\ref{fix-cond-def})).  Moreover, the asymptotic form of 
$Y_+(\z,\al(\mu),\hb)$ as $\z\to+\infty$ and the asymptotics for $Y_-(\z,\al(\mu),\hb)$ 
and $Z_-(\z,\al(\mu),\hb)$ as $\z\to-\infty$ (see (\ref{y1-approx}), 
(\ref{e1-asympt}) and (\ref{y_minus-z_minus-asympt})) show that in the 
presense of an EV, the coefficient $\s_{12}$ in equation (\ref{sigma1}) 
has to be zero. But this is translated to the fact that the fixing conditions 
are fulfilled at the right point $x_1^-(\mu)$ of the well $\mathfrak{W}_0(\mu)$ 
and at the left point $x_L^+(\mu)$ of the well $\mathfrak{W}_L(\mu)$. 
Thus for every $0<\hb<\hb_0$, we have at least $L+1$ fixing conditions for 
$L$ barriers. Hence, there exists a barrier $\mathfrak{B}_\ell(\mu,\hb)$ 
(depending on $\hb$ as well) for which a fixing condition is 
satisfied on each of its two ends. Finally, we refer to Theorem 
\ref{B-R-existence} which gives us the desired results.
\end{proof}

Earlier, from formula (\ref{derivative-phi-lambda}) we saw that 
$\fisf_\ell$ is a one-to-one mapping. Whence there exists the inverse 
$\fisf_\ell^{-1}$. Using this remark, we can write (\ref{B-R-multi}) 
equivalently as
\be\label{B-R-multi-equiv}
\mu(\hb)=\fisf_\ell^{-1}\bigg[\pi\Big(n+\tfrac{1}{2}\Big)\hb\bigg]+
\asympt(\hb^{\frac{5}{3}})
\quad
\text{as}
\quad
\hb\downarrow0.
\ee 
This formula \ref{B-R-multi-equiv} leads to the following definition about 
\textit{WKB EVs}. In a sense, a WKB EV approximates one actual EV of our 
operator.
\begin{definition}
\label{definition-wkb-evs}
If $\lambda(\hb)=i\mu(\hb)$ is an EV of $\mathfrak{D}_\hb$, then from 
Theorem \ref{theorem-b-r-multi} there exists at least one 
$\ell\in\{1,2,\dots,L\}$ so that formula (\ref{B-R-multi-equiv}) 
is true for some $n\in\N$. Then the number
\be\label{wkb-ev}
\lambda_{\ell,n}^{WKB}(\hb)=i\m_{\ell,n}^{WKB}(\hb)=
i\fisf_\ell^{-1}\bigg[\pi\Big(n+\tfrac{1}{2}\Big)\hb\bigg]
\ee
shall be defined to be a \textbf{WKB eigenvalue} related to the actual EV 
$\lambda(\hb)=i\mu(\hb)$. 
\end{definition}

Consider now the intervals
\be
\label{ev-vals}
\Delta_{\ell,n}(\hb)=
\bigg(
\m_{\ell,n}^{WKB}(\hb)-c\hb^{\frac{5}{3}},
\m_{\ell,n}^{WKB}(\hb)+C\hb^{\frac{5}{3}}
\bigg)
\ee
for some arbitrary $\hb$-independent constants $c,C>0$.
The lengths of these intervals are of order $\asympt(\hb^{\frac{5}{3}})$ 
while for different $m,n\in\N$ the distances between the points 
$\m_{\ell,m}^{WKB}(\hb)$ and $\m_{\ell,n}^{WKB}(\hb)$ are 
of order $\asympt(\hb)$. This says that for sufficiently small $\hb$
\be\nn
\Delta_{\ell,m}(\hb)\cap\Delta_{\ell,n}(\hb)=\emptyset
\quad\text{for}\quad
m\neq n.
\ee
However if we consider different $k,\ell\in\{1,2,\dots,L\}$, 
\text{it may occur} that
\be\nn
\Delta_{k,n}(\hb)\cap\Delta_{\ell,n}(\hb)\neq\emptyset
\quad\text{for}\quad
k\neq\ell.
\ee

Theorem (\ref{theorem-b-r-multi}) says that each EV of $\mathfrak{D}_\hb$ 
belongs to one of the intervals $i\Delta_{\ell,n}(\hb)$. Hence for 
sufficiently small $\hb$, equivalently stated this can be written as
\be\nn
\sigma_p(\mathfrak{D}_\hb)\cap
i(\mu_1,\mu_2)\subset
\bigcup_{\ell=1}^L\bigcup_n
i
\Delta_{\ell,n}(\hb)
\ee
or
\be\nn
\text{dist}
\bigg\{
\sigma_p(\mathfrak{D}_\hb)\cap i(\mu_1,\mu_2),
\bigcup_{\ell=1}^L\bigcup_{n\in\N}
\{\lambda_{\ell,n}^{WKB}(\hb)|n\in\N\}
\bigg\}
=
\asympt(\hb^{\frac{5}{3}})
\quad
\text{as}
\quad
\hb\downarrow0.
\ee

Now, we shall be concerned with the converse; namely the existence 
of an actual EV for our operator in (\ref{dirac}). We have the following
theorem.
\begin{theorem}
\label{exist-ev-theorem}
Let Assumption \ref{multi-hump-potential-assumption} be satisfied 
by $A$ and assume Hypothesis \ref{main-hypo} for $\mathfrak{D}_\hb$. 
Then for every $\ell\in\{1,2,\dots,L\}$ and every non-negative 
integer $n$ such that 
\be
\label{assu-theo}
\fisf_\ell^{-1}
\Big[\pi\Big(n+\tfrac{1}{2}\Big)\hb\Big]
\in
(\mu_1,A_{max})
\ee
there exists an EV of $\mathfrak{D}_\hb$,  namely $\lam=i\mu$ 
(where $\lam=\lam(\ell,n,\hb)$ and $\mu=\mu(\ell,n,\hb)$),  that satisfies
\be\nn
\lam=\lambda_{\ell,n}^{WKB}(\hb)+\asympt(\hb^{\frac{5}{3}})
\quad
\text{as}
\quad
\hb\downarrow0.
\ee 
\end{theorem}
\begin{proof}
Fix some $\ell\in\{1,2,\dots,L\}$ and some non-negative integer $n$ so that 
(\ref{assu-theo}) is true. By Theorem \ref{exist-ev-barrier}, there exists 
$\mu=\mu(\ell,n,\hb)$ obeying 
\be
\label{con-proof-1}
\mu=
\fisf_\ell^{-1}\bigg[\pi\Big(n+\tfrac{1}{2}\Big)\hb\bigg]+
\asympt(\hb^{\frac{5}{3}})
\quad
\text{as}
\quad
\hb\downarrow0
\ee
and such that 
\be\nn
Y_-(\z,\al_\ell(\mu),\hb)=
\sigma(\al_\ell(\mu),\hb)
Y_+(\z,\al_\ell(\mu),\hb)
\ee
where 
\begin{align*}
\alpha_\ell(\mu) &=
\sqrt{\frac{2}{\pi}\fisf_\ell(\mu)}\quad\text{and}\\
\sigma(\al_\ell(\mu),\hb) &=
(-1)^n+\asympt(\hb^\frac{2}{3})
\quad
\text{as}
\quad
\hb\downarrow0.
\end{align*}
Fix this value of $\mu$.

Let a cut-off function $\chi_\ell\in C_0^\infty(\R)$ be such that 
$\chi_\ell(\z)=1$ in some neighborhood of the interval 
$\z(\mathfrak{B}_\ell(\mu))$ (recall that $\zeta$ is a continuous and increasing 
function of $x$) and $\chi_\ell(\z)=0$ outside of some larger neighborhood of 
that interval. Particularly, then we have $\chi_\ell(\z)=0$ on all other 
intervals $\z(\mathfrak{B}_k(\mu))$ with $k\neq\ell$. We set 
\be\nn
f_{\ell,n}(\z,\hb)=
Y_-(\z,\al_\ell(\mu),\hb)\chi_\ell(\z).
\ee
Observe that $f_{\ell,n}(\cdot,\hb)\in C_0^2(\R)$. Since the function 
$Y_-(\z,\al_\ell(\mu),\hb)$ satisfies 
\be\nn
\frac{d^2Y_-}{d\z^2}=
\big[\hb^{-2}(\z^2-\alpha_\ell^2(\mu))+
\ps(\z,\alpha_\ell(\mu))\big]Y_-
\ee 
for $f_{\ell,n}$ we have
\be\nn
\begin{split}
\frac{d^2f_{\ell,n}}{d\z^2}-
\big[\hb^{-2}(\z^2-\alpha_\ell^2(\widetilde{\mu}))+
\ps(\z,\alpha_\ell(\mu))\big]f_{\ell,n}=\\
2\frac{dY_-}{d\z}\frac{d\chi_\ell}{d\z}+
Y_-\frac{d^2\chi_\ell}{d\z^2}=\\
\sigma(\al_\ell(\mu),\hb)
\bigg(
2\frac{dY_+}{d\z}\frac{d\chi_\ell}{d\z}+
Y_+\frac{d^2\chi_\ell}{d\z^2}
\bigg).
\end{split}
\ee

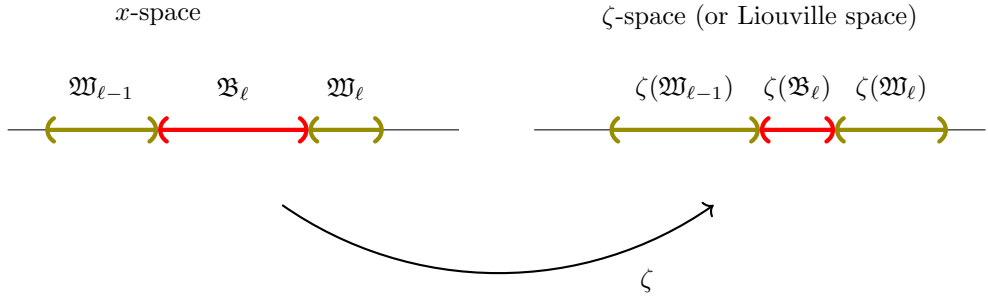
\begin{figure}[H]
\centering
\begin{tikzpicture}

%\draw (6.5,3) node {$\m$-space};
%\draw (4,2) -- (9.5,2);
%\draw [(-), thick, black] (5,2) -- (8,2)
%node[anchor=south west] {$\m_2$};
%\draw [(-), thick, black] (8,2) -- (5,2)
%node[anchor=north east] {$\m_1$};
%
%\node[circle,inner sep=2pt,fill=black,label=below:
%{$\m$}] at (6,2) {};

\draw (0,0) -- (6,0);
\draw [(-), ultra thick, olive] (0.5,0) -- (2,0)
node[anchor=center] {};
\draw [(-), ultra thick, red] (2,0) -- (4,0)
node[anchor=south west] {};
\draw [(-), ultra thick, olive] (4,0) -- (5,0)
node[anchor=south west] {};
%\draw [(-), ultra thick, olive] (5.5,0) -- (0.5,0)
%node[anchor=north east] {$a_2$};

%\node[circle,inner sep=2pt,fill=black,label=below:
%{$a$}] at (4.5,0) {};

%\fill[fill=blue, opacity=0.2](4,0.2)--(4.75,0.2)--(4.75,-0.2)--
%(4,-0.2)--(4,0.2);

\draw (2,1.5) node {$x$-space};

\draw (10,1.5) node {$\zeta$-space (or Liouville space)};

\draw (1.25,0.55) node {$\mathfrak{W}_{\ell-1}$};
\draw (3,0.55) node {$\mathfrak{B}_{\ell}$};
\draw (4.5,0.55) node {$\mathfrak{W}_{\ell}$};

\draw (9,0.55) node {$\z(\mathfrak{W}_{\ell-1})$};
\draw (10.5,0.55) node {$\z(\mathfrak{B}_{\ell})$};
\draw (11.75,0.55) node {$\z(\mathfrak{W}_{\ell})$};

\draw (7,0) -- (13,0);
\draw [(-), ultra thick, olive] (8,0) -- (10,0)
node[anchor=south west] {};
\draw [(-), ultra thick, red] (10,0) -- (11,0)
node[anchor=north east] {};
\draw [(-), ultra thick, olive] (11,0) -- (12.5,0)
node[anchor=north east] {};
%\node[circle,inner sep=2pt,fill=black,label=below:
%{$\pi(n+\tfrac{1}{2})\hbar$}] at (11,0) {};

\draw[thick, ->] (3.65,-1) arc (235:305:5);

\draw (8.5,-2) node {$\zeta$};
\end{tikzpicture}
\caption{Barriers and wells in $x$-space and Liouville space.}
\label{bar-in-liouville}
\end{figure}

Due to the derivatives of $\chi_\ell$, the expression above differs 
from zero only on compact subsets of the intervals 
$\z(\mathfrak{W}_{\ell-1}(\mu))$ and $\z(\mathfrak{W}_\ell(\mu))$ 
(see Figure \ref{bar-in-liouville}). But the definitions of $Y_\pm$ 
and $\sigma$ along with (\ref{U-asympt}) show that this expression 
tends to zero as $\hb\downarrow0$ on both 
$\z(\mathfrak{W}_{\ell-1}(\mu))$ and 
$\z(\mathfrak{W}_\ell(\mu))$. This shows (cf. Proposition \ref{alpha-spectrum}) 
that $\lambda=i\mu$ is the desired EV. Finally, using the definition 
(\ref{wkb-ev}) and (\ref{con-proof-1}) we find that $\lam$ satisfies the 
specified asymptotics as $\hb\downarrow0$.
\end{proof}

\begin{remark}
We cannot exclude the possibility that EVs coming from different barriers (i.e. different 
values of $\ell$) get too close (closer than $\asympt(\hb^{5/3})$) or even coincide.  
So we could even have double EVs.  But this does not affect the applications to NLS.  
For the semiclassical analysis of the inverse scattering, the important fact is that we 
have different sets of EVs from different barriers,  each set with a different density.
\end{remark}

\subsection{Norming constants}

A straightforward application of Theorem \ref{theorem-b-r-multi} above, 
allows us to express the norming constants of the Dirac operator 
$\mathfrak{D}_{\hbar}$. In particular we see that the asymptotics  
obtained, agree with the the ones for a bell-shaped even potential 
(see Chapter 3 of \cite{kmm} or Corrollary 10.5 in \cite{h+k}). 
We have the following corollary.
\begin{corollary}
\label{norm-const}
Consider the Dirac operator $\mathfrak{D}_{\hbar}$ as in (\ref{dirac}) satisfying 
Assumption \ref{multi-hump-potential-assumption} and Hypothesis \ref{main-hypo}.
Suppose that $\lam(\hb)$ is an EV of this operator. Then there is a 
non-negative integer $n$ (depending both on $\lam$ and $\hb$) such that the 
corresponding norming constant has asymptotics 
\be\nn
(-1)^n + \asympt(\hb^\frac{2}{3})
\quad\text{as}\quad
\hbar\downarrow0.
\ee
\end{corollary}
\begin{proof}
Since $\lam(\hb)=i\mu(\hb)$ is an EV, by Theorem \ref{theorem-b-r-multi} 
we arrive at (\ref{B-R-multi-equiv}) for some $\ell\in\{1,\dots,L\}$ and $n\in\N$, 
i.e.
\be\nn
\mu(\hb)=\fisf_\ell^{-1}\bigg[\pi\Big(n+\tfrac{1}{2}\Big)\hb\bigg]+
\asympt(\hb^{\frac{5}{3}})
\quad
\text{as}
\quad
\hb\downarrow0.
\ee
But since $\lam(\hb)=i\mu(\hb)$ is an EV, from (\ref{sigma1}) and 
(\ref{barrier-connection}) we obtain
\be\nn
Y_-(\z,\al_\ell(\mu(\hb)),\hb)=
\sigma_{11}(\al_\ell(\mu(\hb)),\hb)
Y_+(\z,\al_\ell(\mu(\hb)),\hb)
\ee
where 
\begin{align*}
\alpha_\ell(\mu(\hb)) &=
\sqrt{\frac{2}{\pi}\fisf_\ell(\mu(\hb))}\quad\text{and}\\
\sigma_{11}(\al_\ell(\mu(\hb)),\hb) &=
\sin\Big(\tfrac{1}{2}\pi\hb^{-1}\al_\ell(\mu(\hb))^2\Big)+
\asympt(\hb^{\frac{2}{3}})
\quad
\text{as}
\quad
\hb\downarrow0.
\end{align*}
But from the asymptotics for $\mu{\hb}$ above (or equivalently 
(\ref{B-R-multi-equiv})) we find
\be\nn
\sigma_{11}(\al_\ell(\mu(\hb)),\hb)=
(-1)^n+\asympt(\hb^\frac{2}{3})
\quad
\text{as}
\quad
\hb\downarrow0.
\ee
And this completes the proof.
\end{proof}

\subsection{Eigenvalues near zero}
\label{near-zero-evs}

For the applications to the semiclassical theory of the focusing NLS equation,  
it is important to understand the behavior of the EVs near 0.  The whole investigation 
is essentially the same to the one we exploited in \cite{h+k}.  We repeat once more the 
steps here since we wish to generalize the results we presented in that work of ours.
 
We begin with a potential function $A$ that satisfies 
Assumption \ref{multi-hump-potential-assumption}.  Such a function has finitely many 
local minima,  say $N\in\N_0$, accounting for the case of a function having none ($N=0$); 
if there are $N\in\N$ local minima,  we denote them by $m_j$,  $j=1,\dots, N$.  
We set $\tilde{m}$ to be 
\be
\label{mu-tilde}
\tilde{m}=
\begin{cases}
A_{max},\quad\text{if}\quad N=0\\
\min_{j\in\{1,\dots,N\}}m_j,\quad\text{if}\quad N\in\N.
\end{cases}
\ee
So,  in this paragraph,  we would like to investigate the 
(semiclassical) behavior of the EVs of $\mathfrak{D}_\hb$ (with potential $A$) 
that lie in $i(0,\tilde{m})$.   

We start by considering $\mu(\hb)\in(0,\tilde{m})$ so that $\m(\hb)\downarrow0$ 
as $\hb\downarrow0$.  
We emphasize that we are in the presence of only one (finite) barrier 
(see Figure \ref{pic-near-zero-barrier}).  In this setting,  using (\ref{f-schrodi-barrier}),  
(\ref{g-schrodi-barrier}),  (\ref{p-barrier}) and having in mind that $A(b_\pm)=\mu$ 
(for the notation,  consult \S\ref{liouville-transform-barrier}),  we define
\be
\label{f-h}
\bar{f}(x,\hb)=f\big(x,b_+(\mu(\hb))\big)=\mu(\hb)^2-A^2(x)
\ee
\be
\label{g-h}
\bar{g}(x,\hb)=g\big(x,b_+(\mu(\hb))\big)=
\frac{3}{4}\Big[\frac{A'(x)}{A(x)+\mu(\hb)}\Big]^2-
\frac{1}{2}\frac{A''(x)}{A(x)+\mu(\hb)}
\ee 
and
\be\nn
\bar{f}(x,\hb)=
\big(x-b_-(\mu(\hb))\big)\big(x-b_+(\mu(\hb))\big)\bar{p}(x,\hb).
\ee
where
\be\nn
\bar{p}(x,\hb)=p\big(x,b_+(\mu(\hb))\big).
\ee

\begin{figure}[H]
\centering
\begin{tikzpicture}[scale=1]
\begin{axis}[
    legend pos = north west,
    axis lines = none,
    xlabel = {},
    ylabel = {},
    axis line style={draw=none},
    tick style={draw=none}]
\addplot [
    domain=-14:14, 
    samples=500, 
    color=red,]
{1/(1+(x+5)^2)+2/(1+(x-5)^2)};
\addlegendentry{$A(x)$} 
\end{axis}
\draw[scale=0.5,domain=0.6:1.25,dashed,variable=\y,blue]  
plot ({2.7},{\y});
\draw[scale=0.5,domain=0.6:1.25,dashed,variable=\y,blue]  
plot ({11.5},{\y});
\node[circle,inner sep=1.5pt,fill=black] at (1.35,0.55) {};
\node[circle,inner sep=1.5pt,fill=black] at (5.75,0.55) {};
\draw [line width=1mm] (1.34,0.3) -- (5.75,0.3);
\draw[dashed] (0.5,0.55) -- (6.5,0.55);
\draw (0,0.3) -- (7,0.3);
\draw (1.3,0) node {$b_-$};
\draw (5.9,0) node {$b_+$};
\draw (6.5,1) node {$0<\mu(\hb)<\tilde{m}$};
\end{tikzpicture}
\caption{The potential barrier in a case of near zero EVs.}
\label{pic-near-zero-barrier}
\end{figure}
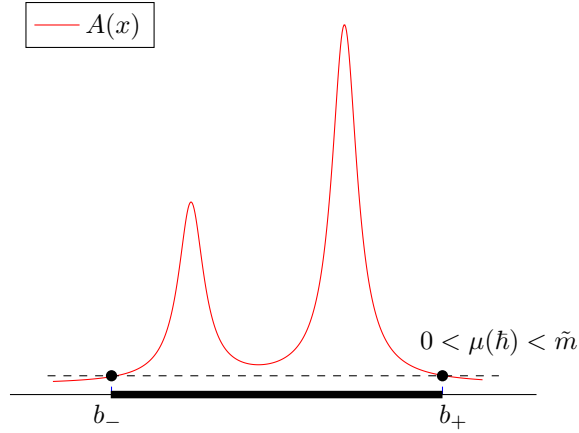

We apply the Liouville transform once again (as in \S \ref{liouville-transform-barrier}),  and 
arrive at the following proposition (cf.  Proposition \ref{propo-schrodi-to-liouvi}).  
\begin{proposition}
For every $\hb>0$,  equation
\be\nn
\frac{d^2y}{dx^2}=[\hbar^{-2}\bar{f}(x,\hb)+\bar{g}(x,\hb)]y,
\quad x\in\R
\ee
is transformed to equation
\be\nn
\frac{d^2X}{d\z^2}=
\big[\hb^{-2}(\z^2-\alpha(\mu(\hb))^2)+\ps(\z,\alpha(\mu(\hb)))\big]X,
\quad 
\z\in\R
\ee
in which $\z$ is given by the Liouville transform (\ref{barrier-case}), 
$\al$ is given by (\ref{alpha}) and the function $\ps(\z,\alpha(\mu(\hb)))$ is given 
by the formula
\begin{multline*}
\ps(\z,\al(\mu(\hb)))=
\frac{1}{4}\frac{3\z^2 +2\al(\mu(\hb))^2}{[\z^2 -\al(\mu(\hb))^2]^2}
+\frac{1}{16}\frac{\z^2 -\al(\mu(\hb))^2}{\bar{f}(x,\hb)^3}\\
\cdot\big[4\bar{f}(x,\hb)\bar{f}''(x,\hb)-5\bar{f}'(x,\hb)^2\big]+
[\z^2 -\al(\mu(\hb))^2]\frac{\bar{g}(x,\hb)}{\bar{f}(x,\hb)}
\end{multline*}
where prime denotes differentiation with respect to $x$.  
\end{proposition}

\begin{remark}
By recalling the definition of $\al$ in (\ref{alpha}),  and the fact that 
$b_\pm(\mu(\hb))\rightarrow\pm\infty$,  as $\hb\downarrow0$,  we obtain
\be
\label{alpha-h-limit}
\al(\mu(\hb))
\uparrow\sqrt{\tfrac{2}{\pi}\|A\|_{L^1(\R)}}
\quad\text{as}\quad\hb\downarrow0.
\ee
 \end{remark}
 
It is easy to see that for each value of $\hb$,  the functions $\bar{f}$, 
$\bar{g}$ and $\bar{p}$ satisfy properties $(i)$ through $(iv)$ in the proof of 
Lemma \ref{lemma-error-cont-barrier} in \S\ref{error-cont-barrier}.  
This in turn implies -again with the help of Lemma I in \cite{olver1975}- that for each 
$\hb$ the function
\begin{multline}
\label{psi-h}
\ps(\z,\al(\mu(\hb)))=
\frac{1}{4}\frac{3\z^2 +2\al(\mu(\hb))^2}{[\z^2 -\al(\mu(\hb))^2]^2}
+\frac{1}{16}\frac{\z^2 -\al(\mu(\hb))^2}{\bar{f}(x,\hb)^3}\\
\cdot\big[4\bar{f}(x,\hb)\bar{f}''(x,\hb)-5\bar{f}'(x,\hb)^2\big]+
[\z^2 -\al(\mu(\hb))^2]\frac{\bar{g}(x,\hb)}{\bar{f}(x,\hb)}
\end{multline}
is continuous in the corresponding region of the $(\z,\al)$-plane. 

So in order to have a conclusion such as Theorem \ref{main-thm-barrier} 
and eventually results like Theorem \ref{B-R-existence} and Theorem \ref{exist-ev-barrier}, 
we need to investigate the convergence of the integral in (\ref{variation-total}) 
(cf.  proof of Theorem \ref{main-thm-barrier} or proof of Theorem $6.3$ in \S $6$ of 
\cite{h+k}),  i.e.  
\be
\label{hbar-total-var}
\int_{0}^{+\infty}\frac{|\ps(t,\al(\mu(\hb)))|}{\Om(t\sqrt{2\hb^{-1}})}dt.
\ee

Here we need to place an additional assumption on the behavior of the potential 
$A$ at $\pm\infty$.  
\begin{assumption}
\label{near-0-evs-potential-assume}
Suppose  there are real positive numbers $1<r^+ \leq s^+$,  so that
$$
\frac{C_1^+(x)}{ |x|^{s^+}} \leq A(x) \leq\frac{C_2^+(x)}{ |x|^{r^+}}\quad\text{for}\quad x>0
$$
where $C_1^+, C_2^+$ are bounded functions and $2 r^+ - s^+ > \frac{1}{3}$; and 
there are real positive numbers $1<r^- \leq s^-$,  so that
$$
\frac{C_1^-(x)}{ |x|^{s^-}} \leq A(x) \leq\frac{C_2^-(x)}{ |x|^{r^-}}\quad\text{for}\quad x<0
$$
where $C_1^-, C_2^-$ are bounded functions and $2 r^- - s^- > \frac{1}{3}$.  
Alternatively,  suppose there are real positive numbers $0<r\leq s$ so that
$$
C_1(x) e^{-|x|^s} \leq A(x) \leq C_2(x) e^{-|x|^r},\quad x\in\R
$$
where $C_1, C_2$ are bounded functions.
\end{assumption}
%\begin{assumption}
%For the function $A$,  there are numbers $r\in\R$,  $s\geq0$ and an integer $n\geq0$, 
%so that 
%\be\nn
%A(x)|x|^r e^{|x|^s}(\log|x|)^{-n}=\asympt(1)
%\quad\text{as}\quad
%x\to\pm\infty
%\ee
%where of course,  it is understood that
%\begin{itemize}
%\item
%if $s=n=0$, then $r>1$
%\item
%if $n=0$ and $s>0$,  then $r\in\R$
%\item
%if $s=0$ and the integer $n\geq1$,  then $r>1$
%\item
%if $s>0$ and the integer $n\geq1$,  then $r\in\R$.
%\end{itemize}
%\end{assumption} 

Finally,  recall (\ref{zeta-barrier-right}) where now $x_2=+\infty$.  It shows that 
$x\uparrow+\infty$ as $\z\uparrow+\infty$.  The lemma below deals with the asymptotic 
behavior of $x$ as $\z\uparrow+\infty$.  It shall be used to allow us understand the 
nature of $\ps$ for \say{big} $\z$. 
\begin{lemma}
\label{x-at-big-zeta}
Considering $x$ as a function of $\z$ we see that
\be
\label{x-asymptotics}
x=\frac{\z^2}{2\mu}\Big[1+\asympt\Big(\tfrac{\log\z}{\z^2}\Big)\Big]
\quad
\text{as}
\hspace{7pt}
\z\uparrow+\infty
\ee
uniformly with respect to $\mu=A(b_\pm)\in(0,\tilde{m})$.
\end{lemma}
\begin{proof}
See Lemma 5.2 in \S 5 of \cite{h+k}.
\end{proof}

It is now straightforward to check that  Olver's theory is uniformly applicable
all the way to $\mu=0$. 
For example, consider first the case where $A$ is rational:
\be\label{a-rational}
A(x)=\frac{1}{|x|^r}\quad\text{for}\quad|x|\geq1
\ee
where $r>1$ (clearly satisfying 
Assumption \ref{finite-barrier-left-right-infinite-wells-assumption} and 
Assumption \ref{near-0-evs-potential-assume}).  In this case, using (\ref{x-asymptotics}) 
we get 
\be
\label{x-vs-z-rational}
x=\frac{\z^2}{2\mu(\hb)}\Big[1+\asympt\Big(\tfrac{\log\z}{\z^2}\Big)\Big]
\quad
\text{as}
\hspace{7pt}
\z\uparrow+\infty
\ee
while using (\ref{f-h}), (\ref{g-h}), (\ref{psi-h}), (\ref{a-rational}) and 
(\ref{x-vs-z-rational}) we arrive at
\be
\ps(\z,\al(\mu(\hb)))=
\ps_{1}(\z,\al(\mu(\hb)))\Big[1+\asympt\Big(\tfrac{\log\z}{\z^2}\Big)\Big]
\quad
\text{as}
\hspace{7pt}
\z\uparrow+\infty
\ee
uniformly in $\al$ and consequently in $\hb$, where
\begin{align*}
\ps_{1}(\z,\al(\mu(\hb)))&=
\frac{1}{4}\frac{3\z^2 +2\al(\mu(\hb))^2}{[\z^2 -\al(\mu(\hb))^2]^2}\\
& -
r(2r+1)2^{2r+1}\mu(\hb)^{2r-2}\z^{4r-4}[\z^2 -\al(\mu(\hb))^2]\cdot\\
&
\hspace{5cm}
\frac{\z^{4r}+\frac{r-2}{2r+1}2^{2r-1}\mu(\hb)^{2r-2}}
{[\z^{4r}-2^{2r}\mu(\hb)^{2r-2}]^3}\\
& -
r(r+1)2^{r+1}\mu(\hb)^{r-1}\z^{4r-4}[\z^2 -\al(\mu(\hb))^2]\cdot\\
&
\hspace{4cm}
\frac{\z^{2r}-\frac{r-2}{r+1}2^{r-1}\mu(\hb)^{r-1}}
{[\z^{2r}-2^r\mu(\hb)^{r-1}][\z^{2r}+2^r\mu(\hb)^{r-1}]^3}.
\end{align*}

Consider now the case where $A$ is exponentially decreasing:
\be\label{a-exponentional}
A(x)=e^{-|x|^r}\quad\text{for}\quad|x|\geq1
\ee
where $r>0$; it clearly satisfies 
Assumption \ref{finite-barrier-left-right-infinite-wells-assumption} and 
Assumption \ref{near-0-evs-potential-assume}).  Using (\ref{f-h}),  (\ref{g-h}),  
(\ref{psi-h}),  (\ref{x-asymptotics}),  (\ref{a-rational}) and (\ref{x-vs-z-rational}) 
we arrive at
\be
\ps(\z,\al(\mu(\hb)))=
\ps_{2}(\z,\al(\mu(\hb)))\Big[1+\asympt\Big(\tfrac{\log\z}{\z^2}\Big)\Big]
\quad
\text{as}
\hspace{7pt}
\z\uparrow+\infty
\ee
uniformly in $\al$ and consequently in $\hb$, where
\begin{align*}
\ps_{2}(\z,\al(\mu(\hb)))&=
\frac{1}{4}\frac{3\z^2 +2\al(\mu(\hb))^2}{[\z^2 -\al(\mu(\hb))^2]^2}\\
& +
\frac{r}{2^r}\frac{\z^2 -\al(\mu(\hb))^2}{\mu(\hb)^{r-2}}
\frac{\z^{2r-4}\exp\big\{-\frac{\z^{2r}}{2^{r-1}\mu(\hb)^r}\big\}}
{\Big[\exp\big\{-\frac{\z^{2r}}{2^{r-1}\mu(\hb)^r}\big\}-\mu(\hb)^2\Big]^3}\cdot\\
&
\hspace{1cm}
\bigg[\frac{r}{2^r}\frac{\z^{2r}}{\mu(\hb)^r}\exp\Big\{-\frac{\z^{2r}}{2^{r-1}\mu(\hb)^r}\Big\}
+2(r-1)\exp\Big\{-\frac{\z^{2r}}{2^{r-1}\mu(\hb)^r}\Big\}\\
& 
\hspace{5cm}
+
\frac{r}{2^{r-2}}\frac{\z^{2r}}{\mu(\hb)^{r-2}}-2(r-1)\mu(\hb)^2\bigg]\\
& -
\frac{r}{2^{r-1}}\frac{\z^2 -\al(\mu(\hb))^2}{\mu(\hb)^{r-2}}
\frac{\z^{2r-4}\exp\big\{-\frac{\z^{2r}}{2^r\mu(\hb)^r}\big\}}
{\Big[\exp\big\{-\frac{\z^{2r}}{2^r\mu(\hb)^r}\big\}-\mu(\hb)\Big]
\Big[\exp\big\{-\frac{\z^{2r}}{2^r\mu(\hb)^r}\big\}+\mu(\hb)\Big]^3}\cdot\\
&
\hspace{1cm}
\bigg[\frac{r}{2^{r+1}}\frac{\z^{2r}}{\mu(\hb)^r}\exp\Big\{-\frac{\z^{2r}}{2^r\mu(\hb)^r}\Big\}
+(r-1)\exp\Big\{-\frac{\z^{2r}}{2^r\mu(\hb)^r}\Big\}\\
& 
\hspace{6cm}
-
\frac{r}{2^r}\frac{\z^{2r}}{\mu(\hb)^{r-1}}+(r-1)\mu(\hb)\bigg].
\end{align*}

These asymptotics imply that for each $\hb>0$, the integral in 
(\ref{hbar-total-var}) converges; furthermore, this convergence is 
uniform in $\alpha$.  Now similar computations 
can be easily performed for any $A$ satisfying 
Assumption \ref{finite-barrier-left-right-infinite-wells-assumption} and 
Assumption \ref{near-0-evs-potential-assume}.  The result still remains the same.  
The integral in (\ref{hbar-total-var}) converges uniformly in $\alpha$.  
A variation of Theorem \ref{main-thm-barrier} can be applied to guarantee the 
existence of approximate solutions in these cases too. 
Hence,  we arrive at the following theorem.  
\begin{theorem}
For every $\hb>0$,  equation 
\be
\label{h-dependent-equation}
\frac{d^2Y}{d\z^2}=
\big[\hb^{-2}\big(\z^2-\alpha(\mu(\hb))^2\big)+\ps(\z,\al(\mu(\hb)))\big]Y
\ee
has in the region $[0,+\infty)\times[0,\al(\mu(\hb))]$ of the 
$(\z,\al)$-plane solutions $Y_+$ and $Z_+$ 
which are continuous, 
have continuous first and second partial $\z$-derivatives, and are 
given by 
\bea\nn
Y_+(\z,\al(\mu(\hb)),\hb)=
U(\z\sqrt{2\hb^{-1}},-\tfrac{1}{2}\hb^{-1}\al(\mu(\hb))^2)+
\eps(\z,\al(\mu(\hb)),\hb)\\
\nn
Z_+(\z,\al(\mu(\hb)),\hb)=
\ol U(\z\sqrt{2\hb^{-1}},-\tfrac{1}{2}\hb^{-1}\al(\mu(\hb))^2)+
\ol{\eps}(\z,\al(\mu(\hb)),\hb)
\eea
(cf. (\ref{y1-approx}), (\ref{y2-approx})) where for the remainders
we have the relations
\begin{multline}\nn
\frac{|\eps(\z,\al(\mu(\hb)),\hb)|}
{\msf(\z\sqrt{2\hb^{-1}},-\tfrac{1}{2}\hb^{-1}\al(\mu(\hb))^2)},
\frac{\Big|\frac{\partial \eps}
{\partial\z}(\z,\al(\mu(\hb)),\hb)\Big|}{\sqrt{2\hb^{-1}}
\nsf(\z\sqrt{2\hb^{-1}},-\tfrac{1}{2}\hb^{-1}\al(\mu(\hb))^2)}\\
\leq
\frac{1}{\esf(\z\sqrt{2\hb^{-1}},-\tfrac{1}{2}\hb^{-1}\al(\mu(\hb))^2)}
\Big(
\exp\big\{
\tfrac{1}{2}(\pi\hb)^{\frac{1}{2}}l_1(-\tfrac{1}{2}\hb^{-1}\al(\mu(\hb))^2)
\mathcal{V}_{\z,+\infty}[H](\al(\mu(\hb)),\hb)
\big\}
-1
\Big)
\end{multline}
and
\begin{multline}\nn
\frac{|\ol{\eps}(\z,\al(\mu(\hb)),\hb)|}
{\msf(\z\sqrt{2\hb^{-1}},-\tfrac{1}{2}\hb^{-1}\al(\mu(\hb))^2)},
\frac{\Big|\frac{\partial\ol{\eps}}{\partial\z}(\z,\al(\mu(\hb)),\hb)\Big|}
{\sqrt{2\hb^{-1}}
\nsf(\z\sqrt{2\hb^{-1}},-\tfrac{1}{2}\hb^{-1}\al(\mu(\hb))^2)}\\
\leq
\esf(\z\sqrt{2\hb^{-1}},-\tfrac{1}{2}\hb^{-1}\al(\mu(\hb))^2)
\Big(
\exp\big\{
\tfrac{1}{2}(\pi\hb)^{\frac{1}{2}}l_1(-\tfrac{1}{2}\hb^{-1}\al(\mu(\hb))^2)
\mathcal{V}_{0,\z}[H](\al(\mu(\hb)),\hb)
\big\}
-1
\Big)
\end{multline}
(analogous to (\ref{barrier-junk1}), (\ref{barrier-junk2})).
\end{theorem}
\begin{proof}
The proof follows exactly the lines of that for Theorem \ref{main-thm-barrier}.  One has only to 
observe that Theorem \ref{thm-on-exist-int-eq} comes into play and ensures 
that everything remains unchanged. 
\end{proof}

Additionally, $l_1$ and $\var_{0,+\infty}[H]$ satisfy the same 
asymptotics as before (cf. (\ref{lambda-asympt}), 
(\ref{v-operator-asympt-barrier})) and consequently one obtains the same 
asymptotic behavior of solutions as in \S\ref{asympt-behave-barrier-sols}; 
namely
\begin{align*}
\eps(\z,\al(\mu(\hb)),\hb) & =
\frac{\msf(\z\sqrt{2\hb^{-1}},-\tfrac{1}{2}\hb^{-1}\al(\mu(\hb))^2)}
{\esf(\z\sqrt{2\hb^{-1}},-\tfrac{1}{2}\hb^{-1}\al(\mu(\hb))^2)}
\asympt(\hb^{\frac{2}{3}})
\\
\nn
\ol{\eps}(\z,\al(\mu(\hb)),\hb) & =
\esf(\z\sqrt{2\hb^{-1}},-\tfrac{1}{2}\hb^{-1}\al(\mu(\hb))^2)
\msf(\z\sqrt{2\hb^{-1}},-\tfrac{1}{2}\hb^{-1}\al(\mu(\hb))^2)
\asympt(\hb^{\frac{2}{3}})
\\
\nn
\frac{\partial\eps}{\partial\z}(\z,\al(\mu(\hb)),\hb) & =
\frac{\nsf(\z\sqrt{2\hb^{-1}},-\tfrac{1}{2}\hb^{-1}\al(\mu(\hb))^2)}
{\esf(\z\sqrt{2\hb^{-1}},-\tfrac{1}{2}\hb^{-1}\al(\mu(\hb))^2)}
\asympt(\hb^{\frac{1}{6}})
\\
\nn
\frac{\partial\ol{\eps}}{\partial\z}(\z,\al(\mu(\hb)),\hb) & =
\esf(\z\sqrt{2\hb^{-1}},-\tfrac{1}{2}\hb^{-1}\al(\mu(\hb))^2)
\nsf(\z\sqrt{2\hb^{-1}},-\tfrac{1}{2}\hb^{-1}\al(\mu(\hb))^2)
\asympt(\hb^{\frac{1}{6}})
\end{align*}
as $\hb\downarrow0$ uniformly for $\z\geq0$ and 
$\al$.

Arguing as in \S\ref{connection-formulas-barrier}, we obtain two more 
solutions of (\ref{h-dependent-equation}), namely $Y_-$ and 
$Z_-$,  satisfying
\begin{align*}
Y_-(\z,\al(\mu(\hb)),\hb) & =
U(-\z\sqrt{2\hb^{-1}},-\tfrac{1}{2}\hb^{-1}\al(\mu(\hb))^2)+
\frac{\msf(-\z\sqrt{2\hb^{-1}},-\tfrac{1}{2}\hb^{-1}\al(\mu(\hb))^2)}
{\esf(-\z\sqrt{2\hb^{-1}},-\tfrac{1}{2}\hb^{-1}\al(\mu(\hb))^2)}
\asympt(\hb^{\frac{2}{3}})\\
\nn
Z_-(\z,\al(\mu(\hb)),\hb) & =
\ol U(-\z\sqrt{2\hb^{-1}},-\tfrac{1}{2}\hb^{-1}\al(\mu(\hb))^2)+\\
&\hspace{2cm}\esf(-\z\sqrt{2\hb^{-1}},-\tfrac{1}{2}\hb^{-1}\al(\mu(\hb))^2)
\msf(-\z\sqrt{2\hb^{-1}},-\tfrac{1}{2}\hb^{-1}\al(\mu(\hb))^2)
\asympt(\hb^{\frac{2}{3}})
\end{align*}
as $\hb\downarrow0$ uniformly for $\z\leq0$ and $\al$.

Consequently we have the same connection formulae (all the results 
of \S\ref{connection-formulas-barrier} are not altered at all).  Indeed, 
expressing $Y_+$,  $Z_+$ in terms of $Y_-$,  $Z_-$ and 
writing
\begin{align*}
Y_+(\z,\al(\mu(\hb)),\hb) & =
\s_{11}(\al(\mu(\hb)),\hb)Y_-(\z,\al(\mu(\hb)),\hb)+
\s_{12}(\al(\mu(\hb)),\hb)Z_-(\z,\al(\mu(\hb)),\hb)\\
Z_+(\z,\al(\mu(\hb)),\hb) & =
\s_{21}(\al(\mu(\hb)),\hb)Y_-(\z,\al(\mu(\hb)),\hb)+
\s_{22}(\al(\mu(\hb)),\hb)Z_-(\z,\al(\mu(\hb)),\hb)
\end{align*}
(confer (\ref{sigma1}), (\ref{sigma2})) in the same way we find that
\be\nn
\begin{split}
\sigma_{11}(\al(\mu(\hb)),\hb)=
\sin(\tfrac{1}{2}\pi\hb^{-1}\al(\mu(\hb))^2)+\asympt(\hb^{\frac{2}{3}})\\
\sigma_{12}(\al(\mu(\hb)),\hb)=
\cos(\tfrac{1}{2}\pi\hb^{-1}\al(\mu(\hb))^2)+\asympt(\hb^{\frac{2}{3}})\\
\sigma_{21}(\al(\mu(\hb)),\hb)=
\cos(\tfrac{1}{2}\pi\hb^{-1}\al(\mu(\hb))^2)+\asympt(\hb^{\frac{2}{3}})\\
\sigma_{22}(\al(\mu(\hb)),\hb)=
-\sin(\tfrac{1}{2}\pi\hb^{-1}\al(\mu(\hb))^2)+\asympt(\hb^{\frac{2}{3}})
\end{split}
\ee
(like (\ref{barrier-connection})) as $\hb\downarrow0$ uniformly for $\al$.

Eventually, this means that the results of \S\ref{applications-barrier} for 
the EVs remain the same.  But before we state this result, let us remind the reader 
of the function $\fisf$ in (\ref{bold-phi}),  namely
\be
\label{bold-phi-again}
\fisf(\mu)=
\frac{\pi}{2}\al(\mu)^2=
\displaystyle\int_{b_-(\mu)}^{b_+(\mu)}\sqrt{A(x)^2-\mu^2}dx.
\ee
where $A(b_\pm)=\mu$.  We have seen that $\fisf$ is a $C^1$ one-to-one mapping 
satisfying
\be\nn
%\label{derivative-phi-lambda}
\frac{d\fisf}{d\mu}(\mu)=
-2\mu
\int_{b_-(\mu)}^{b_+(\mu)}\big[A(x)^2-\mu^2\big]^{-1/2}dx
<0.
\ee

Now we are ready to state the main result of this section.  Combining 
Theorem \ref{B-R-existence} and Theorem \ref{exist-ev-barrier},  we arrive at 
the following.
\begin{theorem}
\label{boso-near-0}
Let the potential function $A$ satisfy 
Assumption \ref{multi-hump-potential-assumption},   
Assumption \ref{near-0-evs-potential-assume} and set $\tilde{m}$ as in (\ref{mu-tilde}).  
Suppose that $\lam(\hb)=i\mu(\hb)\in i(0,\tilde{m})$ is an EV of the operator 
$\mathfrak{D}_{\hbar}$ (see (\ref{dirac})).  
Then there exists a non-negative integer $n$ for which
\begin{equation}
\label{b-r-quanta-near-0}
\fisf(\mu(\hb))=
\pi\Big(n+\tfrac{1}{2}\Big)\hbar+\mathcal{O}(\hbar^{\frac{5}{3}})
\quad\text{as}\quad\hbar\downarrow0.
\end{equation}
Conversely,  for every non-negative integer $n$ such that 
$\pi(n+\frac{1}{2})\hbar\in\big(\fisf(\tilde{m}),\|A\|_{L^1(\R)}\big)$ 
(recall (\ref{alpha-h-limit}),  (\ref{bold-phi-again})) there exists a unique EV of 
$\mathfrak{D}_{\hbar}$,  namely 
$\lam_{n}(\hbar)=i\mu_{n}(\hbar)$,  so that
\begin{equation}\nn
\bigg|\fisf(\mu_n(\hb))-\pi\Big(n+\frac{1}{2}\Big)\hb\bigg|
\leq 
C\hbar^{\frac{5}{3}}
\end{equation}
with a constant $C$ depending neither on $n$ nor on $\hbar$.
\end{theorem}
\begin{proof}
The proof of this theorem is essentially the same to the proof of Theorem $10.1$ 
in \S $10$ of \cite{h+k}.
\end{proof}

Next,  in the spirit of Definition \ref{definition-wkb-evs},  we have the following 
definition.  
\begin{definition}
\label{definition-wkb-evs-near-0}
If $\lambda(\hb)=i\mu(\hb)$ is an EV of $\mathfrak{D}_\hb$,  then from 
Theorem \ref{boso-near-0} there exists some $n\in\N$ so that formula 
(\ref{b-r-quanta-near-0}) is true.  We call the number
\be\label{wkb-ev-near-0}
\lambda_n^{WKB}(\hb)=i\m_n^{WKB}(\hb)=
i\fisf^{-1}\bigg[\pi\Big(n+\tfrac{1}{2}\Big)\hb\bigg]
\ee
a \textbf{WKB eigenvalue} related to the actual EV $\lambda(\hb)=i\mu(\hb)$. 
\end{definition}

So, we have arrived at the following corollary which explains the behavior of 
the EVs of $\mathfrak{D}_\hb$ that lie near zero.
\begin{corollary}
Consider a function $A$ satisfying Assumption \ref{multi-hump-potential-assumption} 
and Assumption \ref{near-0-evs-potential-assume}.  Also,  set $\tilde{m}$ as in (\ref{mu-tilde}).  
Then for every non-negative integer $n$ such that 
$\pi(n+\frac{1}{2})\hbar\in\big(\fisf(\tilde{m}),\|A\|_{L^1(\R)}\big)$,
there exists a unique EV of $\mathfrak{D}_{\hbar}$,  namely 
$\lam_{n}(\hbar)$ satisfying
\begin{equation}\nn
|\lam_n(\hb)-\lam_n^{WKB}(\hb)|=
\asympt(\hbar^{\frac{5}{3}})
\quad\text{as}\quad\hb\downarrow0
\end{equation}
uniformly for $\lam_n(\hb)$ in $i(0,\tilde{m})$.
\end{corollary}

\section{Reflection Coefficient}
\label{refle}

In this paragraph we consider the behavior of the reflection coefficient
for our Dirac operator (\ref{dirac}). This completes the investigation 
of the set of (semiclassical) scattering data for our operator. 
The results in this section were actually obtained rigorously in \cite{h+k}. 
For completeness sake, we briefly present them here 
as well without proof.  We remind the reader that the continuous spectrum of such a Dirac 
operator with a potential $A$ satisfying the asymptotics of Assumption 
\ref{multi-hump-potential-assumption} at $\pm\infty$, is the whole real line. 

\subsection{Reflection away from zero}
\label{refle-away-0}

Let us begin in this subsection by considering a $\lambda\in\R$ that is 
\textit{idependent} of $\hb$. Under the change of variables
\be\nn
y_{\pm}=
\frac{u_{2}\pm u_{1}}{\sqrt{A\pm i\lambda}}
\ee
equation (\ref{ev-problem}) -with the help of (\ref{dirac})- is 
transformed to the following two independent equations
\be\nn
y_{\pm}''(x,\lam,\hb)=
\bigg\{
\hbar^{-2}[-A^2(x)-\lambda^2]+
\frac{3}{4}\Big[\frac{A'(x)}{A(x)\pm i\lam}\Big]^2-
\frac{1}{2}\frac{A''(x)}{A(x)\pm i\lam}
\bigg\}
y_{\pm}(x,\lam,\hb).
\ee
Again we only consider the lower index and work with the equation
\be
\label{refle-eq}
\frac{d^2y}{dx^2}=[-\hbar^{-2}\ti f(x,\lam)+\ti g(x,\lam)]y
\ee
where $\ti f$ and $\ti g$ satisfy
\be\nn
\ti f(x,\lam)=A^2(x)+\lam^2
\ee
and
\be\nn
\ti g(x,\lam)=
\frac{3}{4}\Big[\frac{A'(x)}{A(x)-i\lam}\Big]^2-
\frac{1}{2}\frac{A''(x)}{A(x)-i\lam}.
\ee

Next we define the \textit{Jost solutions}. Equation 
(\ref{refle-eq}) can be put in the form
\be\nn
-\hb^2\frac{d^2y}{dx^2}+[-A^2(x)+\hb^2\ti g(x,\lam)]y=
\lambda^2y.
\ee
This is the Schr\"odinger equation with a complex potential.
The Jost solutions are defined as the components of the 
bases $\{J_-^l, J_+^l\}$ and $\{J_-^r, J_+^r\}$ of the two-dimensional 
linear space of solutions of equation (\ref{refle-eq}), 
which satisfy the asymptotic conditions
\begin{align*}
J_\pm^l(x,\lam) &\sim \exp\big\{\pm i\frac{\lambda}{\hb}x\big\}
\quad\text{as}\quad x\to-\infty\\
J_\pm^r(x,\lam) &\sim \exp\big\{\pm i\frac{\lambda}{\hb}x\big\}
\quad\text{as}\quad x\to+\infty.
\end{align*}

From scattering theory, we know that the reflection coefficient
$R(\lam,\hb)$ for the waves incident on the potential from 
the right, can be expressed in terms of Wronskians of the Jost 
solutions. More presicely, we have
\be
\label{r}
R(\lam,\hb)= \frac{\W[J_-^l, J_-^r]}{\W[J_+^r, J_-^l]}.
\ee

Examination of the behavior of $R(\lam,\hb)$ can be achieved using 
the same techniques as in \cite{h+k}. More precisely, for $\lam\in\R$ 
with $|\lam|\geq\delta>0$, we have the following theorem (the reader 
seeking more information and proofs, is advised to look at \S 12 of the 
aforementioned work).
\begin{theorem}
Let $A$ satisfy Assumption \ref{multi-hump-potential-assumption}. 
The reflection coefficient of equation (\ref{refle-eq}) as defined 
by (\ref{r}), satisfies
\be\nn
R(\lam,\hb)=
\asympt(\hb)
\quad\text{as}\quad
\hb\downarrow0
\ee
uniformly for $|\lambda|\geq\delta>0$.
\end{theorem}

\subsection{Reflection close to zero}
\label{refle-close-0}

Now we turn to the case where $\lambda\in\R$ depends on $\hb$ 
($\lam=\lam(\hb)$) and particularly we let $\lambda$ approach $0$ 
like $\hb^b$ for an $\hb$-independent positive constant $b$. 
Arguing along the same lines as before, we arrive at the following 
theorem (again, for the proof see \S 12 in \cite{h+k}).

\begin{theorem}
\label{spectrum+scatter}
Let $A$ satisfy Assumption \ref{multi-hump-potential-assumption}. 
Consider $b,s>0$ (independent of $\hb$). Then the reflection coefficient 
of equation (\ref{refle-eq}) as defined by (\ref{r}), satisfies
\be\nn
R(\lam(\hb),\hb)=
\asympt\Big(\hb^{1-sb}\Big)
\quad\text{as}\quad
\hb\downarrow0
\ee
uniformly for $\lambda(\hb)$ in any closed interval of $[\hb^b,+\infty)$.  
\end{theorem}
\begin{remark}
We can ensure that $b$ is as large  as we want by letting $s$ very small 
if we are happy with a weak error estimate $\asympt(\hb^{\epsilon})$ for small 
positive $\epsilon$, as $\hb\downarrow0$. We can at best guarantee 
asymptotics of order $\asympt(\hb^{1-\epsilon})$ for small positive 
$\epsilon$,  if we are allowed to accept a small $b$.
\end{remark}

\section{Inverse Scattering and Semiclassical NLS}
\label{nls}

According to the so-called finite gap ansatz (or more properly hypothesis)
the solution $\psi(x,t)$ of (\ref{ivp-nls}) is asymptotically ($\hb\downarrow0$)
described  (locally) as a slowly
modulated $G+1$ phase wavetrain.  Setting $x=x_0+\hb  \hat{x}$
and $t=t_0+\hb  \hat{t}$,
so that $x_0, t_0$ are ``slow" variables
while $\hat{x}, \hat{t}$ are ``fast" variables,
there exist parameters
\begin{itemize}
\item
$a$
\item
$U = (U_0, U_1, \dots,U_G)^T$
\item
$k =(k_0, k_1, \dots, k_G)^T$
\item
$w =(w_0, w_1, \dots, w_G)^T$
\item
$Y =(Y_0, Y_1, \dots, Y_G)^T$
\item
$Z =( Z_0, Z_1, \dots, Z_G)^T $
\end{itemize}
depending on the slow variables
$x_0$ and $t_0$  (but not on $\hat{x}, \hat{t}$)
such that generically
$\psi(x,t)= \psi(x=x_0+ \hb  \hat{x}, t=t_0+\hb  \hat{t})$ has the following
leading order asymptotics as $\hb\downarrow0$:
\begin{multline}
\label{asymptotics}
\psi(x,t) \sim
a(x_0, t_0) e^{\frac{iU_0(x_0, t_0)}{\hb}}
e^{i\big(k_0(x_0, t_0) \hat{x}-w_0(x_0, t_0) \hat{t}\big)}\\
\cdot
 \frac{\Theta\bigg(  Y(x_0, t_0)+
i  \frac{U(x_0, t_0)}{\hb} +
i\Big(  k(x_0, t_0) \hat{x}-  w(x_0, t_0)\hat{t}\Big)\bigg)}
{
\Theta\bigg(  Z(x_0, t_0)+
i \frac{ U(x_0, t_0)}{\hb} +
i\Big(  k(x_0, t_0)\hat{x}-  w(x_0, t_0)\hat{t}\Big)\bigg)}.
\end{multline}

All parameters can be defined in terms of an underlying Riemann surface $X$ which depends 
solely on $x_0, t_0$.   The moduli of $X$ vary slowly with $x, t$, i.e. they depend on  $x_0, t_0$ 
but not on $\hb, \hat{x}, \hat{t}$.  $\Theta$ is the $G$-dimensional Jacobi theta function 
associated with $X$.  The genus of $X$ can vary with $x_0, t_0$. In fact, the $x,t$-plane is 
divided into open regions in each of which $G$ is constant. On the boundaries of such regions 
(sometimes called ``caustics"; they are unions of analytic arcs),  some degeneracies appear in 
the mathematical analysis (we may have ``pinching" of the surfaces $X$ for example) and 
interesting physical phenomena can appear (like the famous Peregrine rogue wave \cite{bt}).  
The above formulae give asymptotics which are  uniform in compact $(x,t)$-sets not containing 
points on the caustics.

For the exact formulae for the parameters as well as the definition  of the theta functions we refer to 
\cite{kmm} or \cite{kr}. Near the caustics the correct interpretation of (1.4) requires some more work.
For an analysis of the somewhat more delicate behaviour (especially for higher order terms in $\hb$)
near the first caustic see \cite{bt}.

In \cite{kmm} we have been able to prove the finite gap hypothesis under some technical 
assumptions that enabled us to proceed with the semiclassical asymptotic analysis of the 
inverse scattering transform (more precisely the equivalent Riemann-Hilbert formulation).  
Such technical assumptions were justified in \cite{kr}. In both works we assumed  the 
possibility of an analytic extension of a function $\rho$ a priori defined on an imaginary interval,  
that gives the density of eigenvalues of the Dirac operator (accumulating on a compact interval 
on the imaginary axis).  Eventually (see \cite{fujii+kamvi}) it was realized that the analyticity 
assumption could be discarded by use of a simple auxiliary scalar Riemann-Hilbert problem.

However, the above proofs have assumed that the reflection coefficient for the related Dirac 
operator is identically zero and that one can safely replace the actual eigenvalues by their 
WKB-approximants.  Strictly speaking, this assumption is not true. But the results in the previous 
sections  enable us to show that the resulting error is only $o(1)$-small as $\hb\downarrow0$.

In \S\ref{multi-hump-quanta} we have established  a 1-1 correspondence between 
WKB-approximants (coming from different wells and barriers) and actual eigenvalues.  
Furthermore  the WKB-approximants  are uniformly $\asympt(\hb^{5/3})$-close to the actual 
eigenvalues. This is an analogous result to our "single-lobe" result in \cite{h+k}, although 
we should underline the fact that while in the single-lobe case it is $known$ that eigenvalues 
are purely imaginary,  here we state this as a hypothesis, at least for small $\hb$.  In fact, we 
conjecture that the eigenvalues are always imaginary for our general multi-humped 
potentials as long as $\hb$ is small enough.

The crucial quantities  considered in the analysis \cite{kmm} are the ``Blaschke" products
\be\nn
\prod_{n=0}^{N-1}
\frac{\lambda-\lambda_n^*}{\lambda-\lambda_n}
\ee 
where $\lambda_n$ runs over either the actual eigenvalues in the upper half-plane, or 
respectively their WKB approximations $\lambda_n^{WKB}$.  
Here $\lambda$ lies on a union of contours encircling  $[-iA_{max}, iA_{max}]$
and only touching it at the point $0$,  transversally.  It follows easily that if
$|\lam_n(\hb)-\lam_n^{WKB}(\hb)|=\asympt(\hbar^{5/3})$ then 
\be\nn
\frac{\lambda-\lambda_n^{WKB*}}{\lambda-\lambda_n^{WKB}}= 
\frac{\lambda-\lambda_n^*}{\lambda-\lambda_n}
\Big(1+\asympt\big(\tfrac{\hb^{5/3}}{|\lam|}\big)\Big)
\ee
and hence the two corresponding Blaschke products are 
$1+\asympt(\hbar^{2/3}/|\lambda|)$-close (since from \S \ref{multi-hump-quanta} the 
total number of EVs $N$ is of order $\asympt(\hb^{-1})$),  which is good enough if $\lambda$ 
is not too close to zero.  For the somewhat intricate details concerning  what happens near 
zero,  we refer to \cite{kmm}; see also the discussion of the reflection coefficient below.

In the previous section we have also shown that the reflection coefficient can be 
ignored as long as we are at a distance $\hb^b$ from 0, with any $b>1$.  On the 
other hand,  it is worth recalling that the Jost functions and hence the reflection 
coefficient are defined via asymptotics of the form 
$\exp\{i(\lambda x+\lambda^2 t)/\hb\}$ as $x\to\pm\infty$.  
This shows that the Jost functions are bounded uniformly in $\hb$  in the region 
$\frac{\lambda}{\hb} <1$.  Apart from possible poles at $0$ (to be discussed later), 
the same thing holds for the reflection coefficient.

It easily follows from  the so-called ``Schwarz reflection" symmetry conditions 
(Appendix A in \cite{kmm}) that the  relevant ``parametrix"  Riemann-Hilbert  problem 
coming from the non-triviality of the reflection coefficient is solvable and in fact its 
solution is $o(1)$ as $\hb\downarrow0$.

More precisely, for the existence of the  solution of the Riemann-Hilbert factorization 
problem that involves only the reflection coefficient $R$ near $0$  and ignores the 
eigenvalues we have the following  result.
\begin{theorem}
\label{theorem-reflection-rh}Let $b>1$.
Define  a Riemann-Hilbert factorization problem as follows.  Find a $2\times2$ 
matrix ${\bf m}$ so that 
\begin{itemize}
\item
its entries are analytic in ${\mathbb C} \setminus {[-\hb^b, \hb^b]}$
\item
${\bf m}_+(\lambda)=
{\bf m}_-(\lambda){\bf v}(\lambda)$ for $\lambda\in{[-\hb^b, \hb^b]}$
where ${\bf v}$ is the matrix 
\begin{center}
$
{\bf v}(\lambda)=
\begin{bmatrix}
1 & R(\lambda) \exp\{-\tfrac{2i\lambda}{\hb}(x+\lambda t)\}\\
R^* (\lambda) \exp\{\tfrac{2i\lambda}{\hb}(x+\lambda t)\} & 1+|R(\lambda)|^2 
\end{bmatrix}
$
\end{center}
and ${\bf m}_\pm$ denote the limiting values of ${\bf m}$ from above ($+$) 
and below ($-$)
\item
${\bf m}\rightarrow\mathbb{I}$ as $z \to \infty$. 
\end{itemize}
Then this Riemann-Hilbert factorization problem has a unique solution. 
The same holds if the discontinuity contour is taken to be the whole real line. \end{theorem}
\begin{proof}
Follows directly from the Schwarz reflection symmetry of the contour and the jump matrix 
as well as the fact that $\Re( {\bf v} + {\bf v}^*)$ is positive definite; see (A.6) and 
Theorem A.1.2 of \cite{kmm}. Uniqueness follows from the fact that  the determinant 
of {\bf v} is 1.
\end{proof}

The fact that the contribution from the above Riemann-Hilbert problem (with jump contour 
$[-\hb^b, \hb^b]$) is  $o(1)$ as $\hb\downarrow0$ comes from the uniform boundedness of the 
resolvent of the related singular integral operator (because of the uniform boundedness 
of the Jost functtions and the reflection coefficient) and the $\hb$-small size of the contour. 
This is standard  Riemann-Hilbert asymptotic theory, for example see Theorem 7.103 and 
Corollary 7.108  in \cite{d}.  Similarly, we can now extend our result to the Riemann-Hilbert 
factorization problem defined on the whole real line and with the same jump as above. The 
crucial fact is that the jump matrix in ${\mathbb R} \setminus {[-\hb^b, \hb^b]} $ is 
$o(1)$-close to the identity in the uniform sense; again see the proof of Corollary 7.108  
in \cite{d}.

Finally, it is easy to combine  the contributions of the  above Riemann-Hilbert problem on 
the whole line and the ``pure soliton" Riemann-Hilbert problem (determined by setting 
$R=0$ but not disallowing the poles at the eigenvalues) by, say,  taking the product of the 
two separate  Riemann-Hilbert problem solutions. The fact that the solution of that with 
jump on the real line  is $o(1)$-small implies that the solution of the full problem 
(EVs + real spectrum) is $o(1)$-close to the ``soliton ensembles"
Riemann-Hilbert problem.

It $can$  happen (non-generically, for isolated values of $\hb$) that the reflection 
coefficient actually has a pole singularity at $0$. In other words there  may be a 
$spectral~singularity$ at $0$.  In such a case one can amend the analysis by considering 
a  small  circle around $0$ say of radius $\asympt(\hb)$ and removing the singularity 
exactly in the same way we have removed the poles due to the eigenvalues in 
\cite{kmm}. The reflection coefficient of course is not analytically extensible in general but 
one can simply extract the singular part of the reflection coefficient which is of course 
rational.  The main result is not affected.

Having estimated the error of the WKB approximation at the level of the scattering data,  
this error can be built into the Riemann-Hilbert analysis of \cite{kmm} and \cite{kr} as 
another layer of approximation and it does not affect the final finite-gap asymptotics.  
The only remaining change in the inverse scattering analysis for a multi-humped potential 
$A$ is that the density function $\rho$ gets to be somewhat more complicated.

\begin{theorem}
\label{theorem-density}
Consider $A_{max}=\max_{x\in\R}A(x)$. Given $\lam\in[0,iA_{max}]$ 
let 
\be\nn
x^-_1(\lam)\leq x^+_1(\lam)\leq  x^-_2(\lam) \leq x^+_2(\lam)\leq
\dots
\leq x^-_{L}(\lam)\leq x^+_L(\lam),
\quad 
L\in\N
\ee
(for the notation, cf. \S\ref{notation}) be the real solutions of the equation 
$A(x)^2+\lam^2=0$ 
(allowing for the non-generic limiting cases $x^-_l(\lam)=x^+_l(\lam)$
and $ x^+_l(\lam)=x^-_{l+1}(\lam)$ for some $l\in\{1,\dots,L\}$). 
Also, let $\Lambda$ be the set
\be\nn
\Lambda=
\{\,\lambda\in(0,iA_{max}]\mid\,
\lambda\,\text{is an EV of}\hspace{3pt}\mathfrak{D}_\hb\,\}
\ee
and consider the signed measure 
\be\nn
d \mu^{\hb}=\hb\sum_{\lambda\in\Lambda}\big(\delta_{\lambda^*}-\delta_{\lambda}\big)
\ee
where $\delta_x$ denotes the Dirac measure centered at $x$. Then, as 
$\hb\downarrow0$, $d \mu^{\hb}$ converges to a continuous measure in the 
weak-$\ast$ sense. More precisely 
\be\nn
d \mu^{\hb}
\quad
\xlongrightarrow[\text{$\hb\downarrow0$}]{\text{weak-$\ast$}}
\quad
\rho(\lam)\chi_{[0,iA_{max}]}d\lam+
\rho(\lam^*)^*\chi_{[-iA_{max},0]}d\lam
\ee
where the density $\rho(\lam)$ satisfies
\begin{equation}
\rho(\lam)=
\frac{\lam}{\pi}\sum_{l=1}^L\int^{x^+_l(\lam)}_{x^-_l(\lam)}
\frac{dx}{(A(x)^2+\lam^2)^{1/2}}
\end{equation}
\end{theorem}
\begin{proof}
The proof of the theorem follows directly from the results in section 
\S\ref{multi-hump-quanta}.
\end{proof}

It is thus clear that $\rho$ is a continous function on $[-iA_{max}, iA_{max}]$ 
(our discussion in \cite{kmm} shows that it is even piecewise analytic).  
Analyticity of $\rho$ was crucial in the proofs of \cite{kmm} and \cite{kr}.  But
as we have shown in \cite{fujii+kamvi} continuity will suffice; indeed the proofs of 
\cite{kmm} actually become more ``natural" by solving an auxiliary scalar Riemann-Hilbert 
problem with jump across $[-iA_{max}, iA_{max}]$, so continuity is more than enough.

We can finally conclude that, at least under the extra hypothesis that the eigenvalues 
of the Dirac operator are imaginary, and of course the  Assumption 
\ref{multi-hump-potential-assumption} and Assumption \ref{near-0-evs-potential-assume},
the finite gap property is generically valid in the sense described above.

\appendix

\section{Airy Functions}
\label{airy_functions}

In this section, some basic properties of \textit{Airy functions} are 
presented. For further reading one may consult \cite{olver1997}.

\begin{figure}[H]
\includegraphics[scale=0.4]{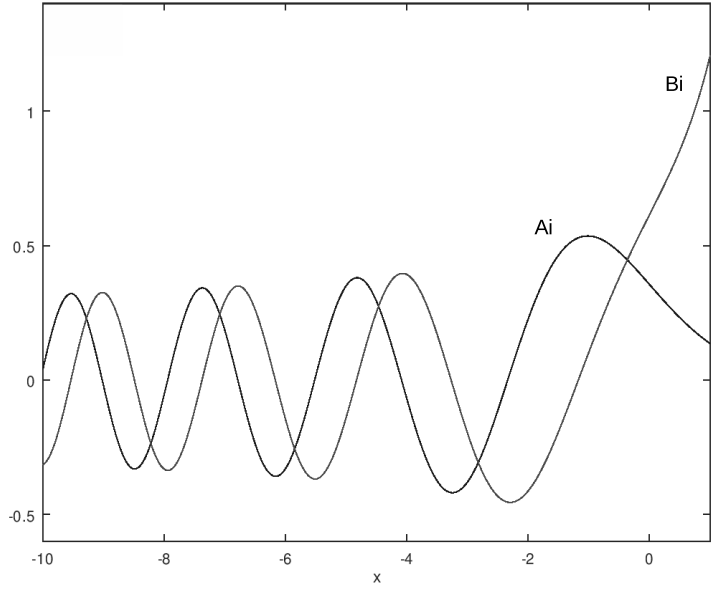}
\caption{The Airy functions $Ai$, $Bi$ on the real line.}
\end{figure}
\noindent
Consider the \textit{Airy equation}
\begin{equation*}
-\dfrac{d^{2}w}{dt^{2}}+tw=0,\quad t\in\mathbb{R}
\end{equation*}
We denote by $Ai$ and $Bi$ its two linearly independent solutions having 
the asymptotics
\begin{equation}\label{airy_1}
Ai(t)=
\frac{1}{2\sqrt{\pi}}
t^{-\frac{1}{4}}\exp\{-\tfrac{2}{3}t^\frac{3}{2}\}
\big[1+O\big(t^{-\frac{3}{2}}\big)\big]
\quad\text{as}\quad t\rightarrow+\infty
\end{equation}
and
\begin{equation}\label{airy_2}
Bi(t)=
-\frac{1}{\sqrt{\pi}}
|t|^{-\frac{1}{4}}\sin\Big(\tfrac{2}{3}|t|^\frac{3}{2}-
\tfrac{\pi}{4}\Big)+O\big(|t|^{-\frac{7}{4}}\big)
\quad\text{as}\quad t\rightarrow-\infty
\end{equation}
\noindent
Their behavior on the opposite side of the real line is known to be
\begin{equation}\label{airy_3}
Ai(t)=
\frac{1}{\sqrt{\pi}}
|t|^{-\frac{1}{4}}\sin\Big(\tfrac{2}{3}|t|^{\frac{3}{2}}+
\tfrac{\pi}{4}\Big)+O\big(|t|^{-\frac{7}{4}}\big)
\quad\text{as}\quad t\rightarrow-\infty
\end{equation}
and
\begin{equation*}
Bi(t)\leq C(1+t)^{-\frac{1}{4}}\exp\{\tfrac{2}{3}t^\frac{3}{2}\},
\quad t\geq0
\end{equation*}
where $C$ is a positive constant.
\noindent
Observe that as $t\to-\infty$, $Ai$ and $Bi$ only differ by a phase shift. 
Also $Ai(t),Bi(t)>0$ for all $t\geq0$. Note that all asymptotic relations 
(\ref{airy_1}), (\ref{airy_2}) and (\ref{airy_3}) can be differentiated 
in $t$; for example
\begin{equation*}
Ai'(t)=
-\frac{1}{\sqrt{\pi}}
|t|^\frac{1}{4}\cos\Big(\tfrac{2}{3}|t|^{\frac{3}{2}}+
\tfrac{\pi}{4}\Big)+
O\big(|t|^{-\frac{5}{4}}\big)
\quad\text{as}\quad t\rightarrow-\infty
\end{equation*}
and
\begin{equation*}
Ai'(t)=
-\frac{1}{2\sqrt{\pi}}
t^\frac{1}{4}\exp\{-\tfrac{2}{3}t^\frac{3}{2}\}
\big[1+O\big(t^{-\frac{3}{2}}\big)\big]
\quad\text{as}\quad t\rightarrow+\infty.
\end{equation*}

Another property says that
\begin{equation}\nn
|Ai(t)|\leq C(1+|t|)^{-\frac{1}{4}},\qquad t\in\mathbb{R}
\end{equation} 
where $C$ is a positive constant. The wronskian of $Ai$, $Bi$ satisfies
\begin{equation*}
\mathcal{W}\big[Ai,Bi\big](t):=Ai(t)Bi'(t)-Ai'(t)Bi(t)=
\frac{1}{\pi},
\quad t\in\mathbb{R}.
\end{equation*}

In order to have a convenient way of assessing the magnitudes of $Ai$ and 
$Bi$ we introduce a \textit{modulus function} $M$, a \textit{phase function} 
$\thv$ and a \textit{weight function} $E$ related by
\begin{equation*}
E(x)Ai(x)=M(x)\sin\thv(x),\quad \frac{1}{E(x)}Bi(x)=M(x)\cos\thv(x),
\quad x\in\mathbb{R}.
\end{equation*}
Actually, we choose $E$ as follows. Denote by $c_{*}$ the biggest 
negative root of the equation $Ai(x)=Bi(x)$ 
(numerical calculations show that $c_{*}=-0.36605$ correct up to five 
decimal places); then define
\begin{equation}\nn
E(x)=
\begin{cases}
1,\quad x\leq c_{*}\\
\sqrt{\frac{Bi(x)}{Ai(x)}},\quad x>c_{*}
\end{cases}
\end{equation}
With this choice in mind, $M$, $\theta$ become
\begin{equation}\nn
M(x)=
\begin{cases}
\sqrt{Ai^{2}(x)+Bi^{2}(x)},\quad x\leq c_{*}\\
\sqrt{2Ai(x)Bi(x)},\quad x>c_{*}
\end{cases}\text{and}\quad
\thv(x)=
\begin{cases}
\arctan\Big[\frac{Ai(x)}{Bi(x)}\Big],\quad x\leq c_{*}\\
\frac{\pi}{4},\quad x>c_{*}
\end{cases}
\end{equation}
where the branch of the inverse tangent is continuous and equal to
$\frac{\pi}{4}$ at $x=c_{*}$. For these functions the asymptotics for 
large $|x|$ read
\begin{align*}
E(x)\sim & 
\begin{cases}
1,\quad x\rightarrow-\infty\\
\sqrt{2}\exp\{\tfrac{2}{3}t^\frac{3}{2}\},
\quad x\rightarrow+\infty
\end{cases}\\
M(x)\sim & \frac{1}{\sqrt{\pi}}
|x|^{-\frac{1}{4}},
\quad |x|\rightarrow+\infty\\
\thv(x)= & 
\begin{cases}
\frac{2}{3}
|x|^\frac{3}{2}+
\frac{\pi}{4}+
\mathcal{O}\big(\frac{3}{2}|x|^{-\frac{3}{2}}\big),
\quad x\rightarrow-\infty\\
\frac{\pi}{4},\quad x\rightarrow+\infty
\end{cases}
\end{align*}

\section{Parabolic Cylinder Functions \& Modified Parabolic Cylinder Functions}
\label{parabolic-cylinder-functions}

The results of the main theorems about existence of approximate 
solutions of the differential equations treated in the main text 
involve PCFs and modified PCFs (cf. \cite{abramo+stegun}).  
So in this section we state a few properties which will be in heavy use, 
especially about their asymptotic character, wronskians and zeros. We
prove none of them. For a rigorous exposition on PCFs and mPCFs one may 
consult \S 5 of \cite{olver1975} or \S 12 of
\cite{olver-et-al} and the references therein.

\subsection{PCFs}
\label{pcf}

Consider \textit{Weber's equation}
\be\label{pcfs}
\frac{d^2 w}{dx^2}=(\tfrac{1}{4}x^2+b)w.
\ee
The behavior of the solutions depends on the sign of $b$. When $b$ is 
negative then there exist two turning points $\pm 2\sqrt{-b}$. The
solutions are of oscillatory type in the interval between these points
but not in the exterior intervals. When $b>0$ there are no
real turning points and there are no oscillations at all. Since only
the case $b\leq0$ will be of interest to us, from now on we seldom mention 
properties having to do with the other case.

Standard solutions of (\ref{pcfs}) are $U(\pm x,b)$ and $\ol U(\pm x,b)$ 
defined by
\begin{multline*}
U(\pm x,b)=
\frac{\pi^{\frac{1}{2}}2^{-\frac{1}{4}(2b+1)}}{\G(\frac{3}{4}+\frac{1}{2}b)}
e^{-\frac{1}{4}x^2}
\Hypergeometric{1}{1}{\tfrac{1}{4}+\tfrac{1}{2}b}{\tfrac{1}{2}}{\tfrac{1}{2}x^2}\\
\mp
\frac{\pi^{\frac{1}{2}}2^{-\frac{1}{4}(2b-1)}}{\G(\frac{1}{4}+\frac{1}{2}b)}
x
e^{-\frac{1}{4}x^2}
\Hypergeometric{1}{1}{\tfrac{3}{4}+\tfrac{1}{2}b}{\tfrac{3}{2}}{\tfrac{1}{2}x^2}
\end{multline*}
\begin{multline*}
\ol U(\pm x,b)=
\pi^{-\frac{1}{2}}2^{-\frac{1}{4}(2b+1)}\G(\tfrac{1}{4}-\tfrac{1}{2}b)
\sin(\tfrac{3}{4}\pi-\tfrac{1}{2}b\pi)
e^{-\frac{1}{4}x^2}
\Hypergeometric{1}{1}{\tfrac{1}{4}+\tfrac{1}{2}b}{\tfrac{1}{2}}{\tfrac{1}{2}x^2}\\
\mp
\pi^{-\frac{1}{2}}2^{-\frac{1}{4}(2b-1)}\G(\tfrac{3}{4}-\tfrac{1}{2}b)
\sin(\tfrac{5}{4}\pi-\tfrac{1}{2}b\pi)
x
e^{-\frac{1}{4}x^2}
\Hypergeometric{1}{1}{\tfrac{3}{4}+\tfrac{1}{2}b}{\tfrac{3}{2}}{\tfrac{1}{2}x^2}
\end{multline*}
where $_1 F_1$ denotes the confluent hypergeometric function 
(again cf. \cite{abramo+stegun}). The pair
$U(x,b), \ol U(x,b)$ is a numerically satisfactory set of solutions (in the 
sense of \cite{miller1950}) when $x\geq0$ and $b\leq0$; both are continuous 
in $x$ and $b$ in this region.

For $b\in\R$, their values at $x=0$ obey
\begin{align*}
U(0,b) &=
\pi^{-\frac{1}{2}}2^{-\frac{1}{4}(2b+1)}
\G(\tfrac{1}{4}-\tfrac{1}{2}b)\sin(\tfrac{\pi}{4}-\tfrac{1}{2}b\pi)\\
U'(0,b) &= -
\pi^{-\frac{1}{2}}2^{-\frac{1}{4}(2b-1)}
\G(\tfrac{3}{4}-\tfrac{1}{2}b)\sin(\tfrac{3\pi}{4}-\tfrac{1}{2}b\pi)\\
\ol U(0,b) &=
\pi^{-\frac{1}{2}}2^{-\frac{1}{4}(2b+1)}
\G(\tfrac{1}{4}-\tfrac{1}{2}b)\sin(\tfrac{3\pi}{4}-\tfrac{1}{2}b\pi)\\
\ol U'(0,b) &= -
\pi^{-\frac{1}{2}}2^{-\frac{1}{4}(2b-1)}
\G(\tfrac{3}{4}-\tfrac{1}{2}b)\sin(\tfrac{5\pi}{4}-\tfrac{1}{2}b\pi).
\end{align*} 

Those values of $b$ that make the Gamma functions in the definitions of
$U$ and $\ol U$ infinite (the Gamma function has simple poles at the 
non-positive integers), are called \textit{exceptional values}. 
For a fixed $b\in\R$ other than an exceptional value,
the behaviors of $U$ and $\ol U$ as $x\to+\infty$ satisfy
\be
\label{U-asympt}
\begin{split}
U(x,b) \sim
x^{-b-\frac{1}{2}}e^{-\frac{1}{4}x^2}\\
U'(x,b) \sim
-\frac{1}{2}x^{-b+\frac{1}{2}}e^{-\frac{1}{4}x^2}\\
\ol U(x,b) \sim
\sqrt{\frac{2}{\pi}}
\G(\tfrac{1}{2}-b)
x^{b-\frac{1}{2}}e^{\frac{1}{4}x^2}\\
\ol U'(x,b) \sim
(2\pi)^{-\frac{1}{2}}
\G(\tfrac{1}{2}-b)
x^{b+\frac{1}{2}}e^{\frac{1}{4}x^2}.
\end{split}
\ee
These estimates are uniform in $b$ when $b$ takes values over a fixed
compact interval not containing exceptional values.

For the wronskian of $U(\cdot,b)$, $\ol U(\cdot,b)$ we have
\be\label{wronskian-pcf}
\W[U(\cdot,b),\ol U(\cdot,b)](x)=2^{\frac{1}{2}}\pi^{-\frac{1}{2}}
\G\Big(\frac{1}{2}-b\Big),
\quad x\in\R.
\ee

When $b=0$ the standard solutions of equation (\ref{pcf}) are related to
the \textit{modified Bessel functions} $K_{\frac{1}{4}}$ and 
$I_{\frac{1}{4}}$ in the following way. For $x\geq0$ we have
\begin{align*}
U(x,0) & =
(2\pi)^{-\frac{1}{2}}x^{\frac{1}{2}}K_{\frac{1}{4}}(\tfrac{1}{4}x^2)\\
\ol U(x,0) & =
(\pi x)^{\frac{1}{2}}I_{\frac{1}{4}}(\tfrac{1}{4}x^2)+
(2\pi x)^{-\frac{1}{2}}x^{\frac{1}{2}}K_{\frac{1}{4}}(\tfrac{1}{4}x^2).
\end{align*}

In order to express the character of these standard solutions
for large  negative $b$, we need some preparations 
first. Take $\nu\gg1$ to be a large positive number and set
$b=-\frac{1}{2}\nu^2$ and $x=\nu y\sqrt{2}$ where $y\geq0$. If we
consider the fuction $\eta$ to be
\be\label{eta-definition}
\eta(y)=
\begin{cases}
-[\frac{3}{2}\int_{y}^{1}(1-s^2)^{\frac{1}{2}}ds]^{\frac{2}{3}}
,\quad 0\leq y\leq1\\
[\frac{3}{2}\int_{1}^{y}(s^2-1)^{\frac{1}{2}}ds]^{\frac{2}{3}}
,\quad y\geq1
\end{cases}
\ee
then as $\nu\to+\infty$ we have
\begin{align}\label{u-asymptotics}
U(\nu y\sqrt{2},-\tfrac{1}{2}\nu^2) & =
\frac{2^{\frac{1}{2}}\pi^{\frac{1}{4}}
\G(\tfrac{1}{2}+\frac{1}{2}\nu^2)^{\frac{1}{2}}}
{\nu^{\frac{1}{6}}} 
\Big(\frac{\eta}{y^2-1}\Big)^{\frac{1}{4}}
\Big[
\Ai(\nu^{\frac{4}{3}}\eta)+
\frac{M(\nu^{\frac{4}{3}}\eta)}{E(\nu^{\frac{4}{3}}\eta)}
\asympt(\nu^{-2})
\Big]
\\
\label{ubar-asymptotics}
\ol U(\nu y\sqrt{2},-\tfrac{1}{2}\nu^2) & =
\frac{2^{\frac{1}{2}}\pi^{\frac{1}{4}}
\G(\tfrac{1}{2}+\frac{1}{2}\nu^2)^{\frac{1}{2}}\eta^{\frac{1}{4}}}
{\nu^{\frac{1}{6}}(y^2-1)^{\frac{1}{4}}} 
\Big[
\Bi(\nu^{\frac{4}{3}}\eta)+
M(\nu^{\frac{4}{3}}\eta)E(\nu^{\frac{4}{3}}\eta)
\asympt(\nu^{-2})
\Big]
\end{align}
where $\Ai$, $\Bi$, $E$ and $M$ are the standard Airy functions'
terminology (cf. section \ref{airy_functions} in the appendix).

For $b\leq0$, the number of zeros of $U(\cdot,b)$ in the interval 
$[0,+\infty)$ is $\floor{\tfrac{1}{4}-\frac{1}{2}b}$ while
$\ol U(\cdot,b)$ has $\floor{\tfrac{3}{4}-\frac{1}{2}b}$ zeros
in $[0,+\infty)$. Actually, the zeros of $U(\cdot,b)$ and 
$\ol U(\cdot,b)$ do not cross each other. They interlace, with the 
largest one belonging to $\ol U(\cdot,b)$. For sufficiently large 
$|b|$, all the real zeros of these two functions lie to the left 
of $2\sqrt{-b}$, the positive turning point of Weber's equation
\footnote{
For $U(\cdot,b)$, this result holds for all $b\leq0$.
}
. 

%\begin{figure}[H]
%\centering
%\includegraphics[scale=1.2]{pcf-new.eps}
%\caption{An example of Parabolic Cylinder Functions $U(\cdot;b)$ (continuous) 
%and $\ol U(\cdot;b)$ (dashed) for some $b<0$.}
%\end{figure}

To express the errors for the approximations of our problem, we need
to define
some auxiliary functions having to do with the nature of $U(\cdot,b)$ 
and $\ol U(\cdot,b)$ for negative $b$.
In this case the character of each is partly oscillatory and partly
exponential, so we introduce one weight function $\esf$, two modulus functions
$\msf$ and $\nsf$, and finally two phase functions $\8$ and $\om$.

We denote by $\ro(b)$ the largest real root of the equation
\be\nn
U(x,b)=\ol U(x,b).
\ee
We know (cf. \S 13 of \cite{olver-et-al} and the references therein) that 
$\ro(0)=0$ and $\ro(b)>0$ for $b<0$.
Also, $\ro$ is continuous when $b\in(-\infty,0]$. An asymptotic estimate 
for large negative $b$ is
\be\label{largest-root-pcf-asympt}
\ro(b)=2(-b)^{\frac{1}{2}}+c_*(-b)^{-\frac{1}{6}}+
\asympt\big(b^{-\frac{5}{6}}\big)
\quad\text{as}\quad
b\to-\infty
\ee
where $c_*$ ($\approx-0.36605$) is the smallest in absolute value 
root of the equation $\Ai(x)=\Bi(x)$. 

For $b\leq0$ we define 
\begin{equation}\nn
\esf(x,b)=
\begin{cases}
1,\quad 0\leq x\leq\ro(b)\\
\Big[\frac{\ol U(x,b)}{U(x,b)}\Big]^{1/2},\quad x>\ro(b).
\end{cases}
\end{equation}
It is seen that $\esf$ is continuous in the region 
$[0,+\infty)\times(-\infty,0]$ of the $(x,b)$-plane and 
for $b\leq0$ fixed, $\esf(\cdot,b)$ is non-decreasing in 
the interval $[0,+\infty)$. Again for $b\leq0$ and $x\geq0$ we set
\be\nn
U(x,b)=\frac{1}{\esf(x,b)}\msf(x,b)\sin\8(x,b),\quad 
\ol U(x,b)=\esf(x,b)\msf(x,b)\cos\8(x,b)
\ee
and
\be\nn
U'(x,b)=\frac{1}{\esf(x,b)}\nsf(x,b)\sin\om(x,b),\quad 
\ol U'(x,b)=\esf(x,b)\nsf(x,b)\cos\om(x,b).
\ee
Thus
\be\label{m-definition}
\msf(x,b)=
\begin{cases}
\big[U(x,b)^{2}+\ol{U}(x,b)^{2}\big]^{1/2},\quad 0\leq x\leq\ro(b)\\
\big[2U(x,b)\ol U(x,b)\big]^{1/2},\quad x>\ro(b)
\end{cases}
\ee
and
\be\nn
\theta(x,b)=
\begin{cases}
\arctan\Big[\frac{U(x,b)}{\ol U(x,b)}\Big],\quad 0\leq x\leq\ro(b)\\
\frac{\pi}{4},\quad x>\ro(b)
\end{cases}
\ee
where the branch of the inverse tangent is continuous and equal to
$\frac{\pi}{4}$ at $x=\ro(b)$. 

Similarly
\be\nn
\nsf(x,b)=
\begin{cases}
\Big[U'(x,b)^{2}+\ol{U}'(x,b)^{2}\Big]^{1/2},\quad 0\leq x\leq\ro(b)\\
 \\
\bigg[\frac{U'(x,b)^{2}\ol U(x,b)^2+
           \ol U'(x,b)^{2}U(x,b)^{2}}{U(x,b)\ol U(x,b)}\bigg]^{1/2},
           \quad x>\ro(b)
\end{cases}
\ee
and
\be\nn
\omega(x,b)=
\begin{cases}
\arctan\Big[\frac{U'(x,b)}{\ol{U}'(x,b)}\Big],\quad 0\leq x\leq\ro(b)\\
 \\
\arctan\Big[\frac{U'(x,b)\ol{U}(x,b)}{\ol{U}'(x,b)U(x,b)}\Big],
\quad \quad x>\ro(b)
\end{cases}
\ee
where the branches of the inverse tangents are chosen to be continuous and 
fixed by $\omega(x,b)\rightarrow-\frac{\pi}{4}$ as $x\rightarrow+\infty$.

For large $x$ we have
\be\nn
\esf(x,b)\sim\Big(\frac{2}{\pi}\Big)^{\frac{1}{4}}
\G(\tfrac{1}{2}-b)^{\frac{1}{2}}x^{b}e^{\frac{1}{4}x^2}
\ee
and
\be\label{M,N-asymptotics}
\msf(x,b)\sim\Big(\frac{8}{\pi}\Big)^{\frac{1}{4}}
\frac{\G(\tfrac{1}{2}-b)^{\frac{1}{2}}}{x^{\frac{1}{2}}},
\quad
\nsf(x,b)\sim
\frac{\G(\tfrac{1}{2}-b)^{\frac{1}{2}}}{(2\pi)^{\frac{1}{4}}}
x^{\frac{1}{2}}.
\ee
Both of these hold for fixed $b$ and are also uniform for $b$ 
ranging over any compact interval in $(-\infty,0]$.

\subsection{mPCFs}
\label{modified-pcfs}

Consider the equation
\be\label{mpcfs}
\frac{d^2 w}{dx^2}=(b-\tfrac{1}{4}x^2)w.
\ee
When $b>0$ there exist two turning points $\pm 2\sqrt{b}$. The
solutions are monotonic in the interval between these points
and oscillate in the two exterior intervals. When $b\leq0$ there are no
real turning points and the entire real axis is an interval of
oscillation. Only the case $b\geq0$ will be of interest to us.

Standard solutions of (\ref{mpcfs}) are $W(\pm x,b)$ defined by
\begin{multline*}
W(\pm x,b)=
2^{-\frac{3}{4}}
\Bigg|
\frac{\G(\frac{1}{4}+\frac{1}{2}ib)}{\G(\frac{3}{4}+\frac{1}{2}ib)}
\Bigg|^{\frac{1}{2}}
e^{\frac{1}{4}ix^2}
\Hypergeometric{1}{1}
{\tfrac{1}{4}+\tfrac{1}{2}ib}
{\tfrac{1}{2}}
{-\tfrac{1}{2}ix^2}\\
\mp
2^{-\frac{1}{4}}
\Bigg|
\frac{\G(\frac{3}{4}+\frac{1}{2}ib)}{\G(\frac{1}{4}+\frac{1}{2}ib)}
\Bigg|^{\frac{1}{2}}
x
e^{\frac{1}{4}ix^2}
\Hypergeometric{1}{1}
{\tfrac{3}{4}+\tfrac{1}{2}ib}
{\tfrac{3}{2}}
{-\tfrac{1}{2}ix^2}
\end{multline*}
where as in \S \ref{pcf}, $_1 F_1$ denotes the confluent hypergeometric 
function (cf. \cite{abramo+stegun}). A numerically satisfactory set of 
solutions is obtained by taking appropriate multiples of   
$W(\pm x,b)$. Both of them are real and continuous for all real values
of $x$ and $b$. 

Before presenting their basic properties that are useful to us, 
we fix some notation first. We set
\be\label{kappa}
k(b)=(1+e^{2\pi b})^{\frac{1}{2}}-e^{\pi b}
\ee
and
\be\label{phi-mpcf}
\phi(b)=\frac{\pi}{4}+\frac{1}{2}{\rm ph}
\Big\{\G\Big(\frac{1}{2}+ib\Big)\Big\}
\ee
where it is being understood that the phase of 
$\G\Big(\frac{1}{2}+ib\Big)$ in (\ref{phi-mpcf}) is continuous 
and vanishes for $b=0$. Also we know that as $b$ increases from 
$-\infty$ to $+\infty$, $k(b)$ decreases monotonically from $1$ 
to $0$.

For $b\in\R$ and $x=0$ we have
\begin{align*}
W(0,b) &=
2^{-\frac{3}{4}}
\Bigg|
\frac{\G(\frac{1}{4}+\frac{1}{2}ib)}{\G(\frac{3}{4}+\frac{1}{2}ib)}
\Bigg|^{\frac{1}{2}}\\
W'(0,b) &= -
2^{-\frac{1}{4}}
\Bigg|
\frac{\G(\frac{3}{4}+\frac{1}{2}ib)}{\G(\frac{1}{4}+\frac{1}{2}ib)}
\Bigg|^{\frac{1}{2}}.
\end{align*} 

For a fixed $b\in\R$ the behavior of $W(\pm\cdot,b)$ and 
$W'(\pm\cdot,b)$ as $x\to+\infty$ satisfy
\begin{align}
\label{W-asympt}
W(x,b) &=
\sqrt{\frac{2k(b)}{x}}
\cos\bigg[\frac{1}{4}x^2-b\ln x+\phi(b)\bigg]+
\asympt(x^{-\frac{5}{2}})\\
\nn
W'(x,b) &=
-\sqrt{\frac{k(b)x}{2}}
\sin\bigg[\frac{1}{4}x^2-b\ln x+\phi(b)\bigg]+
\asympt(x^{-\frac{3}{2}})\\
\nn
W(-x,b) &=
\sqrt{\frac{2}{k(b)x}}
\sin\bigg[\frac{1}{4}x^2-b\ln x+\phi(b)\bigg]+
\asympt(x^{-\frac{5}{2}})\\
\nn
W'(-x,b) &=
-\sqrt{\frac{x}{2k(b)}}
\cos\bigg[\frac{1}{4}x^2-b\ln x+\phi(b)\bigg]+
\asympt(x^{-\frac{3}{2}}).
\end{align}
These estimates are uniform in $b$ lying in any fixed
compact interval.

For the wronskian of $W(\cdot,b)$, $W(-\cdot,b)$ we have
\be\label{wronskian-mpcf}
\W[W(\cdot,b),W(-\cdot,b)](x)=1,
\quad x\in\R.
\ee

When $b=0$ the standard solutions of equation (\ref{mpcfs}) are related 
to the \textit{Bessel functions} $J_{\pm\frac{1}{4}}$ and 
$J_{\pm\frac{3}{4}}$ in the following way. Since $k(0)=\sqrt{2}-1$
by (\ref{kappa}) and $\phi(0)=\frac{\pi}{4}$ by (\ref{phi-mpcf}), 
for $x\geq0$ we have
\begin{align*}
W(\pm x,0) & =
2^{-\frac{5}{4}}(\pi x)^{\frac{1}{2}}
[J_{-\frac{1}{4}}(\tfrac{1}{4}x^2)\mp J_{\frac{1}{4}}(\tfrac{1}{4}x^2)]\\
W'(\pm x,0) & =
2^{-\frac{9}{4}}\pi^{\frac{1}{2}}x^{\frac{3}{2}}
[\mp J_{\frac{3}{4}}(\tfrac{1}{4}x^2)-J_{-\frac{3}{4}}(\tfrac{1}{4}x^2)].
\end{align*}

The behavior of these standard solutions
for large  positive $b$ can be seen by setting 
$b=\frac{1}{2}\nu^2$ and $x=\nu y\sqrt{2}$ where $\nu\gg1$ is a 
large positive number and $y\geq0$. Then as $\nu\uparrow+\infty$ we have
\begin{align}
\label{kappa-asymptotics}
k\Big(\frac{1}{2}\nu^2\Big) &=
\frac{1}{2}e^{-\frac{1}{2}\pi\nu^2}+
\asympt(e^{-\frac{3}{2}\pi\nu^2})
\\
\label{phi-asymptotics}
\phi\Big(\frac{1}{2}\nu^2\Big) &=
\frac{1}{4}\nu^2\ln\Big(\frac{1}{2}\nu^2\Big)-\frac{1}{4}\nu^2+\frac{\pi}{4}+
\asympt(\nu^{-2})
\\
\label{w-plus-asymptotics}
k\Big(\frac{1}{2}\nu^2\Big)^{-\frac{1}{2}}
W(\nu y\sqrt{2},\tfrac{1}{2}\nu^2) & =
\frac{2^{\frac{1}{4}}\pi^{\frac{1}{2}}}{\nu^{\frac{1}{6}}} 
\Big(\frac{\eta}{y^2-1}\Big)^{\frac{1}{4}}
\Big[
\Bi(-\nu^{\frac{4}{3}}\eta)+
M(-\nu^{\frac{4}{3}}\eta)E(-\nu^{\frac{4}{3}}\eta)
\asympt(\nu^{-2})
\Big]
\\
\label{w-minus-asymptotics}
k\Big(\frac{1}{2}\nu^2\Big)^{\frac{1}{2}}
W(-\nu y\sqrt{2},\tfrac{1}{2}\nu^2) & =
\frac{2^{\frac{1}{4}}\pi^{\frac{1}{2}}}{\nu^{\frac{1}{6}}} 
\Big(\frac{\eta}{y^2-1}\Big)^{\frac{1}{4}}
\Big[
\Ai(-\nu^{\frac{4}{3}}\eta)+
\frac{M(-\nu^{\frac{4}{3}}\eta)}{E(-\nu^{\frac{4}{3}}\eta)}
\asympt(\nu^{-2})
\Big]
\end{align}
where $\Ai$, $\Bi$, $E$ and $M$ are the standard Airy functions'
terminology (cf. section \ref{airy_functions} in the appendix)
and $\eta$ is as in (\ref{eta-definition}).
In the last two relations, the $\asympt$-terms are uniformly
valid in any $y$-interval that includes $[0,+\infty)$.

%\begin{figure}[H]
%\centering
%\includegraphics[scale=1.2]{pcf-new.eps}
%\caption{An example of Parabolic Cylinder Functions $U(\cdot;b)$ (continuous) 
%and $\ol U(\cdot;b)$ (dashed) for some $b<0$.}
%\end{figure}

To express the errors for the approximations in Theorem 
\ref{main-thm-well}, we need to define some auxiliary functions 
having to do with the nature of $k(b)^{-\frac{1}{2}}W(\cdot,b)$ 
and $k(b)^{\frac{1}{2}}W(-\cdot,b)$ for positive $b$.
As in the case of the PCFs in \S \ref{pcf}, we introduce one weight 
function $\overline{\esf}$, two modulus functions
$\overline{\msf}$ and $\overline{\nsf}$, and finally two phase 
functions $\overline{\8}$ and $\overline{\om}$.

Take $b\geq0$ and denote by $\overline{\ro}(b)$ the smallest real root 
in $x\in[0,+\infty)$ of the equation
\be\nn
k(b)^{-\frac{1}{2}}W(x,b)=
k(b)^{\frac{1}{2}}W(-x,b).
\ee
We know (cf. \S 13 of \cite{olver-et-al} and the references therein) that 
\be\nn
k(b)^{-\frac{1}{2}}W(x,b)>
k(b)^{\frac{1}{2}}W(-x,b)>0,
\quad
0\leq x<\overline{\ro}(b).
\ee
Also, $\overline{\ro}$ is continuous when $b\in[0,+\infty)$. 
An asymptotic estimate for large positive $b$ is
\be\label{smallest-root-mpcf-asympt}
\overline{\ro}(b)=2b^{\frac{1}{2}}-c_*b^{-\frac{1}{6}}+
\asympt\big(b^{-\frac{5}{6}}\big)
\quad\text{as}\quad
b\to+\infty
\ee
where as in \S\ref{pcf}, $c_*$ ($\approx-0.36605$) is the smallest 
in absolute value root of the equation $\Ai(x)=\Bi(x)$. 

So for $b\geq0$ we define 
\begin{equation}\nn
\overline{\esf}(x,b)=
\begin{cases}
\overline{\esf}(-x,b),
\quad
x<0\\
\Big[\frac{k(b)W(-x,b)}{W(x,b)}\Big]^{1/2},
\quad 
0\leq x\leq\overline{\ro}(b)\\
1,
\quad 
x>\overline{\ro}(b).
\end{cases}
\end{equation}
It is seen that $\overline{\esf}$ is continuous in the region 
$(-\infty,+\infty)\times[0,+\infty)$ of the $(x,b)$-plane and 
for $b\geq0$ fixed, $\overline{\esf}(\cdot,b)$ is non-decreasing in 
the interval $[0,+\infty)$. Again for $b\leq0$ and $x\geq0$ we have
\be\nn
k(b)^{\frac{1}{2}}\leq\overline{\esf}(x,b)\leq1
\ee

For $b\geq0$ and $x\geq0$, modulus and phase functions are defined by
\be\nn
k(b)^{-\frac{1}{2}}W(x,b)=
\frac{\overline{\msf}(x,b)}{\overline{\esf}(x,b)}\sin\overline{\8}(x,b),
\quad 
k(b)^{\frac{1}{2}}W(-x,b)=
\overline{\msf}(x,b)\overline{\esf}(x,b)\cos\overline{\8}(x,b)
\ee
and
\be\nn
k(b)^{-\frac{1}{2}}W'(x,b)=
\frac{\overline{\nsf}(x,b)}{\overline{\esf}(x,b)}\sin\overline{\om}(x,b),
\quad 
k(b)^{\frac{1}{2}}W'(-x,b)=-
\overline{\nsf}(x,b)\overline{\esf}(x,b)\cos\overline{\om}(x,b).
\ee

Thus
\be\label{m-bar-definition}
\overline{\msf}(x,b)=
\begin{cases}
\big[2W(x,b)W(-x,b)\big]^{1/2},
\quad 
0\leq x\leq\overline{\ro}(b)\\
\big[k(b)^{-1}W(x,b)^{2}+k(b)W(-x,b)^{2}\big]^{1/2},
\quad 
x>\overline{\ro}(b)
\end{cases}
\ee
and
\be\nn
\overline{\theta}(x,b)=
\begin{cases}
\frac{\pi}{4},
\quad 
0\leq x\leq\overline{\ro}(b)\\
\arctan\Big[k(b)^{-1}\frac{W(x,b)}{W(-x,b)}\Big],
\quad 
x>\overline{\ro}(b)
\end{cases}
\ee
where the branch of the inverse tangent is continuous and equal to
$\frac{\pi}{4}$ at $x=\overline{\ro}(b)$. 
Similarly
\be\nn
\overline{\nsf}(x,b)=
\begin{cases}
\bigg[\frac{W'(x,b)^{2}W(-x,b)^2+
           W'(-x,b)^{2}W(x,b)^{2}}{W(x,b)W(-x,b)}\bigg]^{1/2},
\quad 
0\leq x\leq\overline{\ro}(b)\\
 \\
\Big[k(b)^{-1}W'(x,b)^{2}+k(b)W'(-x,b)^{2}\Big]^{1/2},
\quad 
x>\overline{\ro}(b)
\end{cases}
\ee
and
\be\nn
\overline{\omega}(x,b)=
\begin{cases}
-\arctan\Big[\frac{W'(x,b)W(-x,b)}{W'(-x,b)W(x,b)}\Big],
\quad 
0\leq x\leq\overline{\ro}(b)\\
 \\
-\arctan\Big[k(b)^{-1}\frac{W'(x,b)}{W'(-x,b)}\Big],
\quad 
x>\overline{\ro}(b)
\end{cases}
\ee
where the branches of the inverse tangents are chosen to be continuous and 
fixed by $\overline{\omega}(x,b)\rightarrow-\frac{\pi}{4}$ as 
$x\rightarrow+\infty$.

For large $|x|$ we have
\be\label{M,N-bar-asymptotics}
\overline{\msf}(x,b)\sim
\Big|\frac{2}{x}\Big|^{\frac{1}{2}},
\quad
\overline{\nsf}(x,b)\sim
\Big|\frac{x}{2}\Big|^{\frac{1}{2}}.
\ee
Both of these hold for fixed $b$ and are also uniform for $b$ ranging over
any compact interval in $[0,+\infty)$.

\section{A Theorem on Integral Equations}
\label{exist-proof}

The proofs of theorems about WKB approximation when there is an absence 
of turning points (like Theorems 2.1 and 2.2 in chapter 6 of
\cite{olver1997}), may be adapted to other types of 
approximate solutions of linear differential equations 
where turning points may be present. 
For second-order equations the basic steps consist of
\begin{itemize}
\item[(i)]
construction of a \textit{Volterra integral equation} for the error 
term -say $h$- of the solution, by the method of 
\textit{variation of parameters}
\item[(ii)]
construction of  the 
\textit{Liouville-Neumann expansion} (a uniformly convergent series) 
for the solution $h$
of the integral equation in (i) by \textit{Picard's method of successive 
approximations}
\item[(iii)]
confirmation that $h$ is twice differentiable by construction of 
similar series for $h'$ and $h''$
\item[(iv)]
production of bounds for $h$ and $h'$ by majoring the 
Liouville-Neumann expansion.
\end{itemize}

It would be tedious to carry out all these steps in every case. But
we have the following general theorem which automatically provides (ii), 
(iii) and (iv) in problems relevant to us.

\begin{theorem}\label{sing-int-eq}
\footnote{
This is Theorem 10.2 found in chapter 6 of \cite{olver1997}. It is a 
variant of Theorem 10.1 from the same reference.
}
Consider the equation
\begin{equation}\label{int_eq1}
h(\z)=\int_{\beta}^{\z}\mathsf{K}(\z,t)\phi(t)\{J(t)+h(t)\}dt
\end{equation}
for the function $h$ accompanied by the following assumptions
\begin{itemize}
\item
the \say{path} of integration consists of a segment $[\beta,\zeta]$ 
of the real axis, finite or infinite where $\beta\leq t\leq\z\leq\g$
\item
the real functions $J$ and $\phi$ are continuous in $(\beta,\g)$ 
except for a finite number of discontinuities and infinities
\item
the real kernel $\mathsf{K}$ and its first two partial derivatives 
with respect to $\z$ are continuous functions of both variables when 
$\z,t\in(\beta,\g)$
\item
$\mathsf{K}(\z,\z)=0,\quad\z\in(\beta,\g)$
\item
when $\z\in(\beta,\g)$ and $t\in(\beta,\z]$ we have
\begin{equation*}
|\mathsf{K}(\z,t)|\leq P_{0}(\z)Q(t),\qquad 
\Big|\frac{\partial\mathsf{K}(\z,t)}{\partial\z}\Big|\leq P_{1}(\z)Q(t),
\qquad\Big|\frac{\partial^{2}\mathsf{K}(\z,t)}{\partial\z^{2}}\Big|
\leq P_{2}(\z)Q(t)
\end{equation*}
where the $P_{j},j=0,1,2$ and $Q$ are continuous real functions, the 
$P_{j},j=0,1,2$ being positive.
\item
when $\z\in(\beta,\g)$, the integral
\begin{equation*}
\Phi(\z)=\int_{\beta}^{\z}|\phi(t)|dt
\end{equation*}
converges and the following suprema
\begin{equation*}
\kappa=\sup_{\z\in(\beta,\g)}\{Q(\z)|J(\z)|\},\qquad
\kappa_{0}=\sup_{\z\in(\beta,\g)}\{P_{0}(\z)Q(\z)\}
\end{equation*}
are finite.
\end{itemize}
Under these assumptions, equation (\ref{int_eq1}) has a unique solution 
$h$ which is continuously differentiable in $(\beta,\g)$
and satisfies
\begin{equation*}
\frac{h(\z)}{P_{0}(\z)}\to0\qquad\frac{h'(\z)}{P_{1}(\z)}
\to0\qquad\text{as}\quad\z\downarrow\beta.
\end{equation*}
Furthermore,
\begin{equation*}
\frac{|h(\z)|}{P_{0}(\z)},\frac{|h'(\z)|}{P_{1}(\z)}\leq
\frac{\kappa}{\kappa_{0}}[\exp\{\kappa_{0}\Phi(\z)\}-1]
\end{equation*}
and $h''$ is continuous except at the discontinuities -if any- of 
$\phi,J$.
\end{theorem}
\begin{proof}
The proof is a slight variation of that for Theorem 10.1 of chapter 6 
in  \cite{olver1997}.
\end{proof}

We are going to use this theorem to prove the existence and behavior of
approximate solutions of the equation
\be\label{eq-append}
\frac{d^2\mathcal{Y}}{d\z^2}=
\big[\hb^{-2}(\z^2-\alpha^2)+\ps(\z,\hb,\al)\big]\mathcal{Y}.
\ee
We have the following
\btheo\label{thm-on-exist-int-eq}
For each value of $\hb$, assume that the function $\ps(\z,\hb,\al)$ is 
continuous in the region $[0,Z)\times[0,\delta]$ of the $(\z,\al)$-plane
\footnote{
Here $Z$ is always positive and may depend continuously on $\al$, or
be infinite. Also, $\delta$ is a positive finite constant.},  
take $\Omega$ as in (\ref{omega-barrier})
and consider that 
\be\nn
\var_{0,Z}[H](\al,\hb)=
\int_{0}^{Z}\frac{|\ps(t,\al)|}{\Om(t\sqrt{2\hb^{-1}})}dt
\ee
converges uniformly with 
respect to $\al$. Then in this region, equation 
(\ref{eq-append}) has 
solutions $\mathcal{Y}_1$ and $\mathcal{Y}_2$ which are continuous, have 
continuous first and 
second partial $\z$-derivatives and are given by
\bea\nn
\mathcal{Y}_1(\z,\al,\hb)=
U(\z\sqrt{2\hb^{-1}},-\tfrac{1}{2}\hb^{-1}\al^2)+
\epsilon_1 (\z,\al,\hb)\\
\nn
\mathcal{Y}_2(\z,\al,\hb)=
\ol U(\z\sqrt{2\hb^{-1}},-\tfrac{1}{2}\hb^{-1}\al^2)+
\epsilon_2 (\z,\al,\hb)
\eea
where
\begin{multline}\label{bound1}
\frac{|\epsilon_1 (\z,\al,\hb)|}
{\msf(\z\sqrt{2\hb^{-1}},-\tfrac{1}{2}\hb^{-1}\al^2)},
\frac{\Big|\frac{\partial \epsilon_1 }
{\partial\z}(\z,\al,\hb)\Big|}{\sqrt{2\hb^{-1}}
\nsf(\z\sqrt{2\hb^{-1}},-\tfrac{1}{2}\hb^{-1}\al^2)}\\
\leq
\frac{1}{\esf(\z\sqrt{2\hb^{-1}},-\tfrac{1}{2}\hb^{-1}\al^2)}
\Big(\exp\big\{\tfrac{1}{2}(\pi\hb)^{\frac{1}{2}}l(-\tfrac{1}{2}\hb^{-1}\al^2)
\mathcal{V}_{\z,Z}[H](\al,\hb)\big\}-1\Big)
\end{multline}
and
\begin{multline}\label{bound2}
\frac{|\epsilon_2 (\z,\al,\hb)|}
{\msf(\z\sqrt{2\hb^{-1}},-\tfrac{1}{2}\hb^{-1}\al^2)},
\frac{\Big|\frac{\partial \epsilon_2 }
{\partial\z}(\z,\al,\hb)\Big|}{\sqrt{2\hb^{-1}}
\nsf(\z\sqrt{2\hb^{-1}},-\tfrac{1}{2}\hb^{-1}\al^2)}\\
\leq
\esf(\z\sqrt{2\hb^{-1}},-\tfrac{1}{2}\hb^{-1}\al^2)
\Big(\exp\big\{\tfrac{1}{2}(\pi\hb)^{\frac{1}{2}}l(-\tfrac{1}{2}\hb^{-1}\al^2)
\mathcal{V}_{0,\z}[H](\al,\hb)\big\}-1\Big).
\end{multline}
\etheo
\begin{proof}
We will prove the theorem only for the first solution since the proof for
the second follows mutatis mutandis. Observe that the approximating function 
$U(\z\sqrt{2\hb^{-1}},-\tfrac{1}{2}\hb^{-1}\al^2)$ satisfies
$\frac{d^2U}{d\z^2}=\hb^{-2}(\z^2-\alpha^2)U$.  If we subtract this from
(\ref{eq-append}) we obtain the following differential equation for the 
error term
\be\nn
\frac{d^2\epsilon_1}{d\z^2}-\hb^{-2}(\z^2-\alpha^2)\epsilon_1=
\ps(\z,\al,\hb)\big[\epsilon_1+
U(\z\sqrt{2\hb^{-1}},-\tfrac{1}{2}\hb^{-1}\al^2)\big].
\ee
By use of the method of variation of parameters
and also  (\ref{wronskian-pcf}) one arrives at the integral equation
\be\nn
\epsilon_1 (\z,\al,\hb)=
\frac{1}{2}
\frac{(\pi\hb)^{\frac{1}{2}}}{\G\big(\frac{1}{2}+\frac{1}{2}\hb^{-1}\al^2\big)}
\int_{\z}^{Z}
\mathcal{K}(\z,t)
\ps(t,\al,\hb)\big[\epsilon_1(t,\al,\hb)+
U(t\sqrt{2\hb^{-1}},-\tfrac{1}{2}\hb^{-1}\al^2)\big]
dt
\ee
in which
\begin{multline*}
\mathcal{K}(\z,t)=
U(\z\sqrt{2\hb^{-1}},-\tfrac{1}{2}\hb^{-1}\al^2)
\ol U(t\sqrt{2\hb^{-1}},-\tfrac{1}{2}\hb^{-1}\al^2)\\
-U(t\sqrt{2\hb^{-1}},-\tfrac{1}{2}\hb^{-1}\al^2)
\ol U(\z\sqrt{2\hb^{-1}},-\tfrac{1}{2}\hb^{-1}\al^2).
\end{multline*}

Bounds for the \textit{kernel} $\mathcal{K}$ and its first two partial
derivatives (with respect to $\z$) are expressible in terms of the 
auxiliary functions $\esf, \msf$ and $\nsf$. We have
\begin{align*}
|\mathcal{K}(\z,t)| & \leq
\frac{\esf(t\sqrt{2\hb^{-1}},-\tfrac{1}{2}\hb^{-1}\al^2)}
{\esf(\z\sqrt{2\hb^{-1}},-\tfrac{1}{2}\hb^{-1}\al^2)}
\msf(\z\sqrt{2\hb^{-1}},-\tfrac{1}{2}\hb^{-1}\al^2)
\msf(t\sqrt{2\hb^{-1}},-\tfrac{1}{2}\hb^{-1}\al^2)\\
\bigg|\frac{\partial\mathcal{K}}{\partial\z}(\z,t)\bigg| & \leq
\sqrt{2\hb^{-1}}
\frac{\esf(t\sqrt{2\hb^{-1}},-\tfrac{1}{2}\hb^{-1}\al^2)}
{\esf(\z\sqrt{2\hb^{-1}},-\tfrac{1}{2}\hb^{-1}\al^2)}
\nsf(\z\sqrt{2\hb^{-1}},-\tfrac{1}{2}\hb^{-1}\al^2)
\msf(t\sqrt{2\hb^{-1}},-\tfrac{1}{2}\hb^{-1}\al^2)
\end{align*}
and similarly
\be\nn
\frac{\partial^{2}\mathcal{K}}{\partial\z^{2}}(\z,t)=
(2\hb^{-1})^{\frac{3}{2}}\z\mathsf{K}(\z,t).
\ee
All these estimates allow us to solve the equation (\ref{eq-append}) 
by applying Theorem \ref{int_eq1}. Using the notation of that theorem
we have
\begin{align*}
\phi(t) & =\frac{\ps(\z,\al,\hb)}{\Om(\z\sqrt{2\hb^{-1}})}\\
\psi_{1}(t) & =0\\
J(t) & =U(t\sqrt{2\hb^{-1}},-\tfrac{1}{2}\hb^{-1}\al^2)\\
\mathsf{K}(\z,t) & =-\frac{1}{2}
\frac{(\pi\hb)^{\frac{1}{2}}}{\G\big(\frac{1}{2}+\frac{1}{2}\hb^{-1}\al^2\big)}
\Omega(t\sqrt{2\hb^{-1}})\mathcal{K}(\z,t)\\
Q(t) & =
\frac{1}{2}
\frac{(\pi\hb)^{\frac{1}{2}}}{\G\big(\frac{1}{2}+\frac{1}{2}\hb^{-1}\al^2\big)}
\Omega(t\sqrt{2\hb^{-1}})
\esf(t\sqrt{2\hb^{-1}},-\tfrac{1}{2}\hb^{-1}\al^2)
\msf(t\sqrt{2\hb^{-1}},-\tfrac{1}{2}\hb^{-1}\al^2)\\
P_0(\z) & =
\frac
{\msf(\z\sqrt{2\hb^{-1}},-\tfrac{1}{2}\hb^{-1}\al^2)}
{\esf(\z\sqrt{2\hb^{-1}},-\tfrac{1}{2}\hb^{-1}\al^2)}\\
P_1(\z) & =\sqrt{2\hb^{-1}}
\frac
{\nsf(\z\sqrt{2\hb^{-1}},-\tfrac{1}{2}\hb^{-1}\al^2)}
{\esf(\z\sqrt{2\hb^{-1}},-\tfrac{1}{2}\hb^{-1}\al^2)}\\
\Phi(\z) & =\var_{0,\z}[H](\al,\hb)\\
\kappa_0 &\leq 
\frac{1}{2}(\pi\hb)^{\frac{1}{2}}l(-\tfrac{1}{2}\hb^{-1}\al^2)
\end{align*}
where the role of $\beta$ is played here by $Z$ and $\kappa$ is replaced 
for simplicity by the upper bound $\kappa_0$. Then the bounds 
(\ref{bound1}) and (\ref{bound2}) follow from Theorem \ref{int_eq1}. 

Finally, observe that all the integrals which
occur in the analysis above, converge uniformly when $\al\in[0,\delta]$
and $\z$ lies in any compact interval of $[0,Z)$; allowing us to state 
that $\epsilon_1$ and its first two partial $\z$-derivatives are 
continuous in $\al$ and $\z$. Consequently, the same stands 
for $\mathcal{Y}_1$ which signifies the end of the proof.
\end{proof}

\section*{Data Availability}

Data sharing is not applicable to this article as no new data were
created or analyzed in this study.

\section*{Acknowledgements} 

%We are grateful to a referee for insisting  on the clarification of the 
%results and proofs of section \S\ref{near-zero-evs}. 
The first author 
acknowledges the support of the Institute of Applied and Computational 
Mathematics of the Foundation of Research and Technology - Hellas 
(\href{https://www.forth.gr/}{FORTH}), via grant 
MIS 5002358 and the support of the University of Crete via grant
10753. Also, the first author expresses his sincere gratitude to 
the Independent Power Transmission Operator 
(\href{https://www.admie.gr/en}{IPTO}) for a scholarship through the
\href{http://www.sse.uoc.gr/en/}{School of Sciences and Engineering} 
of the University of Crete.

\end{document}